%% file: main.tex
\PassOptionsToPackage{prologue,dvipsnames}{xcolor}
\documentclass[acmsmall]{acmart}

\usepackage{graphicx}
\usepackage{mathpartir}
\usepackage{stmaryrd}
\usepackage{amsmath}
\usepackage{amsfonts}
\usepackage{tabularx}
\usepackage{booktabs}
\usepackage{multirow}
\usepackage{mathtools}
\usepackage{cleveref}
\usepackage{longtable}
\usepackage{subcaption}
\usepackage{afterpage}
\usepackage{cleveref}
\usepackage[lined,linesnumbered,noend,ruled]{algorithm2e}
\usepackage[round-mode=places,round-precision=2,round-integer-to-decimal]{siunitx}
\usepackage{import}
\usepackage{pgf}
\usepackage{cancel}

\setcopyright{rightsretained}
\acmDOI{10.1145/3689724}
\acmYear{2024}
\acmJournal{PACMPL}
\acmVolume{8}
\acmNumber{OOPSLA2}
\acmArticle{284}
\acmMonth{10}
\acmSubmissionID{oopslab24main-p208-p}
\received{2024-04-06}
\received[accepted]{2024-08-18}



\begin{document}
\title{Synthesizing Formal Semantics from Executable Interpreters}

\author{Jiangyi Liu}
\orcid{0000-0001-6525-4659}
\affiliation{%
  \institution{University of Wisconsin -- Madison}
  \city{Madison}
  \country{USA}
}
\email{jiangyi.liu@wisc.edu}

\author{Charlie Murphy}
\orcid{0000-0003-4813-7578}
\affiliation{%
  \institution{University of Wisconsin -- Madison}
  \city{Madison}
  \country{USA}
}
\email{tcmurphy4@wisc.edu}

\author{Anvay Grover}
\orcid{0009-0003-4820-3560}
\affiliation{%
  \institution{University of Wisconsin -- Madison}
  \city{Madison}
  \country{USA}
}
\email{anvayg@cs.wisc.edu}

\author{Keith J.C. Johnson}
\orcid{0000-0002-3766-5204}
\affiliation{%
  \institution{University of Wisconsin -- Madison}
  \city{Madison}
  \country{USA}
}
\email{keith.johnson@wisc.edu}

\author{Thomas Reps}
\orcid{0000-0002-5676-9949}
\affiliation{%
  \institution{University of Wisconsin -- Madison}
  \city{Madison}
  \country{USA}
}
\email{reps@cs.wisc.edu}

\author{Loris D’Antoni}
\orcid{0000-0001-9625-4037}
\affiliation{%
  \institution{University of California, San Diego}
  \city{La Jolla}
  \country{USA}
}
\email{ldantoni@ucsd.edu}

\input{macros}

\begin{abstract}
  \input{abstract}
\end{abstract}

\maketitle              

\input{sections/1-introduction}
\input{sections/2-example}
\input{sections/3-problem_definition}

\input{sections/4-algorithm}
\input{sections/5-implementation}

\input{sections/6-evaluation}
\input{sections/7-related_work}
\input{sections/8-conclusion}

\input{sections/artifact-available}

\begin{acks}
Supported, in part, by a \grantsponsor{00001}{Microsoft}{https://www.microsoft.com/en-us/research/academic-program/faculty-fellowship/} Faculty Fellowship; a gift from \grantsponsor{00002}{Rajiv and Ritu Batra}{}; and \grantsponsor{00003}{NSF}{https://www.nsf.gov/} under grants \grantnum{00003}{CCF-1750965}, \grantnum{00003}{CCF-1918211}, \grantnum{00003}{CCF-2023222}, \grantnum{00003}{CCF-2211968}, and \grantnum{00003}{CCF-2212558}. Any opinions, findings, and conclusions or recommendations expressed in this publication are those of the authors, and do not necessarily reflect the views of the sponsoring entities.
\end{acks}

\bibliographystyle{ACM-Reference-Format}
\bibliography{references}
\citestyle{acmauthoryear}   

\newpage
\appendix
\input{sections/appendix}

\end{document}

%% file: macros.tex
\newcommand{\toolname}{\textsc{Synantic}\xspace}
\newcommand{\keithtoolname}{\textsc{Ks2}\xspace}
\newcommand{\algname}{\textsc{SemSynth}\xspace}
\newcommand{\algSemSynth}{\textsc{SynthSemanticConstraint}\xspace}
\newcommand{\algVerify}{\textsc{Verify}\xspace}

\newcommand{\algSingle}{\textsc{SynthesizeForProduction}\xspace}
\newcommand{\algAll}{\textsc{SynthesizeForAllProductions}\xspace}
\newcommand{\algRefine}{\textsc{CompleteSummary}\xspace}
\newcommand{\algSolveSAT}{\textsc{CheckSat}\xspace}
\newcommand{\algSolveSygus}{\textsc{SolveSygus}\xspace}
\newcommand{\algInitSem}{\textsc{MakeInitialSummary}\xspace}
\newcommand{\algApproxSummary}{\textsc{ApproximateWithSummary}\xspace}
\newcommand{\assert}{\textsc{Assert}\xspace}
\newcommand{\runInterp}{\textsc{RunWithTimeLimit}\xspace}
\newcommand{\algGenSygus}{\textsc{GenerateSygus}\xspace}
\newcommand{\algSygus}{\textsc{FindCandidateSemantics}\xspace}
\newcommand{\algSygusMulti}{\textsc{FindCandidateSemanticsMulti}\xspace}
\newcommand{\algGenGrammar}{\textsc{GenerateGrammar}\xspace}

\let\oldbot\bot
\renewcommand{\bot}{\ensuremath{\mathit{false}}}
\renewcommand{\top}{\ensuremath{\mathit{true}}}

\newcommand{\mathdefault}[1][]{} 

\newcommand{\semgus}{SemGuS\xspace}
\newcommand{\sygus}{SyGuS\xspace}

\crefname{algocf}{Algorithm}{Algorithms}
\Crefname{algocf}{Algorithm}{Algorithms}


\newcommand{\twr}[1]{\textcolor{blue}{T: #1}}
\newcommand{\twrchanged}[1]{\textcolor{cyan}{#1}}

\newcommand{\loris}[1]{\textcolor{purple}{L: #1}}
\newcommand{\lorischanged}[1]{\textcolor{brown}{#1}}

\newcommand{\charlie}[1]{\textcolor{ForestGreen}{C: #1}}
\newcommand{\charliechanged}[1]{\textcolor{ForestGreen}{#1}}

\newcommand{\jiangyi}[1]{\textcolor{teal}{\textbf{J: #1}}}
\newcommand{\jiangyichanged}[1]{\textcolor{orange}{#1}}

\newcommand{\anvay}[1]{\textcolor{orange}{A: #1}}
\newcommand{\anvaychanged}[1]{\textcolor{orange}{#1}}


\newcommand{\rone}{(\emph{i})~}
\newcommand{\rtwo}{(\emph{ii})~}
\newcommand{\rthree}{(\emph{iii})~}
\newcommand{\rfour}{(\emph{iv})~}
\newcommand{\rfive}{(\emph{v})~}

\SetKwProg{Proc}{Procedure}{}{}
\SetKwRepeat{DoWhile}{do}{while}
\SetKw{Continue}{continue}
\SetKwInput{Assert}{assert}
\newcommand{\Break}{\textbf{break}}
\newcommand\mycommfont[1]{\footnotesize\ttfamily\textcolor{blue}{#1}}
\SetCommentSty{mycommfont}
\SetKwBlock{Glob}{Global Variables}{}

\newcommand{\NN}{\mathbb{N}}
\newcommand{\ZZ}{\mathbb{Z}}
\newcommand{\mcL}{\mathcal{L}}
\newcommand{\mcS}{\mathcal{S}}
\newcommand{\defeq}{~\triangleq~}

\newcommand{\semanticBrackets}[1]{\llbracket #1 \rrbracket}
\newcommand{\sem}[1]{\semanticBrackets{#1}}
\newcommand{\semanticBracketsUsing}[2]{\llbracket #1 \rrbracket (#2)}
\newcommand{\semanticBracketsWithNonterm}[3]{\llbracket #2 \rrbracket_{#1} (#3)}

\newcommand{\tok}[1]{\textcolor{blue}{\ensuremath{\mathtt{#1}}}}

\newcommand{\rk}{rk} 
\newcommand{\Sem}{Sem}

\newcommand\Tau{\mathrm{T}}
\newcommand{\inputType}{\theta}
\newcommand{\inputVal}{\ensuremath \mathit{in}}
\newcommand{\inputVar}[1][]{\ensuremath x^\inputVal_{#1}}
\newcommand{\outputType}{\tau}
\newcommand{\outputVal}{\ensuremath \mathit{out}}
\newcommand{\outputVar}[1][]{\ensuremath x^\outputVal_{#1}}
\newcommand{\theory}{\mathcal{T}}
\newcommand{\grammar}{G}
\newcommand{\interpreter}{\ensuremath \mathcal{I}}
\newcommand{\oracle}{\ensuremath \mathcal{E}}
\newcommand{\recpred}{\ensuremath P_\mathit{rec}}
\newcommand{\exampleset}{\ensuremath E}
\newcommand{\tuple}[1]{\ensuremath \left\langle #1 \right\rangle}
\newcommand{\Imp}[1][]{$\textsc{Imp}_{#1}$\xspace}
\newcommand{\ImpArr}{\textsc{ImpArr}}
\newcommand{\seq}{;}

%% file: abstract.tex
Program verification and synthesis frameworks that allow one to customize the language in which one is interested typically require the user to provide a formally defined semantics for the language.
Because writing a formal semantics can be a daunting and error-prone task, this requirement stands in the way of such frameworks being adopted by non-expert users. 
We present an algorithm that can automatically synthesize inductively defined syntax-directed semantics when given \rone a grammar describing the syntax of a language and \rtwo an executable (closed-box) interpreter for computing the semantics of programs in the language of the grammar.
Our algorithm synthesizes the semantics in the form of Constrained-Horn Clauses (CHCs), a natural, extensible, and formal logical framework for specifying inductively defined relations that has recently received widespread adoption in program verification and synthesis.
The key innovation of our synthesis algorithm is a Counterexample-Guided Synthesis (CEGIS) approach that breaks the \textit{hard} problem of synthesizing a set of constrained Horn clauses into small, tractable expression-synthesis problems that can be dispatched to existing \sygus synthesizers.
Our tool \toolname synthesized inductively-defined formal semantics from {14} interpreters for languages used in program-synthesis applications.
When synthesizing formal semantics for one of our benchmarks, \toolname unveiled an  inconsistency in the semantics computed by the interpreter for a language of regular expressions; fixing the inconsistency resulted in a more efficient semantics and, for some cases, in a 1.2x speedup for a synthesizer solving synthesis problems over such a language.

%% file: sections/1-introduction.tex
\section{Introduction} \label{sec:introduction}



Recent work on frameworks for program verification and program synthesis has created tools that are parametric in the language that is supported \cite{DBLP:journals/toplas/LimR13,DBLP:conf/setss/0002R19,kim2021semantics}.
A user of such a framework must define the language of interest by giving both a syntactic specification and a formal semantic specification of the language.
The semantic specification assigns a meaning to each program in the language.
However, for most programming languages, and even for simple ones used in program-synthesis applications, it is usually a demanding task to create a \textit{formal} semantics that defines the behaviors of the programs in the language.
Obstacles include:
\rone the language's semantics might only be documented in natural language, and thus may be ambiguous (or worse, inconsistent), and \rtwo the sheer level of detail that is involved in writing such a semantics.


\paragraph{Synthesizing Formal Semantics from Interpreters}
In this paper, we propose an alternative approach---based on synthesis---that is applicable to any programming language for which a compiler or interpreter exists.
Such infrastructure serves as an operational semantics for the language, albeit one for which anything other than closed-box access would be difficult.
Assuming existence of a working compiler or interpreter is not hard --- usually a language (typically not an “academic’’ language) already has an interpreter already implemented, and the language users, if they want to access techniques like verification and synthesis, need a formal semantics.
Thus, we take closed-box access as a given, and ask the following question:
\begin{quote}
  \textit{Is it possible to use an existing compiler or interpreter for a language $L$ to create a formal semantics for $L$ automatically?}
\end{quote}
In this paper, we assume that the given compiler or interpreter is capable of executing any program or subprogram in language $L$.

This question is natural, but answering it formally requires one to address two key challenges.

First, in what formalism should the formal semantics be expressed?
The right formalism should be expressive enough to capture common semantics, yet structured enough to allow synthesis to be possible.
Furthermore, the formalism should not be tied to any specific programming language---i.e., it should be language-agnostic.

Second, how can the synthesis problem be broken down into simple enough small problems for which one can design a practical approach?
The representation of the semantics of most programming languages is usually very large, and a monolithic synthesis approach that does not take advantage of the compositionality of semantics definitions is bound to fail.

\paragraph{Our Approach}
In this paper, we address both of these challenges and present an algorithm that can automatically synthesize an inductively defined syntax-directed semantics when given \rone a grammar describing the syntax of the language, and \rtwo an executable (closed-box) interpreter for computing the semantics of programs in the language on given inputs.

To address the first of the aforementioned challenges, we choose to synthesize the formal semantics in the form of Constrained Horn Clauses (CHCs), a well-studied fragment of first-order logic that already provides the foundation of \semgus~\cite{semgus-cav21,kim2021semantics}, a domain- and solver-agnostic framework for defining arbitrary synthesis problems.
CHCs can naturally express a big-step operational semantics, structured as an inductive definition over a language's abstract syntax, which makes them appropriate for compositional reasoning.

For example, the operational semantics for an assignment to a variable \tok{x} in an imperative programming language can be written as the following CHC:
    \begin{mathpar}
        \inferrule{
            \semanticBracketsUsing{e}{\mathit{s_1}} = r_1\\
            s_0 = s_1 \land r_0 = s_0[x \mapsto r_1]
        }{
            \semanticBracketsUsing{\tok{x}~\tok{:=}~e}{\mathit{s_0}} = r_0
        }
    \end{mathpar}
    The CHC is defined inductively in terms of the semantics of the child term $e$.

To address the second aforementioned challenge, we take advantage of the inductive structure of CHCs and design a synthesis algorithm that inductively synthesizes the semantics of programs in the grammar, starting from simple base constructs and moving up to more complex inductively-defined constructs.
For each construct in the language, our algorithm uses a counter-example-guided inductive synthesis (CEGIS) loop to synthesize the semantic rule---i.e., the CHC---for that construct.
%
For each construct, we use input-output valuations obtained by calling the closed-box interpreter to approximate the behavior of its child terms. 
Such an approximation allows us to synthesize the semantics 
construct-by-construct, rather than all at once, which converts the problem of synthesizing semantics into many smaller problems that only have to synthesize part of the overall semantics.

To evaluate our approach, we implemented it in a tool called \toolname.
Our evaluation of \toolname involved synthesizing the semantics for languages with a wide variety of features, including assignments, conditionals, while loops, bit-vector operations, and regular expressions.
The evaluation revealed that our approach not only can help synthesize semantics of non-trivial languages but can also help debug existing semantics.

\paragraph{Goals and No-goals} Our tool \toolname mainly targets users who want to use verification and synthesis techniques on an existing language. Once \toolname creates the semantics in \semgus format, a wide range of tools based on \semgus can be instantly applied \cite{semgus-cav24}. For example, such a way enables the user to get a synthesizer for the existing language \emph{for free}, because the creation of \semgus files requires minimal manual labor.
Also the original goal is helping \semgus users, our techniques are general and we envision they could potentially be applied to other semantic-specification frameworks (e.g., to help formalize semantics for use with a theorem prover). Synthesizing semantics of purely academic languages is a no-goal for our tool, because most of them already have a formal semantics available before the interpreter is implemented, thus creating the \semgus specifications would be trivial.

\paragraph{Contributions} Our work makes the following contributions:
\begin{itemize}
  \item
    We introduce a new kind of synthesis problem: the \emph{semantics-synthesis problem} (\Cref{sec:problem-definition}).
  \item
    We devise an algorithm for solving semantics-synthesis problems (\Cref{sec:algorithm}).
    In this algorithm, we harness an example-based program synthesizer (specifically a \sygus solver) to synthesize the constraint in each CHC. 
  \item 
    We implement our algorithm in a tool, called \toolname, which also supports an optimization for multi-output productions, i.e., productions whose semantic constraints include multiple output variables (\Cref{sec:implementation}).
  \item
    We evaluate \toolname on a range of different language benchmarks from the program-synthesis literature.
    For one benchmark, the \toolname-generated semantics revealed an inconsistency in the way the original semantics had been formalized. 
    Fixing the inconsistency in the semantics resulted in a more efficient semantics and 
    a speedup (in some case 1.2x) for a synthesizer solving synthesis problems over such a language
    (\Cref{sec:evaluation})
\end{itemize}
\Cref{sec:example} illustrates how our algorithm synthesizes the semantics of an imperative while-loop language.
\Cref{sec:related-work} discusses related work. \Cref{sec:conclusion} concludes.

References of the form Appendix A.1 refer to appendices that are available in the arXiv version of this paper \cite{paper-arxiv}.

%% file: sections/2-example.tex
\section{Illustrative Example}
\label{sec:example}

As discussed in \Cref{sec:introduction}, our technique synthesizes a semantic specification that is
compatible with the Semantics-Guided Synthesis (\semgus) format \cite{kim2021semantics}.
\semgus is a domain- and solver-agnostic framework for specifying program synthesis and verification problems \cite{semgus-cav21}.
A \semgus problem consists of three components
that the user must provide:
\rone a grammar specifying the syntax of programs;
\rtwo a semantics for every program
in the language of the grammar, provided as a set of Constrained Horn Clauses (CHCs) assigned to the productions of the grammar;
and \rthree a specification of the desired program that makes use of the semantic predicates.
Crucially, \semgus enables the development of general tools for program synthesis and verification, thus reducing the burden of creating such tools for custom languages \cite{semgus-cav24}. However, the stumbling block is that the end user must be able to provide a semantics of the language they are interested in working with, a task that can be burdensome and error-prone to perform by hand. In this section, we illustrate how our technique (implemented in \toolname) automatically synthesizes such a semantics for an imperative language \Imp (cf. \Cref{ex:imp-lang})---a simple but illustrative example of \toolname's abilities.


%
%
%

\begin{example}[Syntactic Definition of \Imp] \label{ex:imp-lang}
Consider the grammar $G_\text{\Imp[n]}$ that defines the syntax of \Imp for programs with $n$ variables $\tok{x_1}$, $\dots$, $\tok{x_n}$:
\begin{align*}
  S &\Coloneqq \tok{x_1 \coloneqq}~E \mid \dots \mid \tok{x_n \coloneqq}~E
     \mid S~\tok{\seq}~S \mid \tok{ite}~B~S~S \mid \tok{while}~B~\tok{do}~S \\
     &\mid \tok{do}~S~\tok{while}~B \mid \tok{repeat}~S~\tok{until}~B \\
  B &\coloneqq \tok{false} \mid \tok{true} \mid \tok{\neg}~B \mid B~\tok{\land}~B \mid B~\tok{\lor}~B \mid E~\tok{<}~E \\
  E &\coloneqq \tok{0} \mid \tok{1} \mid \tok{x_1} \mid \dots \mid \tok{x_n} \mid E~\tok{+}~E \mid E~\tok{-}~E 
\end{align*}

The \Imp language consists of arithmetic and Boolean expressions, statements for assignment to the variables $\tok{x_1}$ through $\tok{x_n}$, sequential composition, if-then-else, and various looping constructs.
\Imp also comes equipped with an executable interpreter 
$\interpreter_\text{\Imp}$ that assigns to each term $t \in \mcL(G)$ its standard (denotational) semantics (e.g., arithmetic and Boolean expressions are evaluated as in linear integer arithmetic, $\tok{x_i \coloneqq}~e$ takes as input a state, and outputs the input state with 
$x_i$'s value updated by the result of evaluating $e$,
etc.).
\end{example}

Suppose that we did not know the semantics of \Imp \emph{a priori}; that is, suppose that we only have access to the interpreter $\interpreter_\text{\Imp}$.
How can we synthesize a formal semantics for each program in $G_\text{\Imp}$ using the interpreter?
A na\"ive approach would randomly generate a large set of terms and inputs, and try to learn a function mapping inputs to outputs for each term. 
However, this approach would only provide a semantics for the enumerated terms, and fails to generalize to the entire language. A less na\"ive approach might attempt to form a monolithic synthesis problem to synthesize a semantic function for each production of the grammar that satisfies a set of generated example terms and input-output pairs. However, it is known that synthesizers scale exceptionally poorly in the size of the desired output \cite{alur2017scaling},
even for \Imp[1], which has only 17 productions, this approach would be practically impossible.

\paragraph{Nullary productions.}
One of the key innovations of our approach is that we synthesize the semantics on a per-production basis, i.e., working one production at a time.
We start by synthesizing a semantics for nullary (leaf) productions. For \Imp[1], this means we synthesize a semantics for the productions $\tok{0}$, $\tok{1}$, $\tok{x_1}$, $\tok{false}$, and $\tok{true}$ before we synthesize the semantics of any other productions. For a nullary production $\tok{p}$, we synthesize a semantics of the form:
\[
  \inferrule*{
    \outputVar[0] = f(\inputVar[0])
  }{
    \mathit{Sem}(\tok{p}, \inputVar[0], \outputVar[0])
  }
\]
which states that, because the term $\tok{p}$ has no sub-terms, the output is only a function of the input $\inputVar[0]$. In our approach, we use a Counter-Example-Guided Synthesis (CEGIS) approach to synthesize a function $f$ that captures the behavior of $\interpreter_\text{\Imp}$ on production $p$. Within the CEGIS loop, we synthesize a candidate function $f$, then verify if it is consistent with $\interpreter_\text{\Imp}$ (e.g., on a larger number of inputs $\inputVar[0]$). If $f$ is consistent, then we have successfully learned the semantics of $p$;
otherwise, the verifier generates a counter-example and a new candidate semantic function $f$.

\paragraph{Inductively synthesizing semantics.}
Next, our approach synthesizes the semantics for other arithmetic and Boolean expressions. In this step, we inductively synthesize the semantics of productions by reusing the semantics of previously learned productions to learn the semantics of new productions. At this point, we may assume that we know the semantics of all nullary productions. For instance, suppose that we wish to next learn the semantics of $\tok{+}$.
At first, our algorithm generates examples favoring terms like $\tok{1 + 1}$, $\tok{x + 1}$, etc.\ that contains sub-terms whose semantics have already been learned.
For $\tok{t_1 + t_2}$, our algorithm generates a semantics that can rely on the semantics of its sub-terms $\tok{t_1}$ and $\tok{t_2}$. Specifically, the semantics of $\tok{t_1 + t_2}$ takes the following form:
\[\inferrule*{
\mathit{sem}(\tok{t_1}, \inputVar[1], \outputVar[1]) \qquad
\mathit{sem}(\tok{t_2}, \inputVar[2], \outputVar[2])\\
\inputVar[1] = f_1(\inputVar[0]) \\
\inputVar[2] = f_2(\inputVar[0], \outputVar[1]) \\
\outputVar[0] = f_0(\inputVar[0], \outputVar[1], \outputVar[2])
}{
\mathit{sem}(\tok{t_1 + t_2}, \inputVar[0], \outputVar[1]) 
}\]
which states that the semantics of $\tok{t_1 + t_2}$ is inductively defined in terms of the semantics of $\tok{t_1}$ and the semantics of $\tok{t_2}$. The semantics enforces a left-to-right evaluation order:\footnote{We show how to overcome this restriction in \Cref{sec:implementation:sem-synth}.}
the rule expresses that the input to $\tok{t_1}$, $\inputVar[1]$, is a function of $\tok{t_1 + t_2}$'s input, $\inputVar[0]$, and similarly that $\tok{t_2}$'s input, $\inputVar[2]$, is a function of $\tok{t_1 + t_2}$'s input, $\inputVar[0]$, and $\tok{t_1}$'s output, $\outputVar[1]$.
Finally, it also expresses that the $\tok{t_1 + t_2}$'s output, $\outputVar[0]$, is a function of its input, $\inputVar[0]$, and the outputs of $\tok{t_1}$ ($\outputVar[1]$) and $\tok{t_2}$ ($\outputVar[2]$).

When the semantics of a sub-term $\tok{t_i}$ is known (e.g., for nullary productions), we substitute its learned semantics for $\mathit{sem}(\tok{t_i}, \inputVar[i], \outputVar[i])$;
otherwise, we approximate its semantics using examples.
Again, we use a CEGIS loop to generate examples for the entire term $\tok{t_1 + t_2}$, as well as any sub-terms whose exact semantics have not yet been synthesized (e.g., for a sub-term that uses $\tok{+}$ or $\tok{-}$). The process proceeds analogously for most other productions in \Imp.

\paragraph{Semantically recursive productions.}
The final interesting case is for $\tok{while}$ loops, for which the semantics is recursive on the term itself.
For semantically recursive productions, we assume that the semantics can make a recursive call (i.e., effectively acting as if the term itself is a sub-term). We additionally synthesize a predicate determining if the recursive call should be made or not. For $\tok{while~b~do~s}$, we synthesize two semantic rules, one in which the recursive call is made, and one in which it is not.
\resizebox{\textwidth}{!}{
\begin{minipage}{1.2\textwidth}
\begin{mathpar}
\inferrule*{
  \mathit{sem}(\tok{b}, \inputVar[1], \outputVar[1]) \qquad
  \mathit{sem}(\tok{s}, \inputVar[2], \outputVar[2]) \qquad
  \neg \mathit{Pred}_{\mathit{rec}}(\inputVar[0], \outputVar[1], \outputVar[2])\\
  \inputVar[1] = f_1(\inputVar[0]) \\
  \inputVar[2] = f_2(\inputVar[0], \outputVar[1]) \\
  \outputVar[0] = f_0(\inputVar[0], \outputVar[1], \outputVar[2])
}{
  \mathit{sem}(\tok{while~b~do~s}, \inputVar[0], \outputVar[1]) 
}

\inferrule*{
  \mathit{sem}(\tok{b}, \inputVar[1], \outputVar[1]) \\
  \mathit{sem}(\tok{s}, \inputVar[2], \outputVar[2])\\
  \mathit{sem}(\tok{while~b~do~s}, \inputVar[3], \outputVar[3])\\
  \mathit{Pred}_{\mathit{rec}}(\inputVar[0], \outputVar[1], \outputVar[2])\\
  \inputVar[1] = f_1(\inputVar[0]) \\
  \inputVar[2] = f_2(\inputVar[0], \outputVar[1]) \\
  \inputVar[3] = f_2(\inputVar[0], \outputVar[1], \outputVar[2]) \\
  \outputVar[0] = f_0(\inputVar[0], \outputVar[1], \outputVar[2], \outputVar[3])
}{
  \mathit{sem}(\tok{while~b~do~s}, \inputVar[0], \outputVar[1]) 
}
\end{mathpar}
\end{minipage}
}

As with the previous productions, our algorithm
uses a CEGIS loop
to synthesize a candidate semantics of the above form, verify its correctness, and generate a counter-example if the candidate semantics is incorrect.
While we may employ learned semantics for sub-terms, recursive calls to a sub-term must be approximated using examples because we are still in the process of learning its semantics.
We formally define the semantics-synthesis problem that we solve in \Cref{sec:problem-definition} and
explain how our synthesis algorithm
works in \Cref{sec:algorithm}.

\paragraph{Multi-output productions.}
In the above $\tok{while}$-loop example, we saw that the function $f_0$ had four inputs that must be considered when synthesizing a term to instantiate $f_0$. As the number of input variables and the size of the desired result grows, synthesis scales poorly. In the above examples, the notation is not showing the full picture. For \Imp[n] all input and (most) output variables are an $n$-tuple of variables representing a state of an \Imp[n] program. Even for just \Imp[2], $f_0$ has twice as many inputs. 

To address this problem, we allow synthesizing the semantics of each output of a production independently. For example, consider the production $\tok{x_0 \coloneqq t}$ (for \Imp[2]).
We generate a semantics using two constraints $F$ and $G$, independently. The constraint $F$ (resp. $G$) represents the pair of functions $f_0$ and $f_1$ (resp. $g_0$ and $g_1$). 
\begin{mathpar}
\inferrule*[Right=$F$]{
  \mathit{sem}(\tok{t}, \inputVar[1], \inputVar[1]) \\ \inputVar[1] = f_1(\inputVar[0]) \\ {\outputVar[0]} = f_0(\inputVar[0], \outputVar[1])
}{
  \mathit{sem}(\tok{x_0 \coloneqq t}, \inputVar[0], {\outputVar[0]})
}

\inferrule*[Right=$G$]{
  \mathit{sem}(\tok{t}, \inputVar[1], \inputVar[1]) \\ \inputVar[1] = g_1(\inputVar[0]) \\ {\outputVar[0]} = g_0(\inputVar[0], \outputVar[1])
}{
  \mathit{sem}(\tok{x_0 \coloneqq t}, \inputVar[0], {\outputVar[0]})
}
\end{mathpar}
By independently synthesizing $F$ and $G$, we reduce the burden on the underlying synthesizer; however, now the synthesizer is allowed to return
an $F$ and $G$ for which $f_1 \neq g_1$.
Thus, $F$ and $G$ have inconsistent inputs being provided to the child-term $t$. We use an SMT solver to determine if $f_1$ and $g_1$ are consistent for each of the example inputs to the term $\tok{x_0 \coloneqq t}$. If so, we will return either $f_0, g_0, f_1$ (or $f_0, g_0, g_1$ because $f_1$ and $g_1$ are consistent on all examples---i.e., when evaluated on the same example they return equal outputs---otherwise, we discover that $f_1$ and $g_1$ are inconsistent on some input and add a new constraint to ensure that the same pair of functions $f_1$ and $g_1$ cannot be synthesized again.
This optimization is further discussed in \Cref{sec:implementation:optimization}.

%% file: sections/3-problem_definition.tex
\section{Problem Definition} \label{sec:problem-definition}
In this paper, we consider the problem of synthesizing a formal logical semantics for a deterministic language from an executable interpreter.
While there are many possible ways to logically define a semantics, we are interested in an approach that is language-agnostic and inductive.
The \semgus synthesis framework has proposed using Constrained Horn Clauses as a way of defining program semantics that meets both of our desiderata. 
Concretely, \semgus already supports synthesis for a large number of languages (which we consider in our experimental evaluation) by allowing a user to provide a user-defined semantics. 
As mentioned above, in \semgus, semantics are defined inductively on the structure of the grammar (i.e., per production/language construct) using logical relations represented as Constrained Horn Clauses (CHCs) \cite{kim2021semantics}.
In this paper, we follow suit and address the problem of learning a semantics of this form from an executable interpreter for the given language.
This section formalizes the semantics-synthesis problem that we consider. We begin by detailing our representation of syntax (\Cref{sec:problem-definition:syntax}), interpreters (\Cref{sec:problem-definition:interpreter}),
semantics (\Cref{sec:problem-definition:semantics}), and semantic-equivalence oracles (\Cref{sec:problem-definition:oracle}). Finally, we formalize the semantics-synthesis problem in \Cref{sec:problem-definition:problem}.

\subsection{Syntax} \label{sec:problem-definition:syntax}
We consider languages represented as regular tree grammars (RTGs).
A \textbf{ranked alphabet} is a tuple $\tuple{\Sigma, rk_\Sigma}$ that consists of a finite set of symbols $\Sigma$ and a function $\rk_\Sigma : \Sigma \to \NN$ that associates every symbol with a rank (or arity). For any $n \geq 0$, $\Sigma^n \subseteq \Sigma$ denotes the set of symbols of rank $n$. The set of all \emph{(ranked) Trees} over $\Sigma$ is denoted by $T_\Sigma$. Specifically, $T_\Sigma$ is the least set such that $\Sigma^0 \subseteq T_\Sigma$ and if $\sigma^k \in \Sigma^k$ and $t_1,\dots,t_k \in T_\Sigma$, then $\sigma^k(t_1,\dots,t_k) \in T_\Sigma$.
In the remainder of the paper, we assume a fixed ranked alphabet $\tuple{\Sigma, \mathit{rk}_\Sigma}$.

A \textbf{typed regular tree grammar} (RTG) is a tuple $\grammar = \tuple{N, \Sigma, \delta, \Tau, \inputType, \outputType}$, where $N$ is a finite set of non-terminal symbols of rank 0, $\Sigma$ is a ranked alphabet, $\delta$ is a set of productions over a set of types $\Tau$,
and for each non-terminal $A \in N$, and $\inputType_A$ (resp. $\outputType_A$) assigns $A$ an input-type (resp. output-type) from $\Tau$. Each production in $\delta$ takes the form:
\[
  A_0 \to \sigma\left(A_1, A_2, \ldots, A_{\rk_\Sigma(\sigma)}\right)
\]
where $A_i \in N$ and $\sigma \in \Sigma$. We use $\mcL(A)$ to denote the language of non-terminal $A$ and $\delta(A)$ the set of all productions associated with $A$ (i.e., all productions where $A_0$ is $A$). In the remainder, we assume a fixed grammar $\grammar = \tuple{N, \Sigma, \delta, \Tau, \inputType, \outputType}$.

\begin{example}[$G_\text{\Imp}$ as a Regular Tree Grammar]\label{ex:problem-definition:syntax}
Consider the \Imp language detailed in \Cref{sec:example}, $G_\text{\Imp}$ is a regular tree grammar that has been stylized to ease readability. For example, the non-terminals consist of the rank-0
symbols $E$, $B$, and $S$. The productions include $S \to \tok{x_1 \coloneqq}(E)$, $S \to \tok{\seq}(S, S)$, and $S \to \tok{while}(B, S)$. For \Imp[2] (\Imp with two variables $x_1$ and $x_2$), $\inputType_E$ is the type $\mathbb{Z} \times \mathbb{Z}$, representing the state of the two variables, and $\outputType_E$ is $\mathbb{Z}$, representing the return type of arithmetic expressions.
\end{example}

\subsection{Interpreters} \label{sec:problem-definition:interpreter}
We consider a class of deterministic executable interpreters---i.e., a program evaluator for which we may only observe input-output behavior.

\begin{definition}[Interpreter] \label{def:problem-definition:interpreter}
Formally, an \emph{interpreter} for $\grammar$ maps each non-terminal $A \in N$ to a
partial function $\interpreter_A : (\mcL(A) \times \inputType_A) \rightarrow \outputType_A$---with the
interpretation that the interpreter maps a program $t \in \mcL(A)$ and input value
$\inputVal \in \inputType_A$ to some output $\outputVal \in \outputType_A$
if and only if $t$ starting with the input value $\inputVal$ terminates with the
output value $\outputVal$.
\end{definition}

\begin{example}[Interpreters for {\Imp[1]}] \label{ex:problem-definition:interpreter}
Recall the \Imp language defined in \Cref{sec:example}. The interpreter $\interpreter$ for \Imp consists of three base interpreters $\interpreter_E$, $\interpreter_B$, and $\interpreter_S$, which are used to evaluate arithmetic expressions, Boolean expressions, and statements, respectively. Throughout this paper, we assume the interpreters for \Imp[1] (and all \Imp variants) evaluate according to the standard denotational semantics (e.g., $\tok{0}$ is the expression that always returns $0$ regardless of input state;
$\tok{+}$ is mathematical $+$;
$\tok{while~b~s}$ evaluates $b$, executes the loop body, and recurses if $\tok{b}$ evaluates to $\mathit{true}$
and
otherwise immediately terminates; etc.).
\end{example}

\subsection{Semantics} \label{sec:problem-definition:semantics}
We represent the big-step semantics of a language (defined by some grammar $G$) using a set of Constrained Horn Clauses (CHCs) within some background theory $\theory$ per production. While CHCs (at first glance) seem limiting, this formulation of semantics has been employed by the \semgus framework to represent user-defined semantics for many languages \cite{kim2021semantics,semgus-cav21}, including many variations of \Imp, regular expressions, \sygus expressions within the theory of bit vectors, algebraic data types, linear integer arithmetic. 

\begin{definition}[Constrained Horn Clause] \label{def:chc}
A CHC (in theory $\theory$) is a first-order formula of the form:
\[\forall \bar{x}_1, \dots, \bar{x}_n, \bar{x}.~\phi \land R_1(\bar{x}_1) \land \dots \land R_n(\bar{x}_n) \Rightarrow H(\bar{x})\]
where $R_1, \dots, R_n$ and $H$ are uninterpreted relations, $\bar{x}_1, \dots, \bar{x}_n$ and $\bar{x}$ are variables, and $\phi$ is a quantifier-free $\theory$-constraint over the variables.    
\end{definition}

To specify the big-step semantics of a non-terminal $A \in N$ (for which the interpreter has type $\interpreter_A : (\mcL(A) \times \inputType_A) \rightarrow \outputType_A$), we introduce the semantic relation $\mathit{Sem}_A(t_A, \inputVar[A], \outputVar[A])$, where $t_A$ is a variable representing elements of $\mcL(A)$, $\inputVar[A]$ is a variable of type $\inputType_A$, and $\outputVar[A]$ is a variable of type $\outputType_A$. Throughout this paper, we may also use $\sem{t_A}_{\mathit{Sem}}(\inputVar[A]) = \outputVar[A]$
to denote that $\mathit{Sem}_A(t_A, \inputVar[A], \outputVar[A])$ holds.

\begin{example}[Semantic relations]\label{ex:sem-relations}
Consider the \Imp[1] language introduced in Section~\ref{sec:example}; a semantics for \Imp[1] uses the semantic relations:
\begin{equation*}
\mathit{Sem}_E: \mcL{(E)} \times \ZZ \times \ZZ \to \mathtt{bool} \qquad
\mathit{Sem}_B: \mcL{(B)} \times \ZZ \times \mathtt{bool} \to \mathtt{bool} \qquad
\mathit{Sem}_S: \mcL{(S)} \times \ZZ \times \ZZ \to \mathtt{bool}
\end{equation*}
\end{example}

While CHCs are quite general and capable of defining both deterministic and non-deterministic semantics, we limit our scope to CHCs that represent deterministic semantics. Furthermore, for a grammar $G$, we assume that each production $A_0 \to \sigma(A_1, \dots, A_n) \in G$ evaluates sub-terms in a fixed order from left to right (i.e., for a term $p(t_1, \dots, t_n)$ sub-term $t_1$ is evaluated before $t_2$, etc.). While this
imposed order may seem too restrictive, we later show how this restriction can be lifted by considering all permutations of sub-terms.

\begin{definition}[Semantic Rule, Semantic Constraint] \label{def:problem-definition:semantic-rule}
Given a production $A_0 \to p(A_1, \dots, A_n)$ a \textbf{semantic rule} for $p$ is a CHC of the form:
\begin{equation}\label{eqn:semantic-form}
\inferrule*{
  \mathit{Sem}_{A_1}(t_1, \inputVar[1], \outputVar[1]) \\
  \dots \\
  \mathit{Sem}_{A_n}(t_n, \inputVar[n], \outputVar[n]) \\
  F(\inputVar[0], \dots, \inputVar[n], \outputVar[0], \dots, \outputVar[n]) \\
}{
  \mathit{Sem}_{A_n}(p(t_1, \dots, t_n), \inputVar[0], \outputVar[0])
}
\end{equation}
where $F$ is constraint over theory $\theory$, which we call a \textbf{semantic constraint}, that takes the form:
\begin{equation} \label{eqn:semantic-constraint}
  \inputVar[1] = f_1(\inputVar[0]) \land \dots \land \inputVar[n] = f_n(\outputVar[1], \dots, \outputVar[n-1], \inputVar[0]) \land \outputVar[0] = f_0(\outputVar[1], \dots, \outputVar[n], \inputVar[0]) \land P(\inputVar[0], \outputVar[0], \dots, \outputVar[n])
\end{equation}
where each $f_i$ is a function that returns a term of type $\inputType_{A_i}$ for $i > 0$ and $\outputType_{A_0}$ for $i = 0$. The semantic constraint also includes predicate $P(\inputVar[A_0], \outputVar[A_1], \dots, \outputVar[A_n])$ that determines when the semantic rule is valid (e.g., for conditionals and loops).
\end{definition}

\begin{example}[Semantics of $\tok{do\_while}$]
We give the semantics of the $\tok{do\_while}$ \Imp statement below:
\begin{mathpar}
\inferrule{
  \sem{s}(x_1) = x_1' \\
  \sem{b}(x_2) = r_b \\
  \sem{\tok{do}~s~\tok{while}~b}(x_3) = x_3'\\
  r_b \\ x_1 = x_0 \\ x_2 = x_1' \\ x_3 = x_1' \\ x_0' = x_3'
}{
  \sem{\tok{do}~s~\tok{while}~b}(x_0) = x_0'
}

\inferrule{
  \sem{s}(x_1) = x_1' \\
  \sem{b}(x_2) = r_b \\
  \neg r_b \\ x_1 = x_0 \\ x_2 = x_1' \\ x_0' = x_1'
}{
  \sem{\tok{do}~s~\tok{while}~b}(x_0) = x_0'
}
\end{mathpar}
The first rule executes the statement $s$ and then, if the guard $b$ is true recursively executes the whole loop and returns the resulting value. 
The second rule executes the statement $s$ and then, if the guard $b$ is false returns the output produced when executing the statement $s$.
\end{example}

\subsection{Equivalence Oracle and Semantics Synthesis Problem}\label{sec:problem-definition:oracle} \label{sec:problem-definition:problem}
For a grammar $G$, a semantics $\mathit{Sem}$ for $G$, and an interpreter $\interpreter$ for $G$, an \emph{equivalence oracle} is used to determine whether $\mathit{Sem}$ is equivalent to the semantics defined by the interpreter $\interpreter$.

\begin{definition}[Equivalent, Equivalence Oracle] \label{def:problem-definition:oracle}
Given an interpreter $\interpreter$ for a language $G$, a subgrammar $G' \subseteq G$, and a semantics $\mathit{Sem}$ for $G'$,  we say that $\interpreter$ and $\mathit{Sem}$ are \textbf{equivalent} on $G'$ if and only if for every term $t \in \mcL(G')$, input $\inputVal \in \inputType_A$, and output $\outputVal \in \outputType$, we have:
\[I(t, \inputVal) = \outputVal \Leftrightarrow \sem{t}_\mathit{Sem}(\inputVal) = \outputVal\]

An \textbf{equivalence oracle} $\oracle$ for $\interpreter$ is a function that takes as input a semantics $\mathit{Sem}$ for $G'$ and determines if $\mathit{Sem}$ is equivalent to $\interpreter$ on $G'$. If $\mathit{Sem}$ is not equivalent to $\interpreter$, then $\oracle$ returns an example $\tuple{\inputVal, t, \outputVal}$ for which $\interpreter$ and $\mathit{Sem}$ disagree---i.e., there is some term $t$ and input $\inputVal$ such that $\sem{t}_\mathit{Sem}(\inputVal) \neq \sem{t}_\interpreter(\inputVal)$---and otherwise returns $\mathit{None}$ when $\mathit{Sem}$ and $\interpreter$ are equivalent.
\end{definition}

Given a language (a grammar and accompanying interpreter), the semantics synthesis problem is to find some semantics of the language that is equivalent to the interpreter. We formalize the semantics synthesis problem as follows:

\begin{definition}[Semantics-Synthesis Problem, Solution] \label{def:problem-definition}
A \textbf{semantics-synthesis problem} is a tuple $\mathcal{P} \defeq \tuple{G, \interpreter, \oracle}$, where $G$ is a grammar, $\interpreter$ is an interpreter for $G$, and $\oracle$ is an equivalence oracle for $\interpreter$. A \textbf{solution} to the semantics-synthesis problem $\mathcal{P}$ is a semantics $\mathit{Sem}$ for $G$ that is equivalent to $\interpreter$ as determined by $\oracle$.
\end{definition}

%% file: sections/4-algorithm.tex
\section{Semantics Synthesis}
\label{sec:algorithm}

This section presents an algorithm \algname (Algorithm~\ref{alg:sem-synth}) to synthesize a semantics for a language from an executable interpreter. The input to \algname is a semantics-synthesis problem consisting of \rone a grammar $G$, \rtwo an executable interpreter $\interpreter$ for $G$, and \rthree an equivalence oracle $\oracle$ for $\interpreter$. Upon termination, \algname returns a semantics $\mathit{Sem}$ for $G$ that is equivalent to the executable interpreter $\interpreter$ as determined by the equivalence oracle $\oracle$.

Synthesizing a semantics for arbitrary languages comes with several challenges.
In general, semantics are defined as complex recursively defined functions that provide an interpretation to every program within the language. 
Trying to directly synthesize such a semantics is already impractical for relatively small languages, such as the \Imp language defined in \Cref{ex:imp-lang}.

As described in \Cref{sec:problem-definition:semantics}, we consider semantics represented using logical relations defined by a set of Constrained Horn Clauses per production of $G$ (cf. \Cref{def:problem-definition:semantic-rule}). By formulating the desired semantics as CHCs per production, \algname can synthesize the semantics of $G$ one production at a time.
In fact, because \algname uses examples to approximate the semantics of all sub-terms during synthesis (cf. \Cref{sec:alg:overview}), \algname can synthesize the semantics of each production independently.
Finally, by fixing the shape of the semantics (i.e., as a set of CHCs per production), \algname reduces the monolithic synthesis problem to a series of first-order synthesis problems---specifically, by using a \sygus or sketch-based synthesizer to synthesize the constraint of each semantic rule (CHC) defining the semantics of a production.

The remainder of this section is structured as follows: \Cref{sec:alg:overview} provides a high-level overview of how \algname solves semantic-synthesis problems, \Cref{sec:alg:base-synth,sec:alg:verify} provide specifications for \algSemSynth and \algVerify, which synthesize semantic constraints from examples and verify candidate semantic constraints against the interpreter, respectively. Finally, \Cref{sec:alg:sem-recursive} explains how \algname handles semantically recursive productions.

\begin{algorithm}[t]
\Proc{\upshape\algname($G$, $\interpreter$, $\oracle$)}{
  \ForEach{production $p$ of $G$}{
    $\exampleset \gets \emptyset$ \tcp*{Example Set for Production $p$}
    \DoWhile{$\mathit{CEX} \neq \emptyset$}{
      $\mathit{Sem}[p] \gets \algSemSynth(p, E)$
      \tcp*{Get candidate semantics}
      \label{Li:CallToSynthSemanticEq}
      $\mathit{CEX} \gets \algVerify(\mathit{Sem}[p], p, \interpreter, \oracle)$
      \tcp*{Check candidate semantics}
      \label{Li:CallToVerify}
      \uIf{$CEX \neq \emptyset$}{
        $\exampleset \gets \exampleset \cup \mathit{CEX}$ \tcp*{Update example set}
      }
    }
  }
  \Return{$\mathit{Sem}$}\;
}
\caption{Semantics-Synthesis Algorithm} \label{alg:sem-synth}
\end{algorithm}

\subsection{Overview of \algname} \label{sec:alg:overview}

\algname (Algorithm~\ref{alg:sem-synth}) uses the counter-example-guided synthesis (CEGIS) paradigm to synthesize a semantics for $G$ that is equivalent to $\interpreter$ according to the equivalence oracle $\oracle$. Throughout this section, we will use the \Imp language from \Cref{ex:imp-lang} to illustrate how \algname operates.

\paragraph{Synthesizing a Candidate Semantics.}
After initialization, \algname synthesizes the semantics of each production. \algname employs a CEGIS loop to synthesize the semantics of each production. During each iteration, \algname first synthesizes a candidate semantic constraint (cf. \Cref{def:problem-definition:semantic-rule}) for production $p$ using \algSemSynth. The procedure \algSemSynth returns some semantic constraint for $p$ that satisfies the set of examples $\exampleset$. \Cref{sec:alg:base-synth} provides a formal specification of \algSemSynth's operation.

\algname then uses the procedure \algVerify to determine if the semantics synthesized for production $p$ is consistent with the interpreter $\mathcal{I}$ as determined by the equivalence oracle $\oracle$. A formal specification of \algVerify is provided in \Cref{sec:alg:verify}.
If \algVerify determines that the candidate semantics of $p$ is correct, then \algVerify returns an empty set of examples and \algname proceeds to synthesize the semantics of the next production. Otherwise, if \algVerify determines that the candidate semantics of $p$ is not equivalent to the interpreter $\interpreter$, \algVerify returns a set of examples. The new examples are added to the example set $\exampleset$, and the CEGIS loop repeats and synthesizes a new candidate semantics for $p$.

\subsection{Specification of \algSemSynth} \label{sec:alg:base-synth}
Before formally specifying \algSemSynth (\Cref{sec:alg:base-sem-spec}), we first define example sets (\Cref{sec:alg:example-set}) and when a semantic constraint is consistent with an example set (\Cref{sec:alg:example-consistency}).

\subsubsection{Example Sets} \label{sec:alg:example-set}
For an interpreter $\interpreter$, an example set $E$ is a set of examples consistent with $\interpreter$.

\begin{definition}[Example set for interpreter $\interpreter$] \label{def:alg:example-set}
Given an interpreter $\interpreter$ for grammar $G$, an example set $\exampleset$ for interpreter $\interpreter$ is a finite set of examples of the form $\tuple{\inputVal, t, \outputVal}$, where $t \in L(G)$ and $\interpreter(t, \inputVal) = \outputVal$.
\end{definition}

\begin{example}[Example set for {\Imp[1]}] \label{ex:alg:example-set}
Recall the interpreter $\interpreter_\text{\Imp[1]}$ described in \Cref{ex:problem-definition:interpreter} for language \Imp[1].
An example set $E$ for $\interpreter_\text{\Imp[1]}$ might include the examples $\tuple{0, \tok{0}, 0}$, $\tuple{1, \tok{0}, 0}$, $\tuple{1, \tok{x \coloneqq 0; x \coloneqq x + 4}, 4}$, and $\tuple{10, \tok{while~0 < x~do~x \coloneqq x - 1}, 0}$; however, an example set for $\interpreter_\text{Imp}$ could not include any example of the form $\tuple{n, \tok{while~0 < x~do~x \coloneqq x + 1}, n'}$ where $n$ (the initial value of $\tok{x}$) is some positive number. Since, $\tok{while~0 < x~do~x \coloneqq x + 1}$ would not terminate on the input $n$. The example $\tuple{n, \tok{while~0 < x~do~x \coloneqq x + 1}, n'}$ would violate the assumption that $E$ only contains examples consistent with the interpreter $\interpreter_\text{\Imp}$.
\end{example}

\subsubsection{Example Consistency} \label{sec:alg:example-consistency}
In \algname, we use the example set $E$ to ensure that the semantic constraint returned by \algSemSynth is consistent with $\interpreter$ for at least the examples appearing in $E$.

\begin{definition}[Consistency with Example Set] \label{def:alg:example-consistency}
Given a production $A_0 \to p(A_1, \dots, A_n)$, a semantic rule $R$ with semantic constraint $F$ of the form defined in \Cref{def:problem-definition:semantic-rule}, and example set $E$, we say $R$ is \textbf{consistent} with $E$ if and only if the semantic constraint $F$ is consistent with $E$. Furthermore, the semantic constraint $F$ is \textbf{consistent} with the example set $E$ if for every example $\tuple{\inputVal_{A_0}, p(t_1, \dots, t_n), \outputVal_{A_0}} \in E$ the following condition holds:
\begin{equation} \label{eqn:semantic-constraint-condition}
  \forall \inputVar[0], \dots, \inputVar[n], \outputVar[0], \dots, \outputVar[n].~
  \left(\begin{array}{r@{\hspace{1.0ex}}l}
          & \inputVar[0] = \inputVal_{0} \\
    \land & \mathit{Summary}(t_1) \\
          & \dots \\
    \land & \mathit{Summary}(t_n) \\
    \land & F
  \end{array}\right) \Rightarrow \outputVar[0] = \outputVal_{0}
\end{equation}
where $\mathit{Summary}(t_i) = \bigvee \{{\inputVar[i]} = \inputVal_{i} \land {\outputVar[i]} = \outputVal_{i} : \tuple{\inputVal_i, t_i, \outputVal_i} \in E\}$ summarizes the semantics of $t_i$ according to the examples found in $E$.
\end{definition}

\begin{example}[Example Consistency] \label{ex:examplec-consistent}
Consider the production for the operator $\tok{+}$, and the (correct) semantic constraint $F \defeq \inputVar[1] = \inputVar[0] \land \inputVar[2] = \inputVar[0] \land \outputVar[0] = \outputVar[1] + \outputVar[2]$; $F$ is consistent with the examples $\tuple{0, \tok{x_0 + 1}, 1}$, $\tuple{0, \tok{x_0}, 0}$, and $\tuple{0, \tok{1}, 1}$. Specifically, the following formula is valid:
\[
\forall \inputVar[0], \inputVar[1], \inputVar[2], \outputVar[0], \outputVar[1], \outputVar[2].~(
 \inputVar[0] = 0 \land (\inputVar[1] = 0 \land \outputVar[1] = 0) \land (\inputVar[1] = 0 \land \outputVar[1] = 1) \land F
) \Rightarrow \outputVar[0] = 1.
\]
\end{example}

\subsubsection{Formal Specification of \algSemSynth} \label{sec:alg:base-sem-spec}
The procedure \algSemSynth takes as input the production $p$ whose semantics is to be synthesized and the current example set $E$;
it returns a constraint $F$---of the form defined in \Cref{def:problem-definition:semantic-rule}---defining a semantics for production $p$ that is consistent with the example set $E$.

\begin{example}[Synthesizing semantics of $\tok{x \coloneqq}$ consistent with examples] \label{ex:alg:sem-eq}
Recall that for the language \Imp, the semantics of the production $\tok{x \coloneqq}$ is represented as (a set of) CHC rule(s) of the form:
\[\inferrule*{
  \mathit{Sem}_E(e, \inputVar[1], \outputVar[0]) \land \inputVar[1] = f(\inputVar[0]) \land \outputVar[0] = g(\inputVar[0], \outputVar[1])
}{
  \mathit{Sem}_S(\tok{x \coloneqq} e, \inputVar[0], \outputVar[0])
}\]
for some functions $f$ and $g$ (in the theory of linear integer arithmetic).
The procedure  call $\algSemSynth(\tok{x \coloneqq}, E)$ synthesizes the formulas $f(\inputVar[0]) = t_f$ and $g(\inputVar[0], \outputVar[1]) = t_g$, and returns the constraint $F \defeq \inputVar[1] = t_f \land \outputVar[0] = t_g$ so that $F$ is consistent with $E$. 

We note that for functions expressible in a decidable first-order theory, this problem can be exactly encoded as a Syntax-Guided Synthesis (\sygus) problem~\cite{sygus} and solved by a \sygus solver (e.g., \textsc{cvc5}~\cite{cvc5}).
\end{example}


\subsection{Specification of \algVerify} \label{sec:alg:verify}
The procedure \algVerify takes as input the production $p$, a candidate semantics of $p$,
the interpreter $\interpreter$, and the equivalence oracle $\oracle$; it determines if $\mathit{Sem}$ is equivalent to the interpreter $\mathcal{I}$ for all terms of the form $p(t_1, ..., t_k) \in L(G)$.
If \algVerify determines that the candidate semantics of $p$ is not equivalent to $\interpreter$, \algVerify returns a set of counter-examples $\textit{CEX}$ such that
\rone $\textit{CEX}$ is consistent with $\interpreter$,
\rtwo $\textit{CEX}$ is not consistent with the candidate semantics of production $p$, and
\rthree for the input production $p$, there is exactly one example of the form $\tuple{i, p(t_1, \dots, t_k), o}$ appearing in $\textit{CEX}$ (for any other production $p' \neq p$, there can be many examples of the form $\tuple{i', p'(t_1, \dots, t_k), o'}$ in $\textit{CEX}$).
Otherwise, \algVerify returns an empty-set to signify that the semantics of $p$ is equivalent to $\interpreter$ for all terms of the form $p(t_1, \dots, t_k) \in L(G)$.

\begin{example}[Synthesizing Semantics of $\tok{0}$ for $\grammar_\text{\Imp}$]\label{ex:alg:leaf}
Recall the \Imp language in \Cref{ex:imp-lang}.
On some iterations, \algname will consider
the production $\tok{0}$ (a leaf/nullary production). During the first iteration of the CEGIS loop for $\tok{0}$, the example set $E$ will be empty and
$\algSemSynth$ may return any constraint $F$ of the form $\outputVar[0] = f(\inputVar[0])$. Assume that \algSemSynth returns the constraint $\outputVar[0] = 1$. \algVerify returns the counter-example $\tuple{0, \tok{0}, 0}$, and the example set $E$ is updated.

In the next iteration, the CEGIS loop must return a constraint satisfying the updated example set. For example, suppose that
\algSemSynth returns the constraint $\outputVar[0] = \inputVar[0]$. Again, \algVerify determines that $\outputVar[0] = \inputVar[0]$ is incorrect and returns the new counter-example $\tuple{1, \tok{0}, 0}$. The example set $E$ is updated with the returned counter-example.

A new iteration of the loop is run. On this iteration, \algSemSynth must return a constraint that satisfies both of the previously returned examples. This time \algSemSynth returns the constraint $\outputVar[0] = 0$, \algVerify determines that $\outputVar[0] = 0$ is correct, and \algname proceeds to synthesize the semantics of the next production (e.g., $\tok{1}$).
\end{example}

In \Cref{ex:alg:leaf}, we see how \algname handles nullary (leaf) productions. \algname works nearly identically for most production rules (excluding semantically recursive productions like $\tok{while}$ loops). We demonstrate in \Cref{ex:alg:seq} how \algname synthesizes a semantics for non-nullary productions.

\begin{example}[Synthesizing Semantics of Sequencing for \Imp.]\label{ex:alg:seq}
Continuing from \Cref{ex:alg:leaf}, \algname proceeds and comes to the sequencing operator (i.e., for production $S \to \tok{\seq}(S, S)$). After several attempts at synthesizing the semantics of sequencing, $E$ contains the examples $\tuple{0, \tok{x \coloneqq 1; x \coloneqq 0}, 0}$, $\tuple{0, \tok{x \coloneqq 0; x \coloneqq x + 1}, 1}$, and $\tuple{1, \tok{x \coloneqq 0; (x \coloneqq 1; x \coloneqq x + 1)}, 2}$.

In addition to these examples, we summarize the semantics of each example's sub-term with further examples in the example set $\exampleset$. These summarized examples of sub-terms are generated by data-flow propagation through the term $p(t_1, \dots, t_k)$ using the input $i$.
Because the execution output of a certain sub-term $t_j$ can be used as input for any following term $t_l$ where $l>j$, we repeatedly enumerate all possible inputs for each sub-term (and add them into $E$) until we reach a fix-point, i.e., no new examples for sub-terms are found.
\algname then generates the formula specifying that the desired semantic constraint is consistent with the example set $\exampleset$ using the generated summaries, and produces a new semantic constraint using \algSemSynth. On this iteration, \algSemSynth returns the correct semantic constraint, \algVerify determines that it is correct, and \algname proceeds to synthesize a semantics for the next production.
\end{example}

\subsection{Synthesizing Semantics for Semantically Recursive Productions}
\label{sec:alg:sem-recursive}

So far, we have seen how \algname handles nullary productions and structurally recursive productions (e.g., $\tok{ite}$ and sequencing). However, we have not yet seen how to handle productions that are \emph{semantically} recursive (e.g., $\tok{while}$ loops). To handle semantically recursive productions, we augment the form of the desired constraint to be synthesized:
\algSemSynth must synthesize a predicate $\recpred$ and two base constraints $F_\mathit{nonrec}$ and $F_\mathit{rec}$ such that for every example $\tuple{\inputVal, p(t_1,\dots,t_n), \outputVal}$, the following conditions hold:

\resizebox{\textwidth}{!}{
\begin{minipage}{1.2\textwidth}
\begin{mathpar}
\inferrule*[right=\text{non-rec}]{
  \mathit{Sem}_{A_1}(t_1, \inputVar[A_1], \outputVar[A_1]) \\
  \dots \\
  \mathit{Sem}_{A_n}(t_n, \inputVar[A_n], \outputVar[A_n]) \\
  \neg \recpred(\inputVar[A_0], \outputVar[A_1], \dots, \outputVar[A_n])\\
  F_\mathit{non-rec}(\inputVar[A_0], \outputVar[A_1], \dots, \outputVar[A_n]) \\
  \inputVar[A_0] = \inputVal
}{
  \outputVar[A_0] = \outputVal
}

\inferrule*[right=\text{rec}]{
  \mathit{Sem}_{A_1}(t_1, \inputVar[A_1], \outputVar[A_1]) \\
  \dots \\
  \mathit{Sem}_{A_n}(t_n, \inputVar[A_n], \outputVar[A_n]) \\
  \mathit{Sem}_{A_0}(p(t_1, \dots, t_n), {\inputVar[A_0]}', {\outputVar[A_0]}')\\
  \recpred(\inputVar[A_0], \outputVar[A_1], \dots, \outputVar[A_n])\\
  F_\mathit{rec}(\inputVar[A_0], \outputVar[A_1], \dots, \outputVar[A_n]) \\
  \inputVar[A_0] = \inputVal
}{
  \outputVar[A_0] = \outputVal
}
\end{mathpar}
\end{minipage}
}
where $\recpred$ determines if the non-rec or rec condition should hold. The non-recursive case is similar to the conditions for non-semantically recursive statements (with the addition of asserting that $\recpred$ is false). The recursive case, however additionally allows the semantics to make use of a recursive call to the program term. Other than the change in the shape of the desired semantics, \algname remains unchanged.

\begin{example}[Synthesizing semantics of while loops for \Imp.]\label{ex:alg:while}
Continuing from \Cref{ex:alg:seq}, \algname eventually considers the $\tok{while}$ production. We assume that the grammar $G$ additionally annotates whether each production is semantically recursive.

After several iterations of the CEGIS loop, the example set $\exampleset$ contains the examples $\tuple{0, t, 0}$, $\tuple{1, t, 0}$, and $\tuple{2, t, 0}$, where $t$ is the term
$\tok{while~0 < x~do~x := x - 1}$. In this iteration, \algSemSynth gets called with a recursive summary of $t$ containing the three examples,
and examples for $\tok{x := x - 1}$ and $\tok{0 < x}$.

In this iteration, \algSemSynth finds the correct $\recpred$, $F_\mathit{non-rec}$ and $F_\mathit{rec}$. \algVerify determines that the result is indeed correct and the main loop of \algname continues to the next production. If $\tok{while}$ is the last production of the considered grammar $G$, then \algname terminates and returns the synthesized semantics for each production.
\end{example}

Now that we have defined how \algname handles semantically recursive productions, \algname is fully specified. Theorem~\ref{thm:alg-sound} states that \algname is sound.

\begin{theorem}[\algname is sound] \label{thm:alg-sound}
For any semantics-synthesis problem $\mathcal{P} = \tuple{G, \interpreter, \oracle}$, if $\algname(G, \interpreter, \oracle)$ returns a semantics $\mathit{Sem}$, then $\mathit{Sem}$ is a solution to $\mathcal{P}$.
\end{theorem}

\begin{proof}
    \algname iterates over the productions in some order, say $p_0, \dots, p_{k-1}$. For all iterations $0 \leq i \leq k$, \algname maintains the invariant that the synthesized semantics $\mathit{Sem}$ is correct with respect to the oracle $\oracle$ for all previously considered productions $p_0$ through $p_{i-1}$. This condition trivially holds on the first iteration. To proceed to iteration $i+1$, the CEGIS loop for production $p_i$ must terminate. For the CEGIS loop to terminate, \algVerify must return an empty set of counter-examples, which implies that $\mathit{Sem}$ is correct for the production $p_i$ (and that the semantics for productions $p_1, \dots, p_{i-1}$ were left unmodified)---and thus the invariant is maintained. The algorithm only terminates after exploring all productions. Consequently, upon termination, $\mathit{Sem}$ must be correct for all productions of $G$---i.e., $\mathit{Sem}$ satisfies the given semantics-synthesis problem $\mathcal{P}$.
\end{proof}

While \Cref{thm:alg-sound} states the soundness of \algname, it fails to show that \algname will eventually synthesize a correct semantics. \Cref{thm:alg-progress} states that \algname makes progress. Intuitively,
it states that once a semantic rule for production $p$ is explored during some iteration of the CEGIS loop, it is never explored in any future iteration of the CEGIS loop for production $p$.

\begin{theorem}[\algname makes progress] \label{thm:alg-progress}
For any semantics-synthesis problem $\mathcal{P} = \tuple{G, \interpreter, \oracle}$, if $\algname(G, \interpreter, \oracle)$ is synthesizing the semantics of production $p$ and on the $k^{\textit{th}}$ iteration of the CEGIS loop for production $p$, \algSemSynth produces the semantic relation $R_k$, then for all future iterations $j > k$, \algSemSynth will return some relation $R_j \neq R_k$.
\end{theorem}
\begin{proof}
Assume that the negation holds, i.e., ``$\exists j>k. R_j = R_k$''. By the assumption $j > k$, it must be that \algVerify$(R_k, p, \interpreter, \oracle)$ returned a non-empty set of examples $\textit{CEX}$. Otherwise, the CEGIS loop for production $p$ would have immediately terminated and not continued to iteration $j$. By definition, $R_k$ is inconsistent with the set of counter-examples $\textit{CEX}$. The returned counter-examples $\textit{CEX}$ are then added to the example set $\exampleset$ for all future iterations. By assumption, $R_j$ must be consistent with the example set $\exampleset$, and thus $R_j$ must not be $R_k$, a contradiction.
\end{proof}

%% file: sections/5-implementation.tex
\section{Implementation}
\label{sec:implementation}

This section gives details of \toolname, which implements our approach to synthesizing semantics via the algorithm \algname.
\toolname is developed in Scala (version 2.13), and uses \textsc{cvc5} (version 1.0.3) to solve \sygus problems---which are used within our implementation of \algSemSynth to generate candidate semantic constraints. The remainder of this section is structured as follows: \Cref{sec:implementation:sem-synth} details how we implement \algSemSynth.
\Cref{sec:implementation:verify} summarizes the implementation of \algVerify, and explains
how we approximate an equivalence oracle for an interpreter.
\Cref{sec:implementation:optimization} presents
an optimization of \algSemSynth for productions with multiple outputs (i.e., where the output type of a production is a tuple).

\subsection{Implementation of \algSemSynth} \label{sec:implementation:sem-synth}

In \Cref{sec:algorithm}, \algname is parameterized on the procedure \algSemSynth.
On line \ref{Li:CallToSynthSemanticEq} of Algorithm \ref{alg:sem-synth}, we assume that \algSemSynth produces a semantic constraint $F$ for production $p_i$ that satisfies the example set $E$. To accomplish this task, we construct a \sygus problem consisting of a grammar of allowable semantic constraints and a set of conditions to enforce that the semantic constraint is consistent with the example set. 
To handle productions whose semantics does not evaluate its child terms from left to right, we run in parallel a version of \algSemSynth for each permutation of the child terms and immediately return upon any permutation's success. 
{In practice, for all of our benchmarks, all the productions evaluate their children from left to right.}

We defer discussion of the \sygus grammars we use to \Cref{sec:eval:benchmarks} when we discuss each benchmark. The specification of the semantic constraint is exactly the condition specified in \Cref{eqn:semantic-constraint-condition}.



\subsection{Implementation of \algVerify} \label{sec:implementation:verify}
In \Cref{sec:algorithm}, Algorithm \ref{alg:sem-synth} is parameterized on the procedure \algVerify (line \ref{Li:CallToVerify}), which uses the equivalence oracle $\oracle$ to determine if the learned semantics $\mathit{Sem}$ is consistent with the interpreter for all terms of the form $p(t_1, \dots, t_k) \in L(G)$ for some production $p$.
In \toolname, we approximate an equivalence oracle using fuzzing. Specifically, we randomly generate terms and inputs and use the interpreter $\interpreter$ to generate an output. We then use the learned constraint for $p_i$ to generate inputs to each sub-term (from left to right), and compute outputs for each using interpreter $\interpreter$. In effect, we are computing a new example set $E'$, and testing the semantic constraints learned so far.
If any example disagrees with the learned semantics of production $p$, we return the example (and necessary child-term summaries) as a counter-example.

When \algVerify fuzzes the semantics, it uses the interpreter to generate examples (i.e.,{ terms with corresponding input-output examples)}. During example generation, we set a recursion limit of 1,000 recursive calls. 
We discard an example---i.e., we assume the program does not terminate---if its run exceeds the recursion depth.
We then evaluate the candidate semantic constraint from left-to-right to ensure that the semantic constraint is consistent with each of the generated examples. We return the first example (and the child-term summaries for the example) that is inconsistent with the candidate constraint.

\input{algorithms/verify}

\subsection{Optimized \algSemSynth for Multi-Output Productions} \label{sec:implementation:optimization}
In \Cref{sec:implementation:sem-synth}, we described how \algSemSynth produces and solves (using \textsc{cvc5}) a \sygus problem to synthesize a semantic constraint that is consistent with the current example set. However, it is well known that \sygus solvers scale poorly as a function of the size of the desired grammar/result. This issue is especially problematic when learning a semantic constraint for a language in which productions have multiple outputs (e.g., statements for \Imp with more than one variable) and thus the grammar and resulting constraint grow with the number of outputs.

For some languages, it is possible to augment the semantics of the language to use a suitable theory to encode multiple outputs as a single output---e.g., using the theory of arrays to support multi-variable states in the \ImpArr~language (cf.\  \Cref{sec:eval:benchmarks}). However, for other languages this methodology may require the use of theories that are not well suited for existing \sygus and SMT solvers (e.g., $\textsc{RegEx}(k)$ in \Cref{sec:eval:benchmarks} would require the theory of strings).
Instead, for such instances we developed a variant of
\algSemSynth that synthesizes a constraint for each output independently. However, this process may lead to constraints that do not agree on the internal data flow of the constraints (i.e., the functions determining the input to each child term). To remedy this issue, our implementation of \algSemSynth uses an additional CEGIS loop that resynthesizes the constraint for each output until all agree on the inputs to each child term.

We detail \algSemSynth for $N$ outputs in Algorithm~\ref{alg:sygus-multi}.
For simplicity, we explain how Algorithm~\ref{alg:sygus-multi} works for a production that has two outputs (i.e., $N = 2$). Consider the case for $A_0 \to p(A_1, \dots, A_n)$ where $\outputType_{A_0} \defeq \outputType_1 \times \outputType_2$. In this scenario, our goal is to synthesize two constraints $F$ and $G$ (i.e.,  $F = F_1$ and $G = F_2$),
\begin{align}\label{eq:multi-output-funcs}
F \defeq& x_1 = f_1(x_0) \land \dots \land {x_n} = f_n(x_0, x_1', \dots, x_{n-1}') \land {x_0}_0' = f_0(x_0, x_1', \dots, x_n')\\
G \defeq& x_1 = g_1(x_0) \land \dots \land {x_n} = g_n(x_0, x_1', \dots, x_{n-1}') \land {x_0}_1' = g_0(x_0, x_1', \dots, x_n')
\end{align}

\input{algorithms/multi-output}

To determine if $F$ and $G$ agree on each child term's input for example set $E'$, we generate the formula $\phi$ shown below, for each example $\tuple{\inputVal, p(t_1, \dots, t_n), \outputVal} \in E'$:
\begin{equation} \label{eq:multi-output-check}
\begin{aligned}
x_0^F = x_0^G = \inputVal
&\land \mathit{Summary}(t_1)(x_1^F, x_1^{'F}) \land \dots \land \mathit{Summary}(t_n)(x_n^F, x_n^{'F}) \\
&\land \mathit{Summary}(t_1)(x_1^G, x_1^{'G}) \land \dots \land \mathit{Summary}(t_n)(x_n^G, x_n^{'G}) \\
&\land F \land G \land (x_1^F \neq x_1^G \lor \dots \lor x_n^F \neq x_n^G)
\land \tuple{{x_0}_0^{'F}, {x_0}_1^{'G}} = \outputVal
\end{aligned}
\end{equation}
which asks if $F$ and $G$ agree on the input to each child term for the given example. 
To make this concept concrete, consider the following example.

\begin{example}[Synthesizing Semantic Constraint for Multi-Output Production.] \label{ex:multi-output}
Consider the task of synthesizing a semantics for $\tok{x_0 \coloneqq}$ in the language \Imp[2], using the examples:
$\tuple{\tuple{0,1}, \tok{x_0 \coloneqq x_1}, \tuple{1, 1}}$, $\tuple{\tuple{0,1}, \tok{x_1}, 1}$, $\tuple{\tuple{1,1}, \tok{x_1}, 1}$.

For the above examples, \algSemSynth might generate
$F \defeq \inputVar[{1,0}] = \inputVar[{0,0}] \land \inputVar[{1,1}] = \inputVar[{0,1}] \land \outputVar[{0,0}] = \outputVar[{1,1}]$ and
$G \defeq \inputVar[{1,0}] = \inputVar[{0,1}] \land \inputVar[{1,1}] = \inputVar[{0,1}] \land \outputVar[{0,0}] = \outputVar[{1,1}]$, where 
$\inputVar[{i,j}]$
is the $j^{\textit{th}}$ projection of $\inputVar[i]$.
While both $F$ and $G$ are consistent with the examples, the data-flow of $F$ is not consistent with the data-flow of $G$ (i.e., in $F$, $\inputVar[{1,0}]$ is assigned $\inputVar[{0,0}]$, while in $G$,
$\inputVar[{1,0}]$ is assigned $\inputVar[{0,1}]$).
We can construct the formula in \Cref{eq:multi-output-check} for $F$ and $G$, and find out that in $F$, the variable ${\inputVar[{1,0}]}^F$
takes the value $0$, and in $G$, the variable ${\inputVar[{1,0}]}^G$
takes value $1$. 
Thus, $F$ and $G$ are not consistent on data-flows to children for the provided example. We generate a new condition for the next iteration of \algSemSynth that asserts ${\inputVar[{0,0}]}^{F} \neq 0 \lor {\inputVar[{0,0}]}^{G} \neq 1$.
\end{example}

In practice, we create a copy of each variable indexed by $F$ and $G$, respectively, to avoid clashing variable names when encoding the constraints $F$ and $G$ within a single formula. To check the consistency of $F$ and $G$'s data flows, we use \textsc{cvc5} to check the satisfiability of the formula $\phi$ in \Cref{eq:multi-output-check}. If $\phi$ is unsatisfiable, then $F$ and $G$ must agree on the inputs of all child terms for the given examples. If so, then we may return either $F \land {x_0}_2 = g_0(\dots)$ or $G \land {x_0}_1 = f_0(\dots)$ (i.e., because $F$ and $G$ agree on all child term inputs, we may use either to constrain the data-flow to child terms).

If $\phi$ is satisfiable, then $F$ and $G$ do not agree on the input to all child terms. In this case, we find a model that satisfies $\phi$. If there is some subterm $t_i$ such that there is no example $\tuple{\inputVal, t_i, \outputVal} \in E$ such that $\inputVal = M(x_i^F)$ or $\inputVal = M(x_i^G)$, then we add the example $\tuple{\inputVal, t, \interpreter(t_i, \inputVal)}$ to the set of examples, and resynthesize the constraints $F$ and $G$. Otherwise, we know that the sub-term summaries are sufficient to fully specify both $F$ and $G$ for all examples in $E$. Thus, we must add a new constraint that ensures the pair of constraints $F$ and $G$ are never synthesized again. To do this, we add a new constraint $x_0^F \neq M(x_0^F) \lor x_0^G \neq M(x_0^G) \lor \dots \lor x_n^F \neq M(x_n^F) \lor x_n^G \neq M(x_n^G)$, which ensures that the input of at least one of the child terms for either $F$ or $G$ must change. A new candidate $F$ and $G$ are then synthesized. The CEGIS loop continues until it finds a valid pair of $F$ and $G$ for the set of examples.

%% file: algorithms/verify.tex
\begin{algorithm}[t]
\Proc{\algVerify$(R_p, p, \interpreter, \cdot)$}{
    $E \gets \text{random set of examples of the form } \tuple{\mathit{in}, p(t_1, \dots, t_k), \mathit{out}} \text{ consistent with } \interpreter$\; 

    $(\bigwedge_i \inputVar[i] = f_i) \wedge \outputVar[i] = f_0 \gets R_p$
    \tcp*{ Destruct semantic constraint to recover each $f_i$}

    \For{$\tuple{\mathit{in}, p(t_1, \dots, t_k), \mathit{out}} \in E$}{
        $E' \gets \tuple{\mathit{in}, p(t_1, \dots, t_k), \mathit{out}}$\;
        $M \gets \{\inputVar[0] \mapsto \mathit{in} \}$
        \tcp*{Build up model to evaluate sub-terms's semantics}
        \For{$i \gets 1$ to $k$}{
            $\mathit{in}_i \gets \sem{f_i}_M$
            \tcp*{Get input to term $t_i$.}
            $\mathit{out}_i \gets \interpreter(t_i, \mathit{in}_i)$
            \tcp*{Evaluate term $t_i$}
            $M \gets M[\outputVar[i] \mapsto \mathit{out}_i]$
            \tcp*{Update model. Ensures next term's input is defined.}
            $E' \gets E' \cup \{\tuple{\mathit{in}_i, t_i, \mathit{out}_i}\}$
            \tcp*{Add sub-term's summary to set of examples}
        }
        \uIf{$\sem{f_0}_M \neq \mathit{out}$}{
          \Return{$E'$}
          \tcp*{The output computed by evaluating $R_p$ is inconsistent with $E'$}
        }
    }

    \Return{$\emptyset$}
}
\caption{Verifier implementation using (approximate) fuzzing based oracle.}
\label{alg:verify}
\end{algorithm}

%% file: algorithms/multi-output.tex
\begin{algorithm}[t]
\Proc{$\algSemSynth(p, E)$}{
    $A_0 \to \sigma(A_1, \dots, A_n) \gets p$\;
    $\tau_1 \times \dots \times \tau_N \gets \tau_{A_0}$ \tcp*{Determine number of outputs for production $p$.}
    $D \gets \top$ \tcp*{Data flow constraints.}
    \While{$\top$}{
        $\Gamma \gets D$\;
        \For{$i \gets 1 \text{ to } N$}{
            \tcp{Construct per-output conditions (c.f., \Cref{eq:multi-output-funcs})}
            $F_i \gets x_1 = f_1^{F_i}(x_0) \land \cdots \land x_n = f_n^{F_i}(x_0, x_1', \dots, x_{n-1}') \land {x_0}_i' = f_0^{F_i}(x_0, x_1', \dots, x_n')$\;
            $\Gamma \gets \Gamma \land F_i$\;
        }

        \tcp{Generate \sygus conditions (c.f., lines 1-2 of \Cref{eq:multi-output-check})}
        $\phi \gets  \left(\bigwedge_{i, j} \mathit{Summary}(t_i)(x_{i}^{F_j}, {x_{i}^{F_j}}') \right) \land \langle {{x_0}_{0}^{F_0}}', \dots, {{x_0}_{n}^{F_n}}'  \rangle = \mathit{out}$\;
        $\Gamma \gets \phi \land \left(\bigwedge_i x_0^{F_i} = \mathit{in}\right)$\;
        $m \gets \algSolveSygus(\Gamma)$\;
        $M \gets \algSolveSAT(\phi)$\;

        \tcp{Check if inconsistency is found (line 3 of \Cref{eq:multi-output-check})}
        \eIf{$M.\mathit{sat} \land \exists i,j,k: M(x_i^{F_j}) \neq M(x_i^{F_k})$}{
            \tcp{Inconsistency is caused by inaccurate summary of a child term}
            \lIf{$\exists t_i:\forall \langle \mathit{in}, t_i, \mathit{out}\rangle \in E: \forall j: \mathit{in} \neq M(x_i^{F_j})$}{
                $E \gets E \cup \{\langle \mathit{in}, t, \interpreter(t_i, \mathit{in})\rangle\}$
            }
            \tcp{Real data flow inconsistency}
            \lElse{$D \gets D \land \left(\bigvee_{i,j} x_i^{F_j} \neq M(x_i^{F_j})\right)$}
        }{
        \tcp{No inconsistency}
            merge $f_i^{F_j} \in m$ to form the solution\;
        \Return{solution}\;
        }
    }
}
\caption{\algSemSynth for multi-output productions.}
\label{alg:sygus-multi}
\end{algorithm}

%% file: sections/6-evaluation.tex
\section{Evaluation} \label{sec:evaluation}


The goal of our evaluation is to answer the following questions:
\begin{description}
    \item[RQ1] Can \toolname synthesize the semantics of non-trivial languages?
    \item[RQ2] Where is time spent during synthesis?
    \item[RQ3] Is the multi-output optimization from \Cref{{sec:implementation:optimization}} effective?
    \item[RQ4] How do synthesized semantics compare to manually written ones?
\end{description}

All experiments were run on a machine with an Intel(R) i9-13900K CPU and 32 GB of memory, running NixOS 23.10 and Scala 2.13.13.
All experiments were allotted 2 hours, 4 cores of CPU, and 24 GB of memory.
\textsc{Cvc5} version 1.0.3 is used for SMT solving and \sygus function synthesis.
For the total running time of each experiment, we report the median of 7 runs using different random seeds.
For every language, we record whether \toolname terminates within the given time limit of 2 hours, and when it does, we also record the set of synthesized semantic rules.
A language that does not terminate within the time limit on more than half of the seeds is reported as a timeout.


\subsection{Benchmarks}\label{sec:eval:benchmarks}

We collected \jiangyichanged{15} benchmarks from the two sources discussed below.
For every language discussed in this section, we manually translated the semantics to a simple equivalent interpreter written in Scala;
our goal was then to synthesize an appropriate CHC-based semantics from the interpreter.
The one non-standard feature of our setup is that the interpreter must
be capable of interpreting the programs derived from \emph{any} nonterminal in the grammar.





\paragraph{\semgus benchmarks}
Our first source of benchmarks is the \semgus benchmark repository~\cite{kim2021semantics}.
This dataset contains \semgus synthesis problems where each problem consists of a grammar of terms, a set of CHCs inductively defining the semantics of terms in the grammar, and a specification that the synthesized program should meet.
For our purposes, we ignored the specification and collected the grammar plus semantics for \jiangyichanged{11} distinct languages that appear in the repository.
We do not consider languages that contain abstract data types (e.g., stacks) or require a large range of inputs (e.g., ASCII characters) due to their poor support by the \sygus solver.
These languages gave us \jiangyichanged{11} benchmarks.

Some of the languages used in the \semgus benchmark set are parametric (denoted by a parameter $k$), meaning that the semantics is slightly different based on a given parameter (e.g., number of program variables
for IMP and length of the input string for regular expressions).
For these benchmarks, we ran \toolname on an increasing sequence of parameter values and
reported the largest parameter value for which \toolname succeeds.

\textsc{RegEx}$(k)$ is a language for matching regular expressions on strings of length $k$;
Given a regular expression $r$ and string $s$ of length $k$ (index starts from 0), the semantic functions produce
a Boolean matrix $M \in \mathtt{Bool}^{(k+1) \times (k+1)}$ such that $M_{i,j} = \textit{true}$ iff the substring $s_{i \dots j-1}$ matches regular expression $r$---here $s_{i \dots i}$ denotes the empty string, and by definition, $M_{i,j} = \textit{false}$ for $i \ge j$.


$\textsc{Cnf}(k)$, $\textsc{Dnf}(k)$ and $\textsc{Cube}(k)$
are languages of Boolean formulas (of the syntactic kind indicated by their names, i.e., conjunctive normal form, disjunctive normal form, and cubes) involving up to $k$ variables. 

$\textsc{Imp}$ is an imperative language that contains common control flow structures, such as conditionals and while loops, for programs with $k$ integer variables.
Note that \textsc{Imp} includes operators such as $\tok{while}$ and $\tok{do\_while}$ for which
the semantics involves semantically recursive productions (\Cref{sec:alg:sem-recursive}).
%
The complete semantics of $\textsc{Imp}$ can be found in the supplementary material.
Two versions of \textsc{Imp} are used in our benchmarks.
The first version is called $\textsc{Imp}(k)$, where we explicitly record the states of $k$ variables as $k$ arguments of semantic functions.
\toolname could synthesize its semantics up to $k=2$.
We also present another version of this language called \textsc{ImpArr} where an arbitrary number of variables can be used.
%
In $\ImpArr$, variables are named $\mathit{var}_0, \mathit{var}_1, \dots$ where the subscript is any natural number.
We use the theory of arrays to store the variable states into an array, passing the array as an argument to the semantic function. The array is indexed by variable id.
When we present results later in the section, the results
for both languages (i.e., $\textsc{Imp}(2)$ and $\ImpArr$) are shown for comparison.
(For $\textsc{Imp}(2)$, the goal is to synthesize a semantics that works on states with exactly 2 variables;
for $\ImpArr$, the goal is to synthesize a semantics that works for states with any number of variables.)

$\textsc{IntArith}$ is a benchmark about basic integer calculations, like addition, multiplication, and conditional selection. It also includes three constants whose value can be specified in the input to the semantic relations.

$\textsc{BvSimple}(k)$ describes bit-vector operations involving $k$ bit-vector constants.
$\textsc{BvSimpleImp}(m, n)$ is essentially a variant of $\textsc{BvSimple}(k)$ that augments the language with let-expressions.
Parameters $m$ and $n$ mean that the language can use up to
$m$ bit-vector constants and $n$ bit-vector variables.
$\textsc{BvSaturate}(k)$ and $\textsc{BvSaturateImp}(k)$ use the same syntaxes as $\textsc{BvSimple}(k)$ and $\textsc{BvSimpleImp}(k)$, respectively, but operations use a saturating semantics that never overflows or underflows.

\paragraph{Attribute-grammar synthesis~\cite{panini-paper}}
Our second source of benchmarks is from the Panini tool for synthesizing attribute grammars~\cite{panini-paper}.
An attribute grammar (AG) associates each nonterminal of an underlying context-free grammar with some number of \emph{attributes}.
Each production has a set of attribute-definition rules (sometimes called \emph{semantic actions}) that specify how the value of one attribute of the production is set as a function of the values of other attributes of the production.
In a given derivation tree of the AG, each node has an associated set of \emph{attribute instances}.
The attribute-definition rules are used to obtain a consistent assignment of values to the tree's attribute instances:
each attribute instance has a value equal to its defining function applied to the appropriate (neighboring) attribute instances of the tree.
Effectively, AGs assign a semantics to programs via attributes, and the underlying attribute-definition rules can be captured via CHCs.
While there are AG extensions to handle circular AGs \cite{DBLP:journals/toplas/Jones90,DBLP:journals/entcs/MagnussonH03}---i.e., AGs in which some derivation trees have attribute instances that are defined in terms of themselves---the work of \citeauthor{panini-paper} concerns non-circular AGs.

\citet{panini-paper} present 12 benchmarks.
We ignored 4 benchmarks that are either
\rone not publicly accessible, or
\rtwo use semantic functions that cannot be expressed in SMT-LIB and are thus beyond what can be synthesized using a \sygus solver---e.g., complex data structures, or
\rthree identical to existing benchmarks from other sources.
We did not run their tool on our benchmarks because our problem is more general than theirs,
supporting a wider range of language semantics:
the scope of our work includes recursive semantics, which can be handled only indirectly in a system such as theirs (which supports only non-circular AGs)---i.e., by introducing powerful hard-to-synthesize recursive functions that effectively capture an entire construct's semantics.
The running time is also not directly comparable, because \citeauthor{panini-paper}'s approach uses user-provided sketches (i.e., partial solutions to each semantic action), which simplifies the synthesis problem. 
In contrast, in our work we do not assume that a sketch is provided for the semantic constraints and instead consider general \sygus grammars.

The remaining 8 benchmarks of \citeauthor{panini-paper} are consolidated as 4 languages (i.e., giving us four benchmarks).
\textsc{IteExpr} is a language of basic integer operations, comparison
expressions, and ternary if-then-else expressions (not statements).
Our \textsc{IteExpr} benchmark subsumes benchmarks
B3, B4, and B5 of \citeauthor{panini-paper} because their only differences stem from whether the expression is written in prefix, postfix, or infix notation.
For \toolname, such surface-syntax differences are unimportant because \toolname uses regular tree grammars to express a language's abstract syntax, and the underlying abstract syntax of prefix, postfix, and infix expressions is the same.
\textsc{BinOp} is a language of binary strings (combined from benchmarks B1 and B2 of \citeauthor{panini-paper}), along with built-in functions for popcount (counting the number of ones) and binary-to-decimal conversion.
\textsc{Currency} is a language for currency exchange and calculation. 
\textsc{Diff} is a language for computing finite differences.
Because the original benchmark from \citeauthor{panini-paper} involves differentiation and real numbers (which are not supported by existing \sygus solvers), we modified the benchmark to perform the related operation of finite differencing over integer-valued functions.
Specifically, for a function $f$, its finite difference is defined as \(\Delta f = f(x+1) - f(x)\).
Starting from here, finite differences for sums and products can be obtained compositionally, e.g., $\Delta\left(u\cdot v\right) = u(x) \Delta v(x) + v(x+1) \Delta u(x)$.

\paragraph{\sygus grammars}
For each semantic function, we also provided a grammar for the \sygus solver, which contains the operators of the underlying logical theory and any specific functions that must appear in the target semantics. 

For instance, for all benchmarks using the logic fragment \texttt{NIA}, we allow the use of basic integer operations and integer constants, along with language-specific operations like conditional operators (if-then-else).

For the languages \textsc{Diff} and \textsc{Currency} we did not include conditional operators, because they do not appear in the semantics.

For \textsc{BVSaturated} and \textsc{BVIMPSaturated} we provided operators for detecting overflow and underflow.

Lastly, for languages known to be free of side effects, we modified the \sygus grammars to forbid data flow between siblings, and only allow parent-to-child and child-to-parent assignments.




\subsection{RQ1: Can \toolname Synthesize the Semantics of Non-trivial Languages?}

%

%
%

\input{tables/eval-opt-short}

Table~\ref{tab:eval-short} presents a highlight of the results of running \toolname on each benchmark (column 1) for each production rule (column 2).
For the parametric languages, we ran each benchmark up to the largest parameter $k$ for which the solver timed out and reported the running time and other metrics for the largest such $k$ (more details below).
The third column provides the median number of CEGIS iterations taken to synthesize each production, and the fourth column provides the median number of $\langle \mathit{in}, \mathit{term}, \mathit{out} \rangle$ counterexamples found for one production rule.
We take the median of total execution time on one production rule and list it in column 7. Columns 5--6 are breakdowns of the total time into time for \sygus solving and time for SMT solving.
To summarize, \toolname could synthesize complete semantics for $12/15 \approx 80\%$ of benchmark languages (two languages exist for the $\textsc{Imp}$ benchmark, see below).

For $\textsc{RegEx}(k)$ ($k=2,\dots,8$)
, \toolname could synthesize a semantics for up to $k=2$.
For $\textsc{Cnf}(k)$ ($k=4,\dots,8$), $\textsc{Dnf}(k)$ ($k=4,\dots,8$), and $\textsc{Cube}(k)$ ($k=4,\dots,11$), \toolname could synthesize semantics for all parameters included in the \semgus benchmarks.
For the bit vector benchmarks, \toolname could synthesize a semantics for $\textsc{BVSimple}(k)$ up to $k=3$, and a semantics for $\textsc{BVIMPSimple}(m, n)$ ($(m,n) \in \{(1,2), (3,3)\}$) up to $m=1$ and $n=2$. 

%

For all these parametric cases that timeout, the number of input and output variables in semantic functions is large: 10 inputs and 10 outputs for $\textsc{RegEx}(3)$.

Additionally, \toolname timed out for the benchmarks \textsc{Diff}, \textsc{BVSaturated}, and \textsc{BVIMPSaturated}.\footnote{
  Data for some languages are only listed in the supplementary material.
}
For \textsc{Diff}, 4 of the 7 runs resulted in a timeout, so \textsc{Diff} is reported as a timeout ({even though at least one run could synthesize the semantics of all the productions}).
For the 4 runs that timed out, \toolname can solve the semantics of 5 of the 6 productions in the grammar.
%
\toolname could synthesize the semantics of 9/18 productions for \textsc{BVIMPSaturated}, and 10/17 productions for \textsc{BVSaturated} in at least one run.

In benchmarks that timed out, the time-out happened during a call to the \sygus solver---i.e., the functions to be synthesized were too complex (more details in \Cref{sec:rq2}).

\paragraph{Finding:} To answer RQ1, \toolname can synthesize semantics for many non-trivial languages as long as the semantics does not involve very large functions (more than 20 terms). 

\subsection{RQ2: Where is Time Spent during Synthesis?}
\label{sec:rq2}

\paragraph{\sygus vs SMT Time}
Appendix B also presents the breakdown of how much time the solver spends solving \sygus problems (to find candidate functions) and calling SMT solvers (to compute complete summaries).
Among all the benchmarks, a median of {$16.24\%$} of the total solving time is spent on \sygus problems, and a median of {$19.90\%$} of the time is spent solving SMT queries. However, for the slowest $10\%$ production rules ({>\SI{32.17}{s}}), the median of \sygus solving time grows to {$64.91\%$}, which indicates that \sygus contributes to most of the execution time on slow-running cases.
%
%

{
Among all benchmarks, $90\%$ of the per-production semantics are solved within \SI{32.17}{s}.}
The 12 rules that take longer than \SI{32.17}{s} to be synthesized are all non-leaf rules and their partial semantic constraints fall into the following three categories: \rone 5 of them contain large integers or complex SMT primitives (e.g., 32-bit integer division, theory of arrays); \rtwo 3 of them involve large logical formulas with sizes ranging between 8 and 24 subterms, e.g., formulas representing $3 \times 3$ matrix multiplication or other matrix operations; \rthree 4 of them contain {multiple} input and output parameters of semantic functions that correspond to variable states, e.g., $\tok{while}$ and $\tok{do\_while}$.
In particular, \toolname takes \SI{1374.13}{s} to synthesize the CHC for $\tok{do\_while}$ in $\textsc{Imp}(2)$ because there can be many possible ways to modify the data flow between the production's child terms; this aspect occurs in many CEGIS iterations.
In all of the above cases, as expected from known limitations of {CVC5}, the \sygus and SMT solvers account for most of the execution time---45.63\% and 27.98\% of the total running time is spent calling the \sygus and SMT solvers, respectively.

\paragraph{Relation to CEGIS Iterations and Size of Solutions}
\Cref{tab:eval-short} hints that the cost of synthesizing a semantics may be proportional to the number of CEGIS iterations, which in general is a good indicator of the complexity of a formula (and of how expressive the underlying \sygus grammar is). Additionally, the cost should also be proportional to the size of synthesized parts in the \sygus problems, which directly indicates formula complexity. We plotted \Cref{fig:rel-time} to better understand those relations by using the data from some slowest benchmarks.

\Cref{fig:rel-time-full-size} shows the relationship between the time for synthesizing a per-production rule semantics and the size of the final semantics.
For the same language, the time grows exponentially with the increase in the size of the final solution. \Cref{fig:rel-time-partial-size} shows that the time also grows exponentially with the increase in solution size for per-output partial semantic constraint.

\begin{figure}[t!]
  \begin{subfigure}{0.5\textwidth}
    \centering
    \includegraphics[width=\linewidth]{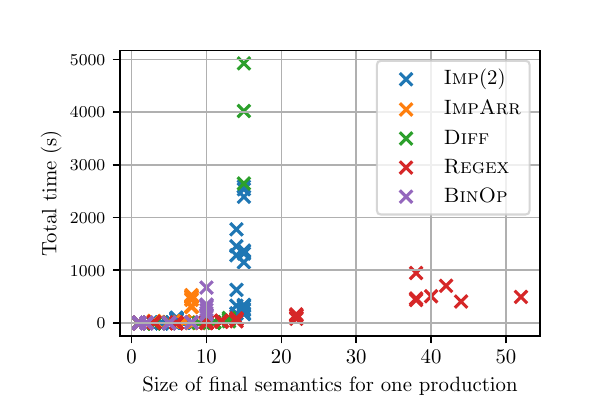}
    \caption{Time vs. semantic constraint size}
    \label{fig:rel-time-full-size}
  \end{subfigure}%
  \begin{subfigure}{0.5\textwidth}
    \centering
    \includegraphics[width=\linewidth]{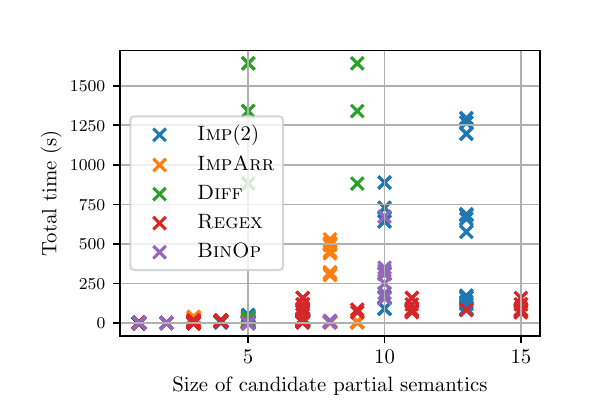}
    \caption{Time vs. partial semantic constraint size}
    \label{fig:rel-time-partial-size}
  \end{subfigure}
  \caption{Plots relating the time to synthesize the semantics of one production rule vs final semantic constraint solution size (a) and partial semantic constraint solution size (b). We only included selected slowest benchmarks due to graph size limit.}
  \label{fig:rel-time}
\end{figure}

Because the performance varies greatly
across different benchmarks, to better understand the impact of CEGIS iterations, we focus our attention on one difficult benchmark, $\textsc{Imp}(2)$.
Specifically, we analyze the time taken to synthesize the semantic rule for $\tok{do\_while}$, which was one of the hardest productions in our benchmark set (2,500s). 
Figure~\ref{fig:imp2-do-while} provides a stack plot detailing the running time for all 16 CEGIS iterations needed to synthesize $\tok{do\_while}$.
As expected, as more examples are accumulated by CEGIS iterations, the \sygus  solver requires more time.
The execution time for different parts is plotted by the areas of different colors. We can conclude that for the rule of $\tok{do\_while}$, \sygus solver takes $64.3\%$ of the execution time.

\paragraph{Simplifying synthesis with appropriate SMT theory.}
To our surprise, the language $\ImpArr$, which uses the theory of array to model an arbitrary number of variables, takes 1041.2s on average, being nearly twice as fast as $\textsc{Imp}(2)$. 
%
To understand why this is the case, note the difference in their semantic signatures
(\Cref{fig:imp-sigs}): the signature of $\textsc{Imp}(2)$'s semantics contains 3 input arguments and 2 output arguments. However, the signature of $\ImpArr$'s semantics contains only 2 input arguments and 1 output argument, packing program states into a single array rather than $k$ arguments (see \Cref{fig:imp-multi} for some examples). By choosing an 
appropriate theory, the signature of the semantics can be simplified, thus shrinking the solution space for synthesis.


\begin{figure}
    \centering
    \begin{align*}
        \sem{\cdot}_{\mathit{Sem.S}}^{\textsc{Imp}(2)}&: \mcL(S) \times \mathbb{Z} \times \mathbb{Z} \times \mathbb{Z} \times \mathbb{Z}\\
        \sem{\cdot}_{\mathit{Sem.S}}^{\ImpArr}&: \mcL(S) \times \mathcal{A}_{\mathbb{Z}}^{\mathbb{Z}} \times \mathcal{A}_{\mathbb{Z}}^{\mathbb{Z}}
    \end{align*}
    \caption{Selected differences of semantic signatures between $\textsc{Imp}(2)$ and $\ImpArr$. $\mathcal{A}_{\mathbb{Z}}^{\mathbb{Z}}$ stands for an SMT array mapping integers to integers.}
    \label{fig:imp-sigs}
\end{figure}

\begin{figure}
    \centering
    \begin{mathpar}     
        \inferrule{\semanticBracketsUsing{b}{\sigma} = v \\ \semanticBracketsUsing{s_1}{\sigma} = \sigma' \\ \semanticBracketsUsing{s_2}{\sigma} = \sigma''}{\semanticBracketsUsing{\tok{if}~b~\tok{then}~s_1~\tok{else}~s_2}{\sigma} = v \mathbin{?} \sigma' \mathbin{:} \sigma''}~\textsc{Ite} \and
        
        \inferrule{\semanticBracketsUsing{b}{\sigma} = \top \\ \semanticBracketsUsing{s}{\sigma} = \sigma' \\ \semanticBracketsUsing{\tok{while}~b~\tok{do}~s}{\sigma'} = \sigma''}{\semanticBracketsUsing{\tok{while}~b~\tok{do}~s}{\sigma} = \sigma''}~\textsc{WhileLoop}
        \and
        
        \inferrule{\semanticBracketsUsing{b}{\sigma} = \bot}{\semanticBracketsUsing{\tok{while}~b~\tok{do}~s}{\sigma} = (\sigma)}~\textsc{WhileEnd} \and
    \end{mathpar}
    \caption{Selected semantic rules for $\ImpArr$. $\sigma, \sigma', \sigma''$ are arrays.}
    \label{fig:imp-multi}
\end{figure}

\begin{figure}
    \centering
    \begin{minipage}{0.49\textwidth}
        \centering
        \resizebox{\linewidth}{!}{\includegraphics{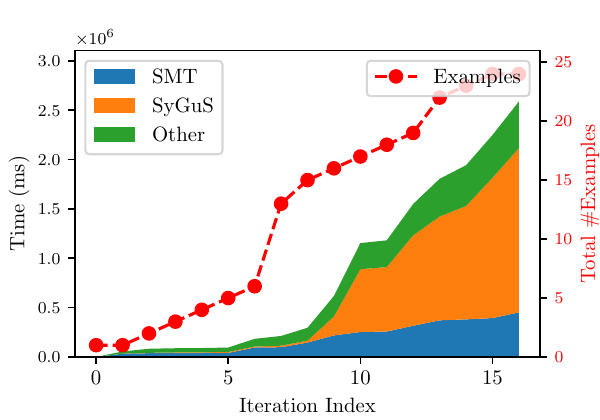}}
        \caption{Execution time per iteration for \tok{do\_while} in $\textsc{Imp}(2)$}
        \label{fig:imp2-do-while}
    \end{minipage}\hfill
    \begin{minipage}{0.5\textwidth}
        \centering
        \resizebox{\linewidth}{!}{\includegraphics{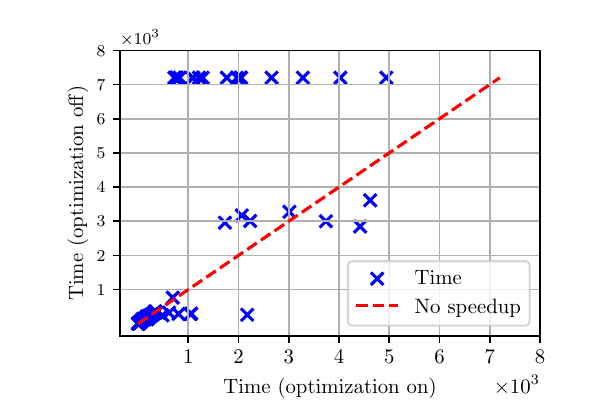}}
        \caption{Speedup provided by optimization}
        \label{fig:speedup}
    \end{minipage}
\end{figure}

\paragraph{Finding:} {To answer RQ2, \toolname spends most of the time ($71.78\%$) solving \emph{\sygus problems}, and the time is affected by the size of the candidate semantic function.}

\subsection{RQ3: Is the Multi-output Optimization from \Cref{sec:implementation:optimization} Effective?}


\Cref{fig:speedup} compares the running time of \toolname with and without the multi-output optimization  (\Cref{sec:implementation:optimization}) on all the runs of our tools for the 7 different random seeds.

With the optimization turned off, \toolname timed out on 10 more runs (specifically all the 7 runs for \textsc{RegEx} and 3 more runs for \textsc{Diff}).
All the benchmarks for which disabling the optimization caused a timeout have 3 or more output variables. 
Comparing \Cref{fig:rel-time-full-size} and \Cref{fig:rel-time-partial-size} shows how the semantic functions used in the \textsc{RegEx} benchmarks are very large (up to size 50), but thanks to the optimization, our algorithm only has to solve \sygus problems on formulas of size at most 15.

On the runs that terminated both with and without the optimization, the non-optimized algorithm is on average $8\%$ faster---i.e., the two versions of the algorithm have comparable performance.
However, for 15/98 runs the optimization results in a 20\% or more slowdown. 
When inspecting these instances, we observed that the multi-output optimization spent many iterations synchronizing the many possible data flows for productions where the final term was actually small but many variables were involved---e.g., sequential composition in $\textsc{Imp}(2)$.

\paragraph{Finding:} The multi-output optimization from \Cref{sec:implementation:optimization} is effective for languages with 3 or more output variables in their semantics.

\subsection{RQ4: How do Synthesized Semantics Compare to Manually Written Ones?}

The synthesized semantics for almost all of our benchmarks are
either identical to the original manually constructed one, or each CHC in the synthesized semantics is logically equivalent to the CHC of the original semantics.

\begin{figure}
    \centering
    \begin{subfigure}{\textwidth}
      \centering
      \begin{mathpar}
        \inferrule{\semanticBracketsUsing{e_1}{s} = {\begin{psmallmatrix}
        M_{0, 0} & M_{0, 1} \\ & M_{1, 1}
    \end{psmallmatrix}} \\ \semanticBracketsUsing{e_2}{s} = {\begin{psmallmatrix}
        M_{0, 0}' & M_{0, 1}' \\ & M_{1, 1}'
    \end{psmallmatrix}}}{\semanticBracketsUsing{e_1\mathbin{\tok{\cdot}}e_2}{s} =
                {\begin{psmallmatrix}
                 (M_{0,0} \land M'_{0,0}) & (M_{0,0} \land M'_{0,1}) \lor (M_{0,1} \land M'_{1,1}) & (M_{0,0} \land M'_{0,2}) \lor (M_{0,1} \land M'_{1,2}) \lor (M_{0,2} \land M'_{2,2})\\
                  & (M_{1,1} \land M'_{1,1}) & (M_{1,1} \land M'_{1,2}) \lor (M_{1,2} \land M'_{2,2})\\
                  & & (M_{2,2} \land M'_{2,2})
        \end{psmallmatrix}}}~\textsc{{ConcatM}} 
      \end{mathpar}
      \caption{Manually written semantics}
      \label{fig:diff-concat:manual}
    \end{subfigure}
    
    \begin{subfigure}{\textwidth}
      \centering
      \begin{mathpar}
        \inferrule{\semanticBracketsUsing{e_1}{s} = {\begin{psmallmatrix}
        M_{0, 0} & M_{0, 1} \\ & M_{1, 1}
    \end{psmallmatrix}} \\ \semanticBracketsUsing{e_2}{s} = {\begin{psmallmatrix}
        M_{0, 0}' & M_{0, 1}' \\ & M_{1, 1}'
    \end{psmallmatrix}}}{\semanticBracketsUsing{e_1\mathbin{\tok{\cdot}}e_2}{s} =
                {\begin{psmallmatrix}
                 M_{2,2} \land M'_{1,1}
                 & (M_{0,0} \land M'_{0,1}) \lor (M'_{0,0} \land M_{0,1})
                 & (M_{2,2} \land M'_{0,2}) \lor (M_{0,2} \land M'_{0,0}) \lor (M_{0,1} \land M'_{1,2})
                 \\
                 & (M_{0,0} \lor M_{2,2}) \land M'_{2,2}
                 & (M_{1,1}' \land M_{1,2}) \lor (M_{1,1} \land M_{1,2}')
                 \\
                 & & 
                 (M_{1,1} \land M'_{0,0})
        \end{psmallmatrix}}}~\textsc{ConcatS} 
      \end{mathpar}
      \caption{Synthesized Semantics}
      \label{fig:diff-concat:synthesized}
    \end{subfigure}
    \caption{Manually-written and synthesized semantics for \textsc{Concat} in \textsc{RegEx}$(2)$}
    \label{fig:diff-concat}
\end{figure}

The one exception is the semantics synthesized for the language of $\textsc{RegEx}(2)$, for which the individual CHCs for \textsc{Or}, \textsc{Concat}, \textsc{Neg}, and \textsc{Star} are not logically equivalent to the manually-written ones.
For instance, consider the Concat rule for the semantics of concatenation.
For this construct, the manually written CHC is shown in \Cref{fig:diff-concat:manual}, whereas \toolname synthesizes the CHC shown in \Cref{fig:diff-concat:synthesized}.
The two CHCs are not logically equivalent.
For example, if the children evaluate to the matrices 
$M=\begin{psmallmatrix}
    \top & \bot & \bot \\
         & \bot & \bot \\
         &      & \top
\end{psmallmatrix}$
and
$M'=\begin{psmallmatrix}
    \top & \bot & \bot \\
         & \bot & \bot \\
         &      & \top
\end{psmallmatrix}$,
the outputs computed by the manually written CHC 
and the synthesized CHC are
$M_{\textit{man}}=\begin{psmallmatrix}
    \top & \bot & \bot \\
         & \bot & \bot \\
         &      & \top \\
\end{psmallmatrix}$,
and
$M_{\textit{syn}}=\begin{psmallmatrix}
    \bot & \bot & \bot \\
         & \top & \bot \\
         &      & \bot \\
\end{psmallmatrix}$, respectively, which have different values on the diagonal.

When inspecting the two rules, we realized that the example matrices $M$ and $M'$ shown above cannot actually be produced by the semantic rules for regular expressions.
In particular, the examples require different Boolean values to appear on the diagonal of one $3 \times 3$ matrix.
However, all the elements on the diagonal represent the semantics of the regular expression on the empty string, so they must all have the same value!
We note that this inconsistency in the semantics can also be observed without a reference semantics to compare against 
because different runs of the algorithm could return logically inequivalent CHCs---in fact, such inequivalence was how we initially discovered the inconsistency.

\toolname helped us discover an inefficiency in the semantics that was being used in the standard regular expressions benchmarks in the \semgus repository.
We thus modified the interpreter so that for the example above it only produces a $2 \times 2$ matrix $M = \begin{psmallmatrix}
    M_{0,1} & M_{0, 2} \\ & M_{1, 2}
\end{psmallmatrix}$ 
(corresponding to the non-empty substrings of the input string) and a \textit{single} variable $M_\epsilon$ to denote whether the regular expression should accept the empty string (instead of the previous multiple copies of logically equivalent variables).
This semantics reduces the total number of variables in the semantic domain from 6 to 4 in this example.

We call this new semantics \textsc{RegExSimp} (see \Cref{fig:regex-simp-concat} for an example).
After modifying the interpreter to produce this new semantics, \toolname synthesized the corresponding CHCs in a median of \SI{1968}{s}.

\begin{figure}
    \centering
    \begin{mathpar}
    \inferrule{\semanticBracketsUsing{e_1}{s} = \left(M_\epsilon, {\begin{psmallmatrix}
        M_{0, 0} & M_{0, 1} \\ & M_{1, 1}
    \end{psmallmatrix}}\right) \\ \semanticBracketsUsing{e_2}{s} = \left(M_\epsilon', {\begin{psmallmatrix}
        M_{0, 0}' & M_{0, 1}' \\ & M_{1, 1}'
    \end{psmallmatrix}}\right)}{\semanticBracketsUsing{e_1\mathbin{\tok{\cdot}}e_2}{s} = \left(
            M_\epsilon \land M_\epsilon',
            {\begin{psmallmatrix}
            (M_\epsilon \land M_{0, 0}') \lor (M_{0, 0} \land M_\epsilon') &
            (M_\epsilon \land M_{0, 1}') \lor (M_{0, 0} \land M_{1, 1}') \lor (M_{0, 1} \land M_\epsilon')\\
            & (M_\epsilon \land M_{1, 1}') \lor (M_{1, 1} \land M_\epsilon')
            \end{psmallmatrix}}\right)}~\textsc{Concat} 
    \end{mathpar}
    \caption{New semantics for \textsc{Concat} in \textsc{RegExSimp}.}
    \label{fig:regex-simp-concat}
\end{figure}





%
To check whether the semantics \textsc{RegExSimp} is indeed more efficient than the original semantics \textsc{RegEx}, we modified all the 28 regular-expression synthesis benchmarks appearing in the \semgus benchmark set.
Each of these benchmarks requires one to find a regular expression that accepts some examples and rejects others.

We then used the \keithtoolname enumeration-based synthesizer to try to solve all the benchmarks with either of the two semantics.
Because \keithtoolname enumerates programs of increasing size and uses the semantics to execute them and discard invalid program candidates, we conjectured that executing programs faster allows \keithtoolname to explore the search space faster.

\keithtoolname was faster at solving synthesis problems with the \textsc{RegExSimp} semantics than with the \textsc{RegEx} ones (although both solved the same set of benchmarks).
Although the speedup over all benchmarks is only 1.1x, the new semantics \textsc{RegExSimp} \emph{was particularly beneficial} for the harder synthesis problems.
When considering the 13 benchmarks for which synthesis using the \textsc{RegEx} took longer than one second, the speedup increased to 1.18x.

\paragraph{Finding:}
\toolname synthesized semantics that were identical to the manually written ones for 13/14 benchmarks. When \toolname found a logically inequivalent semantics, it unveiled a performance bug.

%% file: tables/eval-opt-short.tex
\afterpage{
\begingroup
\renewcommand\arraystretch{0.6}
\footnotesize
\begin{longtable}{
    c
    l
    r
    r
    r
    r
    r
}
\caption{Detailed results for selected benchmarks. See supplementary material for the full list of results.\label{tab:eval-short}\protect}\\
\toprule
{Lang.} & {Rule}%
    & {\# Iter.} & {\# Ex} &  {SyGuS (s)} & {SMT (s)} & {Total (s)}           \\
\midrule
\multirow{18}{*}{\rotatebox[origin=c]{90}{$\ImpArr$}}&$E\to$\,$\tok{0}$&1&1&\SI{0.01}{}&\SI{0.01}{}&\SI{0.04}{}\\
&$E\to$\,$\tok{1}$&1&1&\SI{0.01}{}&\SI{0.01}{}&\SI{0.04}{}\\
&$B\to$\,$\tok{f}$&1&1&\SI{0.01}{}&\SI{0.01}{}&\SI{0.05}{}\\
&$B\to$\,$\tok{t}$&1&1&\SI{0.01}{}&\SI{0.01}{}&\SI{0.1}{}\\
&$S\to$\,$\tok{\mathtt{dec\_var}}_i$&3&2&\SI{4.29}{}&\SI{0.56}{}&\SI{5.36}{}\\
&$S\to$\,$\tok{\mathtt{inc\_var}}_i$&3&2&\SI{3.56}{}&\SI{0.54}{}&\SI{4.51}{}\\
&$B\to$\,$\lnot B$&3&2&\SI{0.01}{}&\SI{1.42}{}&\SI{5.53}{}\\
&$E\to$\,$\tok{\mathtt{var}}_i$&3&2&\SI{0.01}{}&\SI{0.28}{}&\SI{0.66}{}\\
&$E\to$\,$E\mathbin{\tok{+}}E$&3&2&\SI{0.02}{}&\SI{6.51}{}&\SI{12.38}{}\\
&$E\to$\,$E\mathbin{\tok{-}}E$&3&2&\SI{0.01}{}&\SI{6.43}{}&\SI{12.13}{}\\
&$B\to$\,$E\mathbin{\tok{<}}E$&4&3&\SI{0.01}{}&\SI{3.38}{}&\SI{10.33}{}\\
&$B\to$\,$B \land B$&4&3&\SI{0.06}{}&\SI{2.36}{}&\SI{6.23}{}\\
&$S\to$\,$\tok{\mathtt{var}}_i \coloneqq E$&2&1&\SI{0.03}{}&\SI{8.11}{}&\SI{11.6}{}\\
&$B\to$\,$B \lor B$&5&4&\SI{0.03}{}&\SI{2.42}{}&\SI{6.14}{}\\
&$S\to$\,$S~\tok{;}~S$&3&1&\SI{0.02}{}&\SI{13.88}{}&\SI{25.91}{}\\
&$S\to$\,$\tok{do\_while}~S~B$&5&2&\SI{0.22}{}&\SI{342.25}{}&\SI{499.11}{}\\
&$S\to$\,$\tok{while}~B~S$&4&2&\SI{0.1}{}&\SI{218.66}{}&\SI{321.11}{}\\
&$S\to$\,$\tok{ite}~B~S~S$&4&2&\SI{0.03}{}&\SI{7.08}{}&\SI{27.82}{}\\
\midrule
\multirow{22}{*}{\rotatebox[origin=c]{90}{$\textsc{IMP}(2)$}}&$E\to$\,$\tok{0}$&1&1&\SI{0.01}{}&\SI{0.01}{}&\SI{0.05}{}\\
&$E\to$\,$\tok{1}$&1&1&\SI{0.01}{}&\SI{0.01}{}&\SI{0.04}{}\\
&$S\to$\,$\tok{x--}$&2&2&\SI{0.06}{}&\SI{0.02}{}&\SI{0.11}{}\\
&$S\to$\,$\tok{y--}$&2&2&\SI{0.11}{}&\SI{0.03}{}&\SI{0.17}{}\\
&$B\to$\,$\tok{f}$&1&1&\SI{0.01}{}&\SI{0.01}{}&\SI{0.06}{}\\
&$S\to$\,$\tok{x++}$&2&2&\SI{0.04}{}&\SI{0.03}{}&\SI{0.11}{}\\
&$S\to$\,$\tok{y++}$&2&2&\SI{0.12}{}&\SI{0.02}{}&\SI{0.16}{}\\
&$B\to$\,$\tok{t}$&1&1&\SI{0.01}{}&\SI{0.02}{}&\SI{0.13}{}\\
&$E\to$\,$\tok{x}$&2&2&\SI{0.01}{}&\SI{0.01}{}&\SI{0.04}{}\\
&$E\to$\,$\tok{y}$&1&1&\SI{0.01}{}&\SI{0.01}{}&\SI{0.04}{}\\
&$S\to$\,$\tok{x}~\tok{\coloneqq}~E$&2&2&\SI{0.1}{}&\SI{3.23}{}&\SI{6.17}{}\\
&$S\to$\,$\tok{y}~\tok{\coloneqq}~E$&2&2&\SI{0.04}{}&\SI{3.22}{}&\SI{6.19}{}\\
&$B\to$\,$\lnot B$&3&3&\SI{0.02}{}&\SI{2.49}{}&\SI{5.26}{}\\
&$E\to$\,$E\mathbin{\tok{+}}E$&4&3&\SI{0.05}{}&\SI{8.52}{}&\SI{14.83}{}\\
&$E\to$\,$E\mathbin{\tok{-}}E$&5&2&\SI{0.13}{}&\SI{8.03}{}&\SI{13.83}{}\\
&$B\to$\,$E\mathbin{\tok{<}}E$&8&5&\SI{0.08}{}&\SI{7.5}{}&\SI{13.66}{}\\
&$B\to$\,$B \land B$&4&4&\SI{0.03}{}&\SI{5.33}{}&\SI{11.71}{}\\
&$B\to$\,$B \lor B$&4&4&\SI{0.05}{}&\SI{4.61}{}&\SI{8.99}{}\\
&$S\to$\,$S~\tok{;}~S$&5&3&\SI{4.55}{}&\SI{15.0}{}&\SI{72.53}{}\\
&$S\to$\,$\tok{do\_while}~S~B$&27&35&\SI{858.5}{}&\SI{257.33}{}&\SI{1374.13}{}\\
&$S\to$\,$\tok{while}~B~S$&9&7&\SI{16.88}{}&\SI{122.41}{}&\SI{266.8}{}\\
&$S\to$\,$\tok{ite}~B~S~S$&11&5&\SI{525.28}{}&\SI{33.88}{}&\SI{628.71}{}\\
\midrule
\multirow{9}{*}{\rotatebox[origin=c]{90}{\textsc{BinOp}}}&$B\to$\,$\tok{0}$&1&1&\SI{0.01}{}&\SI{0.01}{}&\SI{0.07}{}\\
&$B\to$\,$\tok{1}$&1&1&\SI{0.01}{}&\SI{0.01}{}&\SI{0.22}{}\\
&$B\to$\,$\tok{x}$&2&2&\SI{0.01}{}&\SI{0.01}{}&\SI{0.08}{}\\
&$N\to$\,$\tok{atom}~B$&2&2&\SI{0.09}{}&\SI{0.04}{}&\SI{0.3}{}\\
&$M\to$\,$\tok{atom'}~B$&3&3&\SI{0.07}{}&\SI{0.05}{}&\SI{0.26}{}\\
&$S\to$\,$\tok{bin2dec}~M$&2&2&\SI{0.02}{}&\SI{0.09}{}&\SI{0.3}{}\\
&$S\to$\,$\tok{count}~N$&2&2&\SI{0.04}{}&\SI{0.05}{}&\SI{0.24}{}\\
&$N\to$\,$\tok{concat}~N~B$&5&5&\SI{8.61}{}&\SI{0.22}{}&\SI{10.31}{}\\
&$M\to$\,$\tok{concat'}~M~B$&5&5&\SI{288.81}{}&\SI{0.23}{}&\SI{308.5}{}\\
\midrule
\multirow{10}{*}{\rotatebox[origin=c]{90}{$\textsc{RegEx}(2)$}}&
$Start\to$\,$\tok{eval}~R$&3&3&\SI{0.02}{}&\SI{4.43}{}&\SI{13.4}{}\\
&$R\to$\,$\tok{?}$&3&3&\SI{3.84}{}&\SI{0.07}{}&\SI{4.07}{}\\
&$R\to$\,$\tok{a}$&4&4&\SI{11.1}{}&\SI{0.07}{}&\SI{11.53}{}\\
&$R\to$\,$\tok{b}$&5&5&\SI{11.63}{}&\SI{0.06}{}&\SI{12.01}{}\\
&$R\to$\,$\tok{\epsilon}$&1&1&\SI{0.07}{}&\SI{0.07}{}&\SI{2.38}{}\\
&$R\to$\,$\tok{\emptyset}$&1&1&\SI{0.19}{}&\SI{0.07}{}&\SI{0.46}{}\\
&$R\to$\,$\mathord{\tok{!}}R$&5&5&\SI{2.85}{}&\SI{15.77}{}&\SI{77.36}{}\\
&$R\to$\,$R^{\tok{*}}$&6&6&\SI{0.99}{}&\SI{13.06}{}&\SI{31.91}{}\\
&$R\to$\,$R \mathbin{\tok{\cdot}} R$&24&24&\SI{333.71}{}&\SI{72.58}{}&\SI{495.45}{}\\
&$R\to$\,$R \mathbin{\tok{\mid}} R$&10&10&\SI{10.96}{}&\SI{59.54}{}&\SI{140.82}{}\\
\bottomrule
\end{longtable}
\endgroup
}

%% file: sections/7-related_work.tex
\section{Related Work} \label{sec:related-work}

\textit{Semantics-based Synthesis vs. Library-based Synthesis.}

As discussed throughout the paper, our framework is intended for extracting a formal \semgus semantics that can be used to then take advantage of existing \semgus synthesizers.
One can compare our two-step approach (i.e., first synthesizing the semantics in \semgus format, then using it to synthesize programs) to one-step approaches that only use the given program interpreter in a closed-box fashion to evaluate input-output examples and use them to perform example-based synthesis.
(Such an approach is also used when synthesizing programs that contain calls to closed-box library functions~\cite{library-syntheiss}.)
%
%
On one hand, the closed-box example-based approach is flexible because it can be used for a library/language of any complexity.
On the other hand, our approach allows one to use any program synthesizer, even constraint-based ones~\cite{kim2021semantics}, and to synthesize programs that meet logical (and not just example-based) specifications (e.g., as in \cite{semgus-verifier})
our approach provides an explicit logical representation of the program semantics, whereas example-based approaches are limited to generate-and-test synthesis techniques, such as program enumeration~\cite{semgus-cav24} and example-based specifications.

\textit{Synthesis of Recursive Programs.}
At a high level, the semantics-synthesis problem we consider is similar to a number of works on synthesizing recursively defined programs \cite{partial-bounding-rec,farzan2022recursion,miltner2022bottom,lee2023inductive}. In effect, a semantics for a recursively defined grammar is a recursive program assigning meaning to programs within the language. Both \citet{partial-bounding-rec,farzan2022recursion} use recursion skeletons to reduce their task from synthesizing a recursive program to synthesizing a non-recursive program. Our use of semantic constraints plays a similar role. While both of their techniques assume programs are only structurally recursive (i.e., no recursion on the program term itself), and our framework explicitly allows for program terms that are self-recursive (e.g., $\tok{while}$ loops in \Imp).

Similar to the approach used by \citet{miltner2022bottom} to synthesize simple recursive programs, \algname employs a bottom-up approach to synthesis (i.e., we first synthesize semantics for nullary productions before moving on to other productions). However, unlike \citeauthor{miltner2022bottom}, \algname is well-defined for any ordering of production rules and targets a more complex setting---i.e., synthesizing program semantics. 

Finally, \citet{lee2023inductive} synthesize recursive procedures from examples by first finding likely sub-expressions that can be used to build a complex recursive program and then guessing the recursive structure of the program.
The key difference is that our approach targets a more restricted program, synthesizing program semantics, and therefore already has the recursive structure in hand (thanks to the presence of the grammar in the problems \toolname is solving). 
Because we have limited the synthesis target to an inductively defined program semantics, \toolname can directly focus its effort on
synthesizing the semantic functions of each CHC using well-known synthesis techniques.


\vspace{1mm}\textit{Datalog Synthesis.}
\citet{constraint-based-datalog-synthesis} synthesize Datalog programs (i.e., Horn clauses) with SMT solvers, whereas \citet{sygus-datalog} use a syntax-guided approach. In our work, we use constrained Horn clauses, which are strictly more expressive than Datalog programs, to denote semantics. Aside from the fact that the Datalog-synthesis problem considers different inputs (i.e., the data), a CHC also contains a function in a theory $\mathcal{T}$ (such as \texttt{LIA} or \texttt{BV}), which \toolname has to synthesize.

\vspace{1mm}\textit{Synthesizing Attribute Grammars.}
\citet{panini-paper} proposed a sketch-based method for synthesizing attribute grammars. When provided with a context-free grammar, their tool can automatically create appropriate \emph{semantic actions} from sketches of attribute grammars. Instead of semantic actions, in our work we use CHCs to express program semantics. Our approach can model recursive semantics whereas the technique by \citeauthor{panini-paper} is limited to non-circular attribute grammars. Additionally, while their method requires providing a distinct program sketch (i.e., a partial program) for each production, our approach  only  requires providing a (fairly general) \sygus grammar for each nonterminal in the language. 

%% file: sections/8-conclusion.tex
\section{Conclusion} \label{sec:conclusion}

Writing logical semantics for a language can be a difficult task and our work supplies a method to automatically synthesize a language's semantics from an executable interpreter that is treated as a closed-box. By generating example terms and input-output pairs from the interpreter,
we use a \sygus solver to synthesize semantic rules. Our evaluation shows that the approach applies to a wide range of language features, e.g., recursive semantic functions  with multiple outputs.



As discussed in \Cref{sec:example}, one motivation for this work is to be able to generate automatically the kind of semantics that is needed  to create a program synthesizer using the \semgus framework.
In our algorithm, we harness a \sygus solver to synthesize the constraint in each CHC---i.e., we harness \sygus in service to \semgus---which limits us to synthesizing constraints that are written in theories that \sygus supports.
Going forward, we would like to make use of ``higher-level'' theories, supporting such abstractions as stores or algebraic data types.
As \semgus-based synthesizers and verifiers improve, we might be able to satisfy this wish by using \semgus in service to \semgus!
That is, we could extend \toolname to use \semgus solvers to synthesize semantic constraints.

%% file: sections/artifact-available.tex
\section*{Data-Availability Statement}
The artifact that contains \toolname and all benchmark data is available on Zenodo~\cite{artifact-doi}.

%% file: sections/appendix.tex
\section{Semantics for languages used in benchmark}
\label{appendix:list}
In this section, we present the semantics synthesized by our tool \toolname for any languages referenced within the main text. 
\Cref{appendix:imparr} provides a detailed discussion of how {\textsc{ImpArr}} extends the {$\textsc{Imp}(k)$} benchmarks to support an unbounded number of variables.
\Cref{apendix:semgus-bench} provides the synthesized semantics of \semgus benchmarks and \Cref{apendix:grammar-bench} presents the synthesized semantics of the attribute grammar benchmarks.

\subsection{The {\textsc{ImpArr}} Benchmark with the Theory of Arrays}
\label{appendix:imparr}

Languages like $\textsc{Imp}(k)$ are limited to an \textit{a priori} defined number of variables. For example, $\textsc{Imp}(2)$ only allows two variables to be used in the language, because it \emph{explicitly} passes the values of two variables as \emph{two integer arguments} of semantic functions ($v_0, v_1$ in \Cref{fig:imp-semantics-1,fig:imp-semantics-2}). 
For a language like $\textsc{Imp}$, we would rather have \emph{one} semantics that can handle an unbounded number of variables. 

Instead of using the explicit arguments $v_0, v_1, \dots, v_k$, we can use an array argument to store the variable values. 
The SMT theory of arrays \cite{smtlib2-fmt} provides select-store axioms that conveniently allow us to access the elements of this array. 
The \texttt{(select i a)} axiom extracts the value at index $i$ from an array $a$. The \texttt{(store a i v)} axiom returns a new array that is identical to $a$, except that the value at index $i$ is changed to $v$.
According to these two axioms, the length of the array is not explicitly defined. 
Therefore, arrays are suitable of expressing stores of unbounded size.
Variable evaluation and assignment can be respectively expressed using \texttt{select} and \texttt{store}.
The semantics of \ImpArr~that uses the theory of arrays is given in \Cref{fig:imp-multi-semantics}.

\subsection{\semgus Benchmarks} \label{apendix:semgus-bench}
The \semgus suite of benchmarks consists of a total of 11 languages. We present the synthesized semantics of each as follows:
\begin{enumerate}
\item $\textsc{Cnf}(k)$ is depicted in \Cref{fig:cnf-semantics}.
\item $\textsc{Dnf}(k)$ is depicted in \Cref{fig:dnf-semantics}.
\item $\textsc{Cube}(k)$ is depicted in \Cref{fig:cube-semantics}.
\item $\textsc{IntArith}$ is depicted in \Cref{fig:int-arith-semantics}.
\item $\textsc{RegEx}(2)$ is depicted in \Cref{fig:regex-semantics}.
\item $\textsc{RegExSimp}(2)$ is depicted in \Cref{fig:regex-simp-semantics}.
\item $\textsc{Imp}(2)$ is depicted in \Cref{fig:imp-semantics-1,fig:imp-semantics-2}.
\item $\ImpArr$ is depicted in \Cref{fig:imp-multi-semantics}.
\item $\textsc{BvSimple}(k)$ is depicted in \Cref{fig:bv-simple-semantics}.
\item $\textsc{BvSaturated}(k)$ is depicted in \Cref{fig:bv-saturated-semantics}.
\item $\textsc{BVImpSimple}(m, n)$ is depicted in \Cref{fig:bv-imp-simple-semantics}.
\item $\textsc{BVImpSaturated}(m, n)$ is depicted in \Cref{fig:bv-imp-saturated-semantics}.
\end{enumerate}

\subsection{Attribute-Grammar Synthesis} \label{apendix:grammar-bench}
The suite of attribute-grammar benchmarks from \cite{panini-paper} consists of four languages which we present as follows:
\begin{enumerate}
\item \textsc{BinOp} is presented in \Cref{fig:binop-semantics}.
\item \textsc{Currency} is presented in \Cref{fig:currency-semantics}.
\item \textsc{Diff} is presented in \Cref{fig:diff-semantics}.
\item \textsc{IteExpr} is presented in \Cref{fig:ite-expr-semantics}.
\end{enumerate}

\section{Benchmark Data}
We present the full detailed evaluation results for all languages and productions rules in \Cref{tab:eval}. For each production we present the number of CEGIS iterations, number of generated examples, and execution time \rone to solve \sygus problems, \rtwo to solve SMT queries, and \rthree overall. For each column we repor the number for the median run based on the total run time. See \Cref{sec:evaluation} for full description of experimental setup.

%
%
%

\begin{figure}[h]
    \centering
    \input{semantics/semgus/cnf}
    \caption{Semantics of $\textsc{Cnf}(k)$}
    \label{fig:cnf-semantics}
\end{figure}

\begin{figure}[h]
    \centering
    \input{semantics/semgus/dnf}
    \caption{Semantics of $\textsc{Dnf}(k)$}
    \label{fig:dnf-semantics}
\end{figure}

\begin{figure}[h]
    \centering
    \input{semantics/semgus/cube}
    \caption{Semantics of $\textsc{Cube}(k)$}
    \label{fig:cube-semantics}
\end{figure}

\begin{figure}[h]
    \centering
    \input{semantics/semgus/int-arith}
    \caption{Semantics of \textsc{IntArith}}
    \label{fig:int-arith-semantics}
\end{figure}

\begin{figure}[h]
    \centering
    \input{semantics/semgus/regex}
    \caption{Semantics of \textsc{RegEx}}
    \label{fig:regex-semantics}
\end{figure}

\begin{figure}[h]
    \centering
    \input{semantics/semgus/regex-simp}
    \caption{Semantics of \textsc{RegExSimp}}
    \label{fig:regex-simp-semantics}
\end{figure}

\begin{figure}[h]
    \centering
    \input{semantics/semgus/imp-1}
    \caption{Semantics of $\textsc{Imp}(2)$, part 1}
    \label{fig:imp-semantics-1}
\end{figure}

\begin{figure}[h]
    \centering
    \input{semantics/semgus/imp-2}
    \caption{Semantics of $\textsc{Imp}(2)$, part 2}
    \label{fig:imp-semantics-2}
\end{figure}

\begin{figure}[h]
    \centering
    \input{semantics/semgus/imp-multi}
    \caption{Semantics of $\ImpArr$}
    \label{fig:imp-multi-semantics}
\end{figure}

\begin{figure}[h]
    \centering
    \input{semantics/semgus/bv-simple}
    \caption{Semantics of $\textsc{BvSimple}(k)$}
    \label{fig:bv-simple-semantics}
\end{figure}

\begin{figure}[h]
    \centering
    \input{semantics/semgus/bv-saturated}
    \caption{Semantics of $\textsc{BvSaturated}(k)$}
    \label{fig:bv-saturated-semantics}
\end{figure}

\begin{figure}[h]
    \centering
    \input{semantics/semgus/bv-imp-simple}
    \caption{Semantics of $\textsc{BvImpSimple}(k)$}
    \label{fig:bv-imp-simple-semantics}
\end{figure}

\begin{figure}[h]
    \centering
    \input{semantics/semgus/bv-imp-saturated}
    \caption{Semantics of $\textsc{BvImpSaturated}(k)$}
    \label{fig:bv-imp-saturated-semantics}
\end{figure}

%
%
%

\begin{figure}[hb!]
    \centering
    \input{semantics/panini/binop}
    \caption{Semantics of \textsc{BinOp}}
    \label{fig:binop-semantics}
\end{figure}

\begin{figure}[hb!]
    \centering
    \input{semantics/panini/currency}
    \caption{Semantics of \textsc{Currency}}
    \label{fig:currency-semantics}
\end{figure}

\begin{figure}[hb!]
    \centering
    \input{semantics/panini/diff}
    \caption{Semantics of \textsc{Diff}}
    \label{fig:diff-semantics}
\end{figure}

\begin{figure}[hb!]
    \centering
    \input{semantics/panini/ite-expr}
    \caption{Semantics of \textsc{IteExpr}}
    \label{fig:ite-expr-semantics}
\end{figure}

%
%
%

\input{tables/eval-opt}

%% file: semantics/semgus/cnf.tex
\begin{mathpar}
    \inferrule{i = 0,1,\dots,(k-1)}{\semanticBracketsUsing{\tok{v}_i}{x_0, \dots, x_{k-1}} = x_i}~\textsc{VarAtom}
    \and
    \inferrule{
        \semanticBracketsUsing{v}{x_0, \dots, x_{k-1}} = r
    }{
        \semanticBracketsUsing{\tok{var}~v}{x_0, \dots, x_{k-1}} = r
    }~\textsc{Var}
    \and
    \inferrule{
        \semanticBracketsUsing{v}{x_0, \dots, x_{k-1}} = r
    }{
        \semanticBracketsUsing{\tok{nvar}~v}{x_0, \dots, x_{k-1}} = \lnot r
    }~\textsc{NotVar}
    \and
    \inferrule{
        \semanticBracketsUsing{c}{x_0, \dots, x_{k-1}} = r
    }{\semanticBracketsUsing{\tok{clause}~c}{x_0, \dots, x_{k-1}} = r}~\textsc{Clause}
    \and
    \inferrule{
        \semanticBracketsUsing{c}{x_0, \dots, x_{k-1}} = r_1 \\
        \semanticBracketsUsing{b}{x_0, \dots, x_{k-1}} = r_2
    }{\semanticBracketsUsing{c \mathbin{\tok{\land}} b}{x_0, \dots, x_{k-1}} = r_1 \land r_2}~\textsc{And}
    \and
    \inferrule{
        \semanticBracketsUsing{v}{x_0, \dots, x_{k-1}} = r_1 \\
        \semanticBracketsUsing{c}{x_0, \dots, x_{k-1}} = r_2
    }{\semanticBracketsUsing{v \mathbin{\tok{\lor}} c}{x_0, \dots, x_{k-1}} = r_1 \lor r_2}~\textsc{Or}
\end{mathpar}

%% file: semantics/semgus/dnf.tex
\begin{mathpar}
    \inferrule{i = 0,1,\dots,(k-1)}{\semanticBracketsUsing{\tok{v}_i}{x_0, \dots, x_{k-1}} = x_i}~\textsc{VarAtom}
    \and
    \inferrule{
        \semanticBracketsUsing{v}{x_0, \dots, x_{k-1}} = r
    }{
        \semanticBracketsUsing{\tok{var}~v}{x_0, \dots, x_{k-1}} = r
    }~\textsc{Var}
    \and
    \inferrule{
        \semanticBracketsUsing{v}{x_0, \dots, x_{k-1}} = r
    }{
        \semanticBracketsUsing{\tok{nvar}~v}{x_0, \dots, x_{k-1}} = \lnot r
    }~\textsc{NotVar}
    \and
    \inferrule{
        \semanticBracketsUsing{c}{x_0, \dots, x_{k-1}} = r
    }{\semanticBracketsUsing{\tok{conj}~c}{x_0, \dots, x_{k-1}} = r}~\textsc{Conjunction}
    \and
    \inferrule{
        \semanticBracketsUsing{c}{x_0, \dots, x_{k-1}} = r_1 \\
        \semanticBracketsUsing{b}{x_0, \dots, x_{k-1}} = r_2
    }{\semanticBracketsUsing{c \mathbin{\tok{\lor}} b}{x_0, \dots, x_{k-1}} = r_1 \lor r_2}~\textsc{Or}
    \and
    \inferrule{
        \semanticBracketsUsing{v}{x_0, \dots, x_{k-1}} = r_1 \\
        \semanticBracketsUsing{c}{x_0, \dots, x_{k-1}} = r_2
    }{\semanticBracketsUsing{v \mathbin{\tok{\land}} c}{x_0, \dots, x_{k-1}} = r_1 \land r_2}~\textsc{And}
\end{mathpar}

%% file: semantics/semgus/cube.tex
\begin{mathpar}
    \inferrule{i = 0,1,\dots,(k-1)}{\semanticBracketsUsing{\tok{v}_i}{x_0, \dots, x_{k-1}} = x_i}~\textsc{VarAtom}
    \and
    \inferrule{
        \semanticBracketsUsing{v}{x_0, \dots, x_{k-1}} = r
    }{
        \semanticBracketsUsing{\tok{var}~v}{x_0, \dots, x_{k-1}} = r
    }~\textsc{Var}
    \and
    \inferrule{
        \semanticBracketsUsing{b}{x_0, \dots, x_{k-1}} = r_1 \\
        \semanticBracketsUsing{b}{x_0, \dots, x_{k-1}} = r_2
    }{\semanticBracketsUsing{b \mathbin{\tok{\land}} b}{x_0, \dots, x_{k-1}} = r_1 \land r_2}~\textsc{And}
\end{mathpar}

%% file: semantics/semgus/int-arith.tex
\begin{mathpar}
    \inferrule{i = 0,1,\dots,3}{\semanticBracketsUsing{i}{v_0} = i}~\textsc{IntLiteral}
    \and
    \inferrule{\xspace}{\semanticBracketsUsing{\tok{t}}{v_0} = \mathit{true}}~\textsc{True}
    \and
    \inferrule{\xspace}{\semanticBracketsUsing{\tok{f}}{v_0} = \mathit{false}}~\textsc{False}
    \and
    \inferrule{\xspace}{\semanticBracketsUsing{\tok{x}}{v_0} = v_0}~\textsc{VarX}
    \and
    \inferrule{\xspace}{\semanticBracketsUsing{\tok{y}}{v_0} = v_0}~\textsc{VarY}
    \and
    \inferrule{\xspace}{\semanticBracketsUsing{\tok{z}}{v_0} = v_0}~\textsc{VarZ}
    \and
    \inferrule{
        \semanticBracketsUsing{b}{v_0} = r_0 \\
        \semanticBracketsUsing{e_1}{v_0} = r_1 \\
        \semanticBracketsUsing{e_2}{v_0} = r_2 \\
        r_0
    }{\semanticBracketsUsing{\tok{ite}~b~e_1~e_2}{v_0} = r_1}~\textsc{Ite1}
    \and
    \inferrule{
        \semanticBracketsUsing{b}{v_0} = r_0 \\
        \semanticBracketsUsing{e_1}{v_0} = r_1 \\
        \semanticBracketsUsing{e_2}{v_0} = r_2 \\
        \lnot r_0
    }{\semanticBracketsUsing{\tok{ite}~b~e_1~e_2}{v_0} = r_2}~\textsc{Ite2}
    \and
    \inferrule{
        \semanticBracketsUsing{e_1}{v_0} = r_1 \\
        \semanticBracketsUsing{e_2}{v_0} = r_2
    }{\semanticBracketsUsing{e_1 \mathbin{\tok{+}} e_2}{v_0} = r_1 + r_2}~\textsc{Plus}
    \and
    \inferrule{
        \semanticBracketsUsing{e_1}{v_0} = r_1 \\
        \semanticBracketsUsing{e_2}{v_0} = r_2
    }{\semanticBracketsUsing{e_1 \mathbin{\tok{\times}} e_2}{v_0} = r_1 \cdot r_2}~\textsc{Multiply}
    \and
    \inferrule{
        \semanticBracketsUsing{e_1}{v_0} = r_1 \\
        \semanticBracketsUsing{e_2}{v_0} = r_2
    }{
        \semanticBracketsUsing{e_1\mathbin{\tok{<}}e_2}{v_0} = (r_1 < r_2)
    }~\textsc{LessThan}
    \and
    \inferrule{
        \semanticBracketsUsing{b_1}{v_0} = r_1 \\
        \semanticBracketsUsing{b_2}{v_0} = r_2 
    }{
        \semanticBracketsUsing{\mathtt{and}~b_1~b_2}{v_0} = (r_1 \land r_2)
    }~\textsc{And}
    \and
    \inferrule{
        \semanticBracketsUsing{b_1}{v_0} = r_1 \\
        \semanticBracketsUsing{b_2}{v_0} = r_2 
    }{
        \semanticBracketsUsing{\mathtt{or}~b_1~b_2}{v_0} = (r_1 \lor r_2)
    }~\textsc{Or}
    \and
    \inferrule{
        \semanticBracketsUsing{b}{v_0} = r
    }{
        \semanticBracketsUsing{\mathtt{not}~b}{v_0} = (\lnot r)
    }~\textsc{Not}
\end{mathpar}

%% file: semantics/semgus/regex.tex
\begin{mathpar}
    {M \coloneqq \begin{psmallmatrix}
        M_{0, 0} & M_{0, 1} & M_{0, 2} \\
        & M_{1, 1} & M_{1, 2} \\
        & & M_{2, 2}
    \end{psmallmatrix}} \and
    
    {M' \coloneqq \begin{psmallmatrix}
        M_{0, 0}' & M_{0, 1}' & M_{0, 2}' \\
        & M_{1, 1}' & M_{1, 2}' \\
        & & M_{2, 2}'
    \end{psmallmatrix}} \and
    
    \inferrule{\semanticBracketsUsing{e}{s} = M}{\semanticBracketsUsing{\tok{eval}~e}{s}=M_{0,2}}~\textsc{Eval} \and
    \inferrule{\xspace}{\semanticBracketsUsing{\tok{\epsilon}}{s} =
            {\begin{psmallmatrix}
             \top & \bot & \bot \\
              & \top & \bot \\
              & & \top
    \end{psmallmatrix}}}~\textsc{Eps} \and
    \inferrule{\xspace}{\semanticBracketsUsing{\tok{\phi}}{s} =
            {\begin{psmallmatrix}
             \bot & \bot & \bot \\
              & \bot & \bot \\
              & & \bot
    \end{psmallmatrix}}}~\textsc{Phi} \and
    \inferrule{\xspace}{\semanticBracketsUsing{\tok{a}}{s} =
            {\begin{psmallmatrix}
             \bot & (s_0=\tok{a}) & \bot \\
              & \bot & (s_1=\tok{a}) \\
              & & \bot
    \end{psmallmatrix}}}~\textsc{CharA} \and
    \inferrule{\xspace}{\semanticBracketsUsing{\tok{b}}{s} =
            {\begin{psmallmatrix}
             \bot & (s_0=\tok{b}) & \bot \\
              & \bot & (s_1=\tok{b}) \\
              & & \bot
    \end{psmallmatrix}}}~\textsc{CharB} \and
    \inferrule{\xspace}{\semanticBracketsUsing{\tok{?}}{s} =
            {\begin{psmallmatrix}
             \bot & (s_0\neq\epsilon) & \bot \\
              & \bot & (s_1\neq\epsilon) \\
              & & \bot
    \end{psmallmatrix}}}~\textsc{Any} \and
    \inferrule{\semanticBracketsUsing{e_1}{s} = M \\ \semanticBracketsUsing{e_2}{s} = M'}{\semanticBracketsUsing{e_1\mathbin{\tok{+}}e_2}{s} =
            {\begin{psmallmatrix}
             (M_{0,0} \lor M'_{0,0}) & (M_{0,1} \lor M'_{0,1}) & (M_{0,2} \lor M'_{0,2}) \\
              & (M_{1,1} \lor M'_{1,1}) & (M_{1,2} \lor M'_{1,2}) \\
              & & (M_{2,2} \lor M'_{2,2})
    \end{psmallmatrix}}}~\textsc{Or} \and
    \inferrule{\semanticBracketsUsing{e_1}{s} = M \\ \semanticBracketsUsing{e_2}{s} = M'}{\semanticBracketsUsing{e_1\mathbin{\tok{\cdot}}e_2}{s} =
            {\begin{psmallmatrix}
             (M_{0,0} \land M'_{0,0}) & \bigvee_{i=0 \dots 1}(M_{0,i} \land M'_{i,1}) & \bigvee_{i=0 \dots 2}(M_{0,i} \land M'_{i,2})\\
              & (M_{1,1} \land M'_{1,1}) & \bigvee_{i=1 \dots 2}(M_{1,i} \land M'_{i,2})\\
              & & (M_{2,2} \land M'_{2,2})
    \end{psmallmatrix}}}~\textsc{Concat} \and
    \inferrule{\semanticBracketsUsing{e}{s} = M}{\semanticBracketsUsing{e^{\tok{*}}}{s} =
            {\begin{psmallmatrix}
             \top & M_{0,1} & M_{0,2} \lor (M_{0,1} \land M_{1,2}) \\
                 & \top & M_{1,2}\\
                 &      & \top
    \end{psmallmatrix}}}~\textsc{Star} \and
    \inferrule{\semanticBracketsUsing{e}{s} = M}{\semanticBracketsUsing{e^{\tok{*}}}{s} =
            {\begin{psmallmatrix}
             \lnot M_{0,0} & \lnot M_{0,1} & \lnot M_{0,2} \\
                 & \lnot M_{1,1} & \lnot M_{1,2} \\
                 &      & \lnot M_{2,2}
    \end{psmallmatrix}}}~\textsc{Neg} 
\end{mathpar}

%% file: semantics/semgus/regex-simp.tex
\begin{mathpar}
    {M \coloneqq \begin{psmallmatrix}
        M_{0, 0} & M_{0, 1} & M_{0, 2} \\
        & M_{1, 1} & M_{1, 2} \\
        & & M_{2, 2}
    \end{psmallmatrix}} \and
    
    {M' \coloneqq \begin{psmallmatrix}
        M_{0, 0}' & M_{0, 1}' & M_{0, 2}' \\
        & M_{1, 1}' & M_{1, 2}' \\
        & & M_{2, 2}'
    \end{psmallmatrix}} \and
    \inferrule{\semanticBracketsUsing{e}{s} = (M_\epsilon, M)}{\semanticBracketsUsing{\tok{eval}~e}{s}=M_{0,1}}~\textsc{Eval} \and
    \inferrule{\xspace}{
        \semanticBracketsUsing{\tok{\epsilon}}{s} = (\top, 
            {\begin{psmallmatrix} \bot & \bot \\ & \bot \end{psmallmatrix}}
        )
    }~\textsc{Eps} \and
    \inferrule{\xspace}{\semanticBracketsUsing{\tok{\phi}}{s} = \left(\bot, 
            {\begin{psmallmatrix} \bot & \bot \\ & \bot \end{psmallmatrix}}\right)}~\textsc{Phi} \and
    \inferrule{\xspace}{\semanticBracketsUsing{\tok{a}}{s} = \left(\bot, 
    {\begin{psmallmatrix} s_0=\tok{a} & \bot \\ & s_0=\tok{a} \end{psmallmatrix}}
    \right))}~\textsc{CharA} \and
    \inferrule{\xspace}{\semanticBracketsUsing{\tok{b}}{s} = \left(\bot,
    {\begin{psmallmatrix} s_0=\tok{b} & \bot \\ & s_0=\tok{b} \end{psmallmatrix}}
    \right)}~\textsc{CharB} \and
    \inferrule{\xspace}{\semanticBracketsUsing{\tok{?}}{s} = \left(\bot, 
    {\begin{psmallmatrix} s_0\neq\epsilon & \bot \\ & s_0\neq\epsilon \end{psmallmatrix}}
    \right)}~\textsc{Any} \and
    \inferrule{\semanticBracketsUsing{e_1}{s} = (M_\epsilon, M) \\ \semanticBracketsUsing{e_2}{s} = (M_\epsilon', M')}{\semanticBracketsUsing{e_1\mathbin{\tok{+}}e_2}{s} =
        \left(
            M_\epsilon \lor M_\epsilon',
            {
            \begin{psmallmatrix}
                M_{0, 0} \lor M_{0, 0}'
                & M_{0, 1} \lor M_{0, 1}' \\
                & M_{1, 1} \lor M_{1, 1}'
            \end{psmallmatrix}
            }
        \right)}~\textsc{Or} \and
    \inferrule{\semanticBracketsUsing{e_1}{s} = (M_\epsilon, M) \\ \semanticBracketsUsing{e_2}{s} = (M_\epsilon', M')}{\semanticBracketsUsing{e_1\mathbin{\tok{\cdot}}e_2}{s} = \left(
            M_\epsilon \land M_\epsilon',
            {\begin{bsmallmatrix}
            (M_\epsilon \land M_{0, 0}') \lor (M_{0, 0} \land M_\epsilon') &
            (M_\epsilon \land M_{0, 1}') \lor (M_{0, 0} \land M_{1, 1}') \lor (M_{0, 1} \land M_\epsilon')\\
            & (M_\epsilon \land M_{1, 1}') \lor (M_{1, 1} \land M_\epsilon')
            \end{bsmallmatrix}}\right)}~\textsc{Concat} \and
    \inferrule{\semanticBracketsUsing{e}{s} = (M_\epsilon, M)}{\semanticBracketsUsing{e^{\tok{*}}}{s} =
            \left(\top, {\begin{psmallmatrix}
             M_{0,0} & M_{0,1} \lor (M_{0,0} \land M_{1,1}) \\
             & M_{1,1}
            \end{psmallmatrix}}\right)}~\textsc{Star} \and
    \inferrule{\semanticBracketsUsing{e}{s} = (M_\epsilon, M)}{\semanticBracketsUsing{e^{\tok{*}}}{s} =
            \left(\lnot M_\epsilon,
            {\begin{bsmallmatrix}
              \lnot M_{0,0} & \lnot M_{0,1} \\
                & \lnot M_{1,1} \\
    \end{bsmallmatrix}}
    \right)}~\textsc{Meg} 
\end{mathpar}

%% file: semantics/semgus/imp-1.tex
    \begin{mathparpagebreakable}
        \inferrule{\xspace}{\semanticBracketsUsing{\tok{0}}{v_0,v_1} = 0}~\textsc{Const0} \and
        \inferrule{\xspace}{\semanticBracketsUsing{\tok{1}}{v_0,v_1} = 1}~\textsc{Const1} \and
        \inferrule{\xspace}{\semanticBracketsUsing{\tok{t}}{v_0,v_1} = \top}~\textsc{ConstT} \and
        \inferrule{\xspace}{\semanticBracketsUsing{\tok{f}}{v_0,v_1} = \bot}~\textsc{ConstF} \and
        \inferrule{\xspace}{\semanticBracketsUsing{\tok{x}}{v_0,v_1} = v_0}~\textsc{VarX} \and
        \inferrule{\xspace}{\semanticBracketsUsing{\tok{y}}{v_0,v_1} = v_1}~\textsc{VarY} \and
        \inferrule{\semanticBracketsUsing{e_1}{v_0,v_1} = u_1 \\ \semanticBracketsUsing{e_2}{v_0,v_1} = u_2}{\semanticBracketsUsing{e_1~\mathbin{\tok{+}}~e_2}{v_0,v_1} = u_1 + u_2}~\textsc{Plus} \and
        \inferrule{\semanticBracketsUsing{e_1}{v_0,v_1} = u_1 \\ \semanticBracketsUsing{e_2}{v_0,v_1} = u_2}{\semanticBracketsUsing{e_1~\mathbin{\tok{-}}~e_2}{v_0,v_1} = u_1 - u_2}~\textsc{Minus} \and
        \inferrule{\semanticBracketsUsing{e_1}{v_0,v_1} = u_1 \\ \semanticBracketsUsing{e_2}{v_0,v_1} = u_2 \\ u_1 < u_2}{\semanticBracketsUsing{e_1~\mathbin{\tok <}~e_2}{v_0,v_1} = \top}~\textsc{LessThanTrue} \and
        \inferrule{\semanticBracketsUsing{e_1}{v_0,v_1} = u_1 \\ \semanticBracketsUsing{e_2}{v_0,v_1} = u_2 \\ u_1 \geq u_2}{\semanticBracketsUsing{e_1~\mathbin{\tok <}~e_2}{v_0,v_1} = \bot}~\textsc{LessThanFalse} \and
        \inferrule{\semanticBracketsUsing{b_1}{v_0,v_1} = u_1 \\ \semanticBracketsUsing{b_2}{v_0,v_1} = u_2 }{\semanticBracketsUsing{b_1~\mathbin{\tok{and}}~b_2}{v_0,v_1} = u_1 \land u_2}~\textsc{BoolAnd} \and
        \inferrule{\semanticBracketsUsing{b_1}{v_0,v_1} = u_1 \\ \semanticBracketsUsing{b_2}{v_0,v_1} = u_2 }{\semanticBracketsUsing{b_1~\mathbin{\tok{or}}~b_2}{v_0,v_1} = u_1 \lor u_2}~\textsc{BoolOr} \and
        \inferrule{\semanticBracketsUsing{b}{v_0,v_1} = u }{\semanticBracketsUsing{\mathord{\tok{not}}~b}{v_0,v_1} = \lnot u}~\textsc{BoolNot} \and
        \inferrule{\semanticBracketsUsing{e}{v_0,v_1} = v}{\semanticBracketsUsing{\tok{x:=}~e}{v_0,v_1} = (v, v_1)}~\textsc{AssignX} \and
        
        \inferrule{\semanticBracketsUsing{e}{v_0,v_1} = v}{\semanticBracketsUsing{\tok{y:=}~e}{v_0,v_1} = (v_0, v)}~\textsc{AssignY} \and
        
        \inferrule{\xspace}{\semanticBracketsUsing{\tok{x++}}{v_0,v_1} = (v_0 + 1, v_1)}~\textsc{IncX} \and
        
        \inferrule{\xspace}{\semanticBracketsUsing{\tok{y++}}{v_0,v_1} = (v_0, v_1 + 1)}~\textsc{IncY} \and
        
        \inferrule{\xspace}{\semanticBracketsUsing{\tok{x--}}{v_0,v_1} = (v_0 - 1, v_1)}~\textsc{DecX} \and
        
        \inferrule{\xspace}{\semanticBracketsUsing{\tok{y--}}{v_0,v_1} = (v_0, v_1 - 1)}~\textsc{DecY}
    \end{mathparpagebreakable}

%% file: semantics/semgus/imp-2.tex
    \begin{mathparpagebreakable}
        \inferrule{\semanticBracketsUsing{s_1}{v_0,v_1} = (v_0',v_1') \\ \semanticBracketsUsing{s_2}{v_0',v_1'} = (v_0'', v_1'')}{\semanticBracketsUsing{s_1\tok{;}s_2}{v_0,v_1} = (v_0'', v_1'')}~\textsc{Seq} \and

        \inferrule{\semanticBracketsUsing{b}{v_0,v_1} = v \\ \semanticBracketsUsing{s_1}{v_0,v_1} = (v_0', v_1') \\ \semanticBracketsUsing{s_2}{v_0,v_1} = (v_0'', v_1'')}{\semanticBracketsUsing{\tok{if}~b~\tok{then}~s_1~\tok{else}~s_2}{v_0,v_1} = v \mathbin{?} (v_0', v_1') \mathbin{:} (v_0'', v_1'')}~\textsc{Ite} \and
        
        \inferrule{\semanticBracketsUsing{b}{v_0,v_1} = \top \\ \semanticBracketsUsing{s}{v_0,v_1} = (v_0', v_1') \\ \semanticBracketsUsing{\tok{while}~b~\tok{do}~s}{v_0', v_1'} = (v_0'', v_1'')}{\semanticBracketsUsing{\tok{while}~b~\tok{do}~s}{v_0,v_1} = (v_0'', v_1'')}~\textsc{WhileLoop}
        \and
        
        \inferrule{\semanticBracketsUsing{b}{v_0,v_1} = \bot}{\semanticBracketsUsing{\tok{while}~b~\tok{do}~s}{v_0,v_1} = (v_0, v_1)}~\textsc{WhileEnd} \and
        
        \inferrule{\semanticBracketsUsing{s}{v_0,v_1} = (v_0', v_1') \\ \semanticBracketsUsing{b}{v_0',v_1'} = \top \\ \semanticBracketsUsing{\tok{do}~s~\tok{while}~b}{v_0',v_1'} = (v_0'', v_1'')}{\semanticBracketsUsing{\tok{do}~s~\tok{while}~b}{v_0,v_1} = (v_0'', v_1'')}~\textsc{DoWhileLoop}
        \and
        
        \inferrule{\semanticBracketsUsing{s}{v_0,v_1} = (v_0', v_1') \\ \semanticBracketsUsing{b}{v_0', v_1'} = \bot}{\semanticBracketsUsing{\tok{do}~s~\tok{while}~b}{v_0,v_1} = (v_0', v_1')}~\textsc{DoWhileEnd}
        
    \end{mathparpagebreakable}

%% file: semantics/semgus/imp-multi.tex
    \begin{mathparpagebreakable}
        \inferrule{\xspace}{\semanticBracketsUsing{\tok{0}}{\sigma} = 0}~\textsc{Const0} \and
        \inferrule{\xspace}{\semanticBracketsUsing{\tok{1}}{\sigma} = 1}~\textsc{Const1} \and
        \inferrule{\xspace}{\semanticBracketsUsing{\tok{t}}{\sigma} = \top}~\textsc{ConstT} \and
        \inferrule{\xspace}{\semanticBracketsUsing{\tok{f}}{\sigma} = \bot}~\textsc{ConstF} \and
        \inferrule{\xspace}{\semanticBracketsUsing{\tok{var}_i}{\sigma} = \sigma[i]}~\textsc{Var} \and
        \inferrule{\semanticBracketsUsing{e_1}{\sigma} = u_1 \\ \semanticBracketsUsing{e_2}{\sigma} = u_2}{\semanticBracketsUsing{e_1~\mathbin{\tok{+}}~e_2}{\sigma} = u_1 + u_2}~\textsc{Plus} \and
        \inferrule{\semanticBracketsUsing{e_1}{\sigma} = u_1 \\ \semanticBracketsUsing{e_2}{\sigma} = u_2}{\semanticBracketsUsing{e_1~\mathbin{\tok{-}}~e_2}{\sigma} = u_1 - u_2}~\textsc{Minus} \and
        \inferrule{\semanticBracketsUsing{e_1}{\sigma} = u_1 \\ \semanticBracketsUsing{e_2}{\sigma} = u_2 \\ u_1 < u_2}{\semanticBracketsUsing{e_1~\mathbin{\tok <}~e_2}{\sigma} = \top}~\textsc{LessThanTrue} \and
        \inferrule{\semanticBracketsUsing{e_1}{\sigma} = u_1 \\ \semanticBracketsUsing{e_2}{\sigma} = u_2 \\ u_1 \geq u_2}{\semanticBracketsUsing{e_1~\mathbin{\tok <}~e_2}{\sigma} = \bot}~\textsc{LessThanFalse} \and
        \inferrule{\semanticBracketsUsing{b_1}{\sigma} = u_1 \\ \semanticBracketsUsing{b_2}{\sigma} = u_2 }{\semanticBracketsUsing{b_1~\mathbin{\tok{and}}~b_2}{\sigma} = u_1 \land u_2}~\textsc{BoolAnd} \and
        \inferrule{\semanticBracketsUsing{b_1}{\sigma} = u_1 \\ \semanticBracketsUsing{b_2}{\sigma} = u_2 }{\semanticBracketsUsing{b_1~\mathbin{\tok{or}}~b_2}{\sigma} = u_1 \lor u_2}~\textsc{BoolOr} \and
        \inferrule{\semanticBracketsUsing{b}{\sigma} = u }{\semanticBracketsUsing{\mathord{\tok{not}}~b}{\sigma} = \lnot u}~\textsc{BoolNot} \and
        \inferrule{\semanticBracketsUsing{e}{\sigma} = v}{\semanticBracketsUsing{\tok{var}_i~\tok{:=}~e}{\sigma} = \sigma[i \mapsto v]}~\textsc{Assign} \and
        
        \inferrule{\xspace}{\semanticBracketsUsing{\tok{inc\_var}_i}{\sigma} = \sigma[i \mapsto \sigma[i] + 1]}~\textsc{Inc} \and
    
        \inferrule{\xspace}{\semanticBracketsUsing{\tok{dec\_var}_i}{\sigma} = \sigma[i \mapsto \sigma[0] - 1]}~\textsc{Dec} \and
        
        \inferrule{\semanticBracketsUsing{s_1}{\sigma} = \sigma' \\ \semanticBracketsUsing{s_2}{\sigma'} = \sigma''}{\semanticBracketsUsing{s_1\tok{;}s_2}{\sigma} = \sigma''}~\textsc{Seq} \and

        \inferrule{\semanticBracketsUsing{b}{\sigma} = v \\ \semanticBracketsUsing{s_1}{\sigma} = \sigma' \\ \semanticBracketsUsing{s_2}{\sigma} = \sigma''}{\semanticBracketsUsing{\tok{if}~b~\tok{then}~s_1~\tok{else}~s_2}{\sigma} = v \mathbin{?} \sigma' \mathbin{:} \sigma''}~\textsc{Ite} \and
        
        \inferrule{\semanticBracketsUsing{b}{\sigma} = \top \\ \semanticBracketsUsing{s}{\sigma} = \sigma' \\ \semanticBracketsUsing{\tok{while}~b~\tok{do}~s}{\sigma'} = \sigma''}{\semanticBracketsUsing{\tok{while}~b~\tok{do}~s}{\sigma} = \sigma''}~\textsc{WhileLoop}
        \and
        
        \inferrule{\semanticBracketsUsing{b}{\sigma} = \bot}{\semanticBracketsUsing{\tok{while}~b~\tok{do}~s}{\sigma} = (\sigma)}~\textsc{WhileEnd} \and
        
        \inferrule{\semanticBracketsUsing{s}{\sigma} = \sigma' \\ \semanticBracketsUsing{b}{\sigma'} = \top \\ \semanticBracketsUsing{\tok{do}~s~\tok{while}~b}{\sigma'} = \sigma''}{\semanticBracketsUsing{\tok{do}~s~\tok{while}~b}{\sigma} = \sigma''}~\textsc{DoWhileLoop}
        \and
        
        \inferrule{\semanticBracketsUsing{s}{\sigma} = \sigma' \\ \semanticBracketsUsing{b}{\sigma'} = \bot}{\semanticBracketsUsing{\tok{do}~s~\tok{while}~b}{\sigma} = \sigma'}~\textsc{DoWhileEnd}
        
    \end{mathparpagebreakable}

%% file: semantics/semgus/bv-simple.tex
\begin{mathpar}
    \inferrule{i = 0,1,\dots,(k-1)}{\semanticBracketsUsing{\tok{v}_i}{x_0, \dots, x_{k-1}} = x_i}~\textsc{VarAtom}
    \and
    \inferrule{\xspace}{\semanticBracketsUsing{\tok{0}}{x_0, \dots, x_{k-1}}=\mathtt{0x00000000}}~\textsc{BvZero} \and
    \inferrule{\xspace}{\semanticBracketsUsing{\tok{1}}{x_0, \dots, x_{k-1}}=\mathtt{0x00000001}}~\textsc{BvOne}
    \and
    \inferrule{
    \semanticBracketsUsing{e_1}{x_0, \dots, x_{k-1}} = r_1 \\
    \semanticBracketsUsing{e_2}{x_0, \dots, x_{k-1}} = r_2
    }{
    \semanticBracketsUsing{e_1 \mathbin{\tok{<}} e_2}{x_0, \dots, x_{k-1}} = r_1 <_{\mathit{unsigned}} r_2
    }~\textsc{BvUlt}
    \and
    \inferrule{
    \semanticBracketsUsing{e_1}{x_0, \dots, x_{k-1}} = r_1 \\
    \semanticBracketsUsing{e_2}{x_0, \dots, x_{k-1}} = r_2
    }{
    \semanticBracketsUsing{e_1 \mathbin{\tok{\ge}} e_2}{x_0, \dots, x_{k-1}} = \lnot (r_1 <_{\mathit{unsigned}} r_2)
    }~\textsc{BvUge}
    \and
    \inferrule{
    \semanticBracketsUsing{e_1}{x_0, \dots, x_{k-1}} = r_1 \\
    \semanticBracketsUsing{e_2}{x_0, \dots, x_{k-1}} = r_2
    }{
    \semanticBracketsUsing{e_1 \mathbin{\tok{\le}} e_2}{x_0, \dots, x_{k-1}} = \lnot (r_2 <_{\mathit{unsigned}} r_1)
    }~\textsc{BvUle}
    \and
    \inferrule{
    \semanticBracketsUsing{e_1}{x_0, \dots, x_{k-1}} = r_1 \\
    \semanticBracketsUsing{e_2}{x_0, \dots, x_{k-1}} = r_2 \\
    \odot \in \{\&, \mid, \oplus, \ggg, \ll \}
    }{
    \semanticBracketsUsing{e_1 \mathbin{\tok{\odot}} e_2}{x_0, \dots, x_{k-1}} =  (r_1 \odot r_2)
    }~\textsc{BvBitwise}
    \and
    \inferrule{
    \semanticBracketsUsing{e}{x_0, \dots, x_{k-1}} = r \\ r \neq \mathtt{0x00000000}
    }{
    \semanticBracketsUsing{\tok{any\_bit}~e}{x_0, \dots, x_{k-1}} = \mathtt{0x00000001}
    }~\textsc{AnyBit1}
    \and
    \inferrule{
    \semanticBracketsUsing{e}{x_0, \dots, x_{k-1}} = r \\ r = \mathtt{0x00000000}
    }{
    \semanticBracketsUsing{\tok{any\_bit}~e}{x_0, \dots, x_{k-1}} = \mathtt{0x00000000}
    }~\textsc{AnyBit0}
    \and
    \inferrule{
    \semanticBracketsUsing{e}{x_0, \dots, x_{k-1}} = r
    }{
    \semanticBracketsUsing{\mathord{\tok{\sim}}~e}{x_0, \dots, x_{k-1}} = \sim e
    }~\textsc{BvNot}
    \and
    \inferrule{
    \semanticBracketsUsing{e}{x_0, \dots, x_{k-1}} = r
    }{
    \semanticBracketsUsing{\mathord{\tok{\lnot}}~e}{x_0, \dots, x_{k-1}} = \lnot e
    }~\textsc{BvNeg}
    \and
    \inferrule{
    \semanticBracketsUsing{e_1}{x_0, \dots, x_{k-1}} = r_1 \\
    \semanticBracketsUsing{e_2}{x_0, \dots, x_{k-1}} = r_2 \\
    \odot \in \{+,-,\times,\div\}
    }{
    \semanticBracketsUsing{e_1 \mathbin{\tok{\odot}} e_2}{x_0, \dots, x_{k-1}} =  (r_1 \odot r_2)
    }~\textsc{BvArith}
\end{mathpar}

%% file: semantics/semgus/bv-saturated.tex
\begin{mathpar}
    \inferrule{i = 0,1,\dots,(k-1)}{\semanticBracketsUsing{\tok{v}_i}{x_0, \dots, x_{k-1}} = x_i}~\textsc{VarAtom}
    \and
    \inferrule{\xspace}{\semanticBracketsUsing{\tok{0}}{x_0, \dots, x_{k-1}}=\mathtt{0x00000000}}~\textsc{BvZero} \and
    \inferrule{\xspace}{\semanticBracketsUsing{\tok{1}}{x_0, \dots, x_{k-1}}=\mathtt{0x00000001}}~\textsc{BvOne}
    \and
    \inferrule{
    \semanticBracketsUsing{e_1}{x_0, \dots, x_{k-1}} = r_1 \\
    \semanticBracketsUsing{e_2}{x_0, \dots, x_{k-1}} = r_2
    }{
    \semanticBracketsUsing{e_1 \mathbin{\tok{<}} e_2}{x_0, \dots, x_{k-1}} = r_1 <_{\mathit{unsigned}} r_2
    }~\textsc{BvUlt}
    \and
    \inferrule{
    \semanticBracketsUsing{e_1}{x_0, \dots, x_{k-1}} = r_1 \\
    \semanticBracketsUsing{e_2}{x_0, \dots, x_{k-1}} = r_2
    }{
    \semanticBracketsUsing{e_1 \mathbin{\tok{\ge}} e_2}{x_0, \dots, x_{k-1}} = \lnot (r_1 <_{\mathit{unsigned}} r_2)
    }~\textsc{BvUge}
    \and
    \inferrule{
    \semanticBracketsUsing{e_1}{x_0, \dots, x_{k-1}} = r_1 \\
    \semanticBracketsUsing{e_2}{x_0, \dots, x_{k-1}} = r_2
    }{
    \semanticBracketsUsing{e_1 \mathbin{\tok{\le}} e_2}{x_0, \dots, x_{k-1}} = \lnot (r_2 <_{\mathit{unsigned}} r_1)
    }~\textsc{BvUle}
    \and
    \inferrule{
    \semanticBracketsUsing{e_1}{x_0, \dots, x_{k-1}} = r_1 \\
    \semanticBracketsUsing{e_2}{x_0, \dots, x_{k-1}} = r_2 \\
    \odot \in \{\&, \mid, \oplus, \ggg, \ll \}
    }{
    \semanticBracketsUsing{e_1 \mathbin{\tok{\odot}} e_2}{x_0, \dots, x_{k-1}} =  (r_1 \odot r_2)
    }~\textsc{BvBitwise}
    \and
    \inferrule{
    \semanticBracketsUsing{e}{x_0, \dots, x_{k-1}} = r \\ r \neq \mathtt{0x00000000}
    }{
    \semanticBracketsUsing{\tok{any\_bit}~e}{x_0, \dots, x_{k-1}} = \mathtt{0x00000001}
    }~\textsc{AnyBit1}
    \and
    \inferrule{
    \semanticBracketsUsing{e}{x_0, \dots, x_{k-1}} = r \\ r = \mathtt{0x00000000}
    }{
    \semanticBracketsUsing{\tok{any\_bit}~e}{x_0, \dots, x_{k-1}} = \mathtt{0x00000000}
    }~\textsc{AnyBit0}
    \and
    \inferrule{
    \semanticBracketsUsing{e}{x_0, \dots, x_{k-1}} = r
    }{
    \semanticBracketsUsing{\mathord{\tok{\sim}}~e}{x_0, \dots, x_{k-1}} = \sim e
    }~\textsc{BvNot}
    \and
    \inferrule{
    \semanticBracketsUsing{e}{x_0, \dots, x_{k-1}} = r
    }{
    \semanticBracketsUsing{\mathord{\tok{\lnot}}~e}{x_0, \dots, x_{k-1}} = \lnot e
    }~\textsc{BvNeg}
    \and
    \inferrule{
    \semanticBracketsUsing{e_1}{x_0, \dots, x_{k-1}} = r_1 \\
    \semanticBracketsUsing{e_2}{x_0, \dots, x_{k-1}} = r_2 \\
    \odot \in \{+,-,\times,\div\}
    }{
    \semanticBracketsUsing{e_1 \mathbin{\tok{\odot}} e_2}{x_0, \dots, x_{k-1}} =  (r_1 \odot_{\mathit{sat}} r_2)
    }~\textsc{BvArith}
\end{mathpar}

\begin{align*}
    a \mathbin{\odot_{sat}} b &= a\mathbin{\odot}b &\text{(when no overflow or underflow)}\\
    a \mathbin{\odot_{sat}} b &= \mathtt{0xffffffff} &\text{(when overflow happens)} \\
    a \mathbin{\odot_{sat}} b &= \mathtt{0x00000000} &\text{(when underflow happens)} \\
\end{align*}

%% file: semantics/semgus/bv-imp-simple.tex
\begin{mathpar}
    \inferrule{i = 0,1,\dots,(m-1)}{\semanticBracketsUsing{\tok{v}_i}{u_0, \dots, u_{m-1}, x_0, \dots, x_{n-1}} = u_i}~\textsc{ConstAtom}
    \and
    \inferrule{i = 0,1,\dots,(n-1)}{\semanticBracketsUsing{\tok{o}_i}{u_0, \dots, u_{m-1}, x_0, \dots, x_{n-1}} = x_i}~\textsc{VarAtom}
    \and
    \inferrule{
        \semanticBracketsUsing{e}{u_0, \dots, u_{m-1}, x_0, \dots, x_{n-1}} = r
    }{\semanticBracketsUsing{\tok{o}_i~\tok{:=}~e}{u_0, \dots, u_{m-1}, x_0, \dots, x_{n-1}} = (u_0, \dots, u_{m-1}, x_0, \dots, x_{i-1}, r, x_{i+1}, \dots, x_{n-1})}~\textsc{Assign}
    \and
    \inferrule{
        \semanticBracketsUsing{s_1}{\vec{u}, \vec{x}} = (\vec{u}', \vec{x}') \\
        \semanticBracketsUsing{s_1}{\vec{u}', \vec{x}'} = (\vec{u}'', \vec{x}'')
    }{
        \semanticBracketsUsing{s_1~\tok{;}~s_2}{\vec{u}, \vec{x}} = (\vec{u}'', \vec{x}'')
    }~\textsc{Seq}
    \and
    \vec{u} = (u_0,\dots,u_{m-1}), \vec{x} = (x_0,\dots,x_{n-1})
    \and
    \inferrule{\xspace}{\semanticBracketsUsing{\tok{0}}{u_0, \dots, u_{m-1}, x_0, \dots, x_{n-1}}=\mathtt{0x00000000}}~\textsc{BvZero} \and
    \inferrule{\xspace}{\semanticBracketsUsing{\tok{1}}{u_0, \dots, u_{m-1}, x_0, \dots, x_{n-1}}=\mathtt{0x00000001}}~\textsc{BvOne}
    \and
    \inferrule{
    \semanticBracketsUsing{e_1}{u_0, \dots, u_{m-1}, x_0, \dots, x_{n-1}} = r_1 \\
    \semanticBracketsUsing{e_2}{u_0, \dots, u_{m-1}, x_0, \dots, x_{n-1}} = r_2
    }{
    \semanticBracketsUsing{e_1 \mathbin{\tok{<}} e_2}{u_0, \dots, u_{m-1}, x_0, \dots, x_{n-1}} = r_1 <_{\mathit{unsigned}} r_2
    }~\textsc{BvUlt}
    \and
    \inferrule{
    \semanticBracketsUsing{e_1}{u_0, \dots, u_{m-1}, x_0, \dots, x_{n-1}} = r_1 \\
    \semanticBracketsUsing{e_2}{u_0, \dots, u_{m-1}, x_0, \dots, x_{n-1}} = r_2
    }{
    \semanticBracketsUsing{e_1 \mathbin{\tok{\ge}} e_2}{u_0, \dots, u_{m-1}, x_0, \dots, x_{n-1}} = \lnot (r_1 <_{\mathit{unsigned}} r_2)
    }~\textsc{BvUge}
    \and
    \inferrule{
    \semanticBracketsUsing{e_1}{u_0, \dots, u_{m-1}, x_0, \dots, x_{n-1}} = r_1 \\
    \semanticBracketsUsing{e_2}{u_0, \dots, u_{m-1}, x_0, \dots, x_{n-1}} = r_2
    }{
    \semanticBracketsUsing{e_1 \mathbin{\tok{\le}} e_2}{u_0, \dots, u_{m-1}, x_0, \dots, x_{n-1}} = \lnot (r_2 <_{\mathit{unsigned}} r_1)
    }~\textsc{BvUle}
    \and
    \inferrule{
    \semanticBracketsUsing{e_1}{u_0, \dots, u_{m-1}, x_0, \dots, x_{n-1}} = r_1 \\
    \semanticBracketsUsing{e_2}{u_0, \dots, u_{m-1}, x_0, \dots, x_{n-1}} = r_2 \\
    \odot \in \{\&, \mid, \oplus, \ggg, \ll \}
    }{
    \semanticBracketsUsing{e_1 \mathbin{\tok{\odot}} e_2}{u_0, \dots, u_{m-1}, x_0, \dots, x_{n-1}} =  (r_1 \odot r_2)
    }~\textsc{BvBitwise}
    \and
    \inferrule{
    \semanticBracketsUsing{e}{u_0, \dots, u_{m-1}, x_0, \dots, x_{n-1}} = r \\ r \neq \mathtt{0x00000000}
    }{
    \semanticBracketsUsing{\tok{any\_bit}~e}{u_0, \dots, u_{m-1}, x_0, \dots, x_{n-1}} = \mathtt{0x00000001}
    }~\textsc{AnyBit1}
    \and
    \inferrule{
    \semanticBracketsUsing{e}{u_0, \dots, u_{m-1}, x_0, \dots, x_{n-1}} = r \\ r = \mathtt{0x00000000}
    }{
    \semanticBracketsUsing{\tok{any\_bit}~e}{u_0, \dots, u_{m-1}, x_0, \dots, x_{n-1}} = \mathtt{0x00000000}
    }~\textsc{AnyBit0}
    \and
    \inferrule{
    \semanticBracketsUsing{e}{u_0, \dots, u_{m-1}, x_0, \dots, x_{n-1}} = r
    }{
    \semanticBracketsUsing{\mathord{\tok{\sim}}~e}{u_0, \dots, u_{m-1}, x_0, \dots, x_{n-1}} = \sim e
    }~\textsc{BvNot}
    \and
    \inferrule{
    \semanticBracketsUsing{e}{u_0, \dots, u_{m-1}, x_0, \dots, x_{n-1}} = r
    }{
    \semanticBracketsUsing{\mathord{\tok{\lnot}}~e}{u_0, \dots, u_{m-1}, x_0, \dots, x_{n-1}} = \lnot e
    }~\textsc{BvNeg}
    \and
    \inferrule{
    \semanticBracketsUsing{e_1}{u_0, \dots, u_{m-1}, x_0, \dots, x_{n-1}} = r_1 \\
    \semanticBracketsUsing{e_2}{u_0, \dots, u_{m-1}, x_0, \dots, x_{n-1}} = r_2 \\
    \odot \in \{+,-,\times,\div\}
    }{
    \semanticBracketsUsing{e_1 \mathbin{\tok{\odot}} e_2}{u_0, \dots, u_{m-1}, x_0, \dots, x_{n-1}} =  (r_1 \odot r_2)
    }~\textsc{BvArith}
\end{mathpar}

%% file: semantics/semgus/bv-imp-saturated.tex
\begin{mathpar}
    \inferrule{i = 0,1,\dots,(m-1)}{\semanticBracketsUsing{\tok{v}_i}{u_0, \dots, u_{m-1}, x_0, \dots, x_{n-1}} = u_i}~\textsc{ConstAtom}
    \and
    \inferrule{i = 0,1,\dots,(n-1)}{\semanticBracketsUsing{\tok{o}_i}{u_0, \dots, u_{m-1}, x_0, \dots, x_{n-1}} = x_i}~\textsc{VarAtom}
    \and
    \inferrule{
        \semanticBracketsUsing{e}{u_0, \dots, u_{m-1}, x_0, \dots, x_{n-1}} = r
    }{\semanticBracketsUsing{\tok{o}_i~\tok{:=}~e}{u_0, \dots, u_{m-1}, x_0, \dots, x_{n-1}} = (u_0, \dots, u_{m-1}, x_0, \dots, x_{i-1}, r, x_{i+1}, \dots, x_{n-1})}~\textsc{Assign}
    \and
    \inferrule{
        \semanticBracketsUsing{s_1}{\vec{u}, \vec{x}} = (\vec{u}', \vec{x}') \\
        \semanticBracketsUsing{s_1}{\vec{u}', \vec{x}'} = (\vec{u}'', \vec{x}'')
    }{
        \semanticBracketsUsing{s_1~\tok{;}~s_2}{\vec{u}, \vec{x}} = (\vec{u}'', \vec{x}'')
    }~\textsc{Seq}
    \and
    \vec{u} = (u_0,\dots,u_{m-1}), \vec{x} = (x_0,\dots,x_{n-1})
    \and
    \inferrule{\xspace}{\semanticBracketsUsing{\tok{0}}{u_0, \dots, u_{m-1}, x_0, \dots, x_{n-1}}=\mathtt{0x00000000}}~\textsc{BvZero} \and
    \inferrule{\xspace}{\semanticBracketsUsing{\tok{1}}{u_0, \dots, u_{m-1}, x_0, \dots, x_{n-1}}=\mathtt{0x00000001}}~\textsc{BvOne}
    \and
    \inferrule{
    \semanticBracketsUsing{e_1}{u_0, \dots, u_{m-1}, x_0, \dots, x_{n-1}} = r_1 \\
    \semanticBracketsUsing{e_2}{u_0, \dots, u_{m-1}, x_0, \dots, x_{n-1}} = r_2
    }{
    \semanticBracketsUsing{e_1 \mathbin{\tok{<}} e_2}{u_0, \dots, u_{m-1}, x_0, \dots, x_{n-1}} = r_1 <_{\mathit{unsigned}} r_2
    }~\textsc{BvUlt}
    \and
    \inferrule{
    \semanticBracketsUsing{e_1}{u_0, \dots, u_{m-1}, x_0, \dots, x_{n-1}} = r_1 \\
    \semanticBracketsUsing{e_2}{u_0, \dots, u_{m-1}, x_0, \dots, x_{n-1}} = r_2
    }{
    \semanticBracketsUsing{e_1 \mathbin{\tok{\ge}} e_2}{u_0, \dots, u_{m-1}, x_0, \dots, x_{n-1}} = \lnot (r_1 <_{\mathit{unsigned}} r_2)
    }~\textsc{BvUge}
    \and
    \inferrule{
    \semanticBracketsUsing{e_1}{u_0, \dots, u_{m-1}, x_0, \dots, x_{n-1}} = r_1 \\
    \semanticBracketsUsing{e_2}{u_0, \dots, u_{m-1}, x_0, \dots, x_{n-1}} = r_2
    }{
    \semanticBracketsUsing{e_1 \mathbin{\tok{\le}} e_2}{u_0, \dots, u_{m-1}, x_0, \dots, x_{n-1}} = \lnot (r_2 <_{\mathit{unsigned}} r_1)
    }~\textsc{BvUle}
    \and
    \inferrule{
    \semanticBracketsUsing{e_1}{u_0, \dots, u_{m-1}, x_0, \dots, x_{n-1}} = r_1 \\
    \semanticBracketsUsing{e_2}{u_0, \dots, u_{m-1}, x_0, \dots, x_{n-1}} = r_2 \\
    \odot \in \{\&, \mid, \oplus, \ggg, \ll \}
    }{
    \semanticBracketsUsing{e_1 \mathbin{\tok{\odot}} e_2}{u_0, \dots, u_{m-1}, x_0, \dots, x_{n-1}} =  (r_1 \odot r_2)
    }~\textsc{BvBitwise}
    \and
    \inferrule{
    \semanticBracketsUsing{e}{u_0, \dots, u_{m-1}, x_0, \dots, x_{n-1}} = r \\ r \neq \mathtt{0x00000000}
    }{
    \semanticBracketsUsing{\tok{any\_bit}~e}{u_0, \dots, u_{m-1}, x_0, \dots, x_{n-1}} = \mathtt{0x00000001}
    }~\textsc{AnyBit1}
    \and
    \inferrule{
    \semanticBracketsUsing{e}{u_0, \dots, u_{m-1}, x_0, \dots, x_{n-1}} = r \\ r = \mathtt{0x00000000}
    }{
    \semanticBracketsUsing{\tok{any\_bit}~e}{u_0, \dots, u_{m-1}, x_0, \dots, x_{n-1}} = \mathtt{0x00000000}
    }~\textsc{AnyBit0}
    \and
    \inferrule{
    \semanticBracketsUsing{e}{u_0, \dots, u_{m-1}, x_0, \dots, x_{n-1}} = r
    }{
    \semanticBracketsUsing{\mathord{\tok{\sim}}~e}{u_0, \dots, u_{m-1}, x_0, \dots, x_{n-1}} = \sim e
    }~\textsc{BvNot}
    \and
    \inferrule{
    \semanticBracketsUsing{e}{u_0, \dots, u_{m-1}, x_0, \dots, x_{n-1}} = r
    }{
    \semanticBracketsUsing{\mathord{\tok{\lnot}}~e}{u_0, \dots, u_{m-1}, x_0, \dots, x_{n-1}} = \lnot e
    }~\textsc{BvNeg}
    \and
    \inferrule{
    \semanticBracketsUsing{e_1}{x_0, \dots, x_{k-1}} = r_1 \\
    \semanticBracketsUsing{e_2}{x_0, \dots, x_{k-1}} = r_2 \\
    \odot \in \{+,-,\times,\div\}
    }{
    \semanticBracketsUsing{e_1 \mathbin{\tok{\odot}} e_2}{x_0, \dots, x_{k-1}} =  (r_1 \odot_{\mathit{sat}} r_2)
    }~\textsc{BvArith}
\end{mathpar}

%% file: semantics/panini/binop.tex
    \begin{mathpar}
        \inferrule{\xspace}{\semanticBracketsUsing{\tok{1}}{v_0} = 1}~\textsc{One}
        \and
        \inferrule{\xspace}{\semanticBracketsUsing{\tok{0}}{v_0} = 0}~\textsc{Zero}
        \and
        \inferrule{\xspace}{\semanticBracketsUsing{\tok{x}}{v_0} = v_0}~\textsc{VarX}
        \and
        \inferrule{\semanticBracketsUsing{n}{v_0} = r}{\semanticBracketsUsing{\tok{count}~n}{v_0} = r}~\textsc{CountBit}
        \and
        \inferrule{\semanticBracketsUsing{m}{v_0} = r}{\semanticBracketsUsing{\tok{bin2dec}~m}{v_0} = r}~\textsc{BinToDec}
        \and
        \inferrule{\semanticBracketsUsing{n}{v_0} = r_1 \\ \semanticBracketsUsing{b}{v_0} = r_2}{\semanticBracketsUsing{\tok{concat}~n~b}{v_0} = r_1 + \mathtt{int}(r_2)}~\textsc{CountBitConcat}
        \and
        \inferrule{\semanticBracketsUsing{m}{v_0} = r_1 \\ \semanticBracketsUsing{b}{v_0} = r_2}{\semanticBracketsUsing{\tok{concat'}~m~b}{v_0} = 2 \cdot r_1 + \mathtt{int}(r_2)}~\textsc{BinToDecConcat}
        \and
        \inferrule{\semanticBracketsUsing{b}{v_0} = r}{\semanticBracketsUsing{\tok{atom}~b}{v_0} = r}~\textsc{CountBitAtom}
        \and
        \inferrule{\semanticBracketsUsing{b}{v_0} = r}{\semanticBracketsUsing{\tok{atom'}~b}{v_0} = r}~\textsc{BinToDecAtom}
    \end{mathpar}

%% file: semantics/panini/currency.tex
    \begin{mathpar}
        \inferrule{\xspace}{\semanticBracketsUsing{\tok{0}}{v_0} = 0}~\textsc{Zero}
        \and
        \inferrule{\xspace}{\semanticBracketsUsing{\tok{1}}{v_0} = 1}~\textsc{One}
        \and
        \inferrule{\xspace}{\semanticBracketsUsing{\tok{2}}{v_0} = 2}~\textsc{Two}
        \and
        \inferrule{\xspace}{\semanticBracketsUsing{\tok{4}}{v_0} = 4}~\textsc{Four}
        \and
        \inferrule{\xspace}{\semanticBracketsUsing{\tok{8}}{v_0} = 8}~\textsc{Eight}
        \and
        \inferrule{\xspace}{\semanticBracketsUsing{\tok{x}}{v_0} = v_0}~\textsc{VarX}
        \and
        \inferrule{
            \semanticBracketsUsing{k_1}{v_0} = r_1
            \semanticBracketsUsing{k_2}{v_0} = r_2
        }{\semanticBracketsUsing{k_1 \mathbin{\tok{+_k}} k_2}{v_0} = r_1 + r_2}~\textsc{ScalarPlus}
        \and
        \inferrule{
            \semanticBracketsUsing{s_1}{v_0} = r_1
            \semanticBracketsUsing{s_2}{v_0} = r_2
        }{\semanticBracketsUsing{s_1 \mathbin{\tok{+}} s_2}{v_0} = r_1 + r_2}~\textsc{CurrencyPlus}
        \and
        \inferrule{
            \semanticBracketsUsing{s_1}{v_0} = r_1
            \semanticBracketsUsing{s_2}{v_0} = r_2
        }{\semanticBracketsUsing{s_1 \mathbin{\tok{-}} s_2}{v_0} = r_1 - r_2}~\textsc{CurrencySubtract}
        \and
        \inferrule{
            \semanticBracketsUsing{s}{v_0} = r_1
            \semanticBracketsUsing{k}{v_0} = r_2
        }{\semanticBracketsUsing{s \mathbin{\tok{\times}} k}{v_0} = r_1 \cdot r_2}~\textsc{CurrencyTimesScalar}
        \and
        \inferrule{
            \semanticBracketsUsing{k}{v_0} = r
        }{\semanticBracketsUsing{\tok{jpy}~k}{v_0} = r}~\textsc{CurrencyJpy}
        \and
        \inferrule{
            \semanticBracketsUsing{k}{v_0} = r
        }{\semanticBracketsUsing{\tok{cny}~k}{v_0} = 21 \cdot r}~\textsc{CurrencyCny}
        \and
        \inferrule{
            \semanticBracketsUsing{k}{v_0} = r
        }{\semanticBracketsUsing{\tok{usd}~k}{v_0} = 152 \cdot r}~\textsc{CurrencyUsd}
        
    \end{mathpar}

%% file: semantics/panini/diff.tex
    \begin{mathpar}
        \inferrule{\xspace}{\semanticBracketsUsing{\tok{0}}{v_0} = (0, 0, 0))}~\textsc{Zero}
        \and
        \inferrule{\xspace}{\semanticBracketsUsing{\tok{1}}{v_0} = (1, 1, 0)}~\textsc{One}
        \and
        \inferrule{\xspace}{\semanticBracketsUsing{\tok{2}}{v_0} = (2, 2, 0)}~\textsc{Two}
        \and
        \inferrule{\xspace}{\semanticBracketsUsing{\tok{x}}{v_0} = (v_0, v_0 + 1, 0)}~\textsc{VarX}
        \and
        \inferrule{
            \semanticBracketsUsing{e_1}{v_0} = (r_1, s_1, t_1)
            \semanticBracketsUsing{e_2}{v_0} = (r_2, s_2, t_2)
        }{\semanticBracketsUsing{e_1 \mathbin{\tok{+}} e_2}{v_0} = (r_1+r_2, s_1+s_2, t_1 + t_2)}~\textsc{Plus}
        \and
        \inferrule{
            \semanticBracketsUsing{e_1}{v_0} = (r_1, s_1, t_1)
            \semanticBracketsUsing{e_2}{v_0} = (r_2, s_2, t_2)
        }{\semanticBracketsUsing{e_1 \mathbin{\tok{\times}} e_2}{v_0} = (r_1 \cdot r_2, s_1 \cdot s_2, r_1 \cdot t_2 + r_2 \cdot s_1)}~\textsc{Multiply}
    \end{mathpar}

%% file: semantics/panini/ite-expr.tex
\begin{mathpar}
    \inferrule{i = 0,1,\dots,8}{\semanticBracketsUsing{i}{v_0} = i}~\textsc{IntLiteral}
    \and
    \inferrule{\xspace}{\semanticBracketsUsing{\tok{x}}{v_0} = v_0}~\textsc{VarX}
    \and
    \inferrule{
        \semanticBracketsUsing{e}{v_0} = r
    }{\semanticBracketsUsing{\tok{expr}~e}{v_0} = r}~\textsc{Expr}
    \and
    \inferrule{
        \semanticBracketsUsing{b}{v_0} = r_0 \\
        \semanticBracketsUsing{e_1}{v_0} = r_1 \\
        \semanticBracketsUsing{e_2}{v_0} = r_2 \\
        r_0
    }{\semanticBracketsUsing{\tok{ite}~b~e_1~e_2}{v_0} = r_1}~\textsc{Ite1}
    \and
    \inferrule{
        \semanticBracketsUsing{b}{v_0} = r_0 \\
        \semanticBracketsUsing{e_1}{v_0} = r_1 \\
        \semanticBracketsUsing{e_2}{v_0} = r_2 \\
        \lnot r_0
    }{\semanticBracketsUsing{\tok{ite}~b~e_1~e_2}{v_0} = r_2}~\textsc{Ite2}
    \and
    \inferrule{
        \semanticBracketsUsing{e}{v_0} = r_1 \\
        \semanticBracketsUsing{f}{v_0} = r_2
    }{\semanticBracketsUsing{e \mathbin{\tok{+}} f}{v_0} = r_1 + r_2}~\textsc{Plus}
    \and
    \inferrule{
        \semanticBracketsUsing{e}{v_0} = r_1 \\
        \semanticBracketsUsing{f}{v_0} = r_2
    }{\semanticBracketsUsing{e \mathbin{\tok{-}} f}{v_0} = r_1 - r_2}~\textsc{Minus}
    \and
    \inferrule{
        \semanticBracketsUsing{f}{v_0} = r
    }{\semanticBracketsUsing{\tok{atom}~f}{v_0} = r}~\textsc{Atom}
    \and
    \inferrule{
        \semanticBracketsUsing{f}{v_0} = r_1 \\
        \semanticBracketsUsing{g}{v_0} = r_2
    }{\semanticBracketsUsing{f \mathbin{\tok{*}} g}{v_0} = r_1 \cdot r_2}~\textsc{Multiply}
    \and
    \inferrule{
        \semanticBracketsUsing{f}{v_0} = r_1 \\
        \semanticBracketsUsing{g}{v_0} = r_2
    }{\semanticBracketsUsing{f \mathbin{\tok{\div}} g}{v_0} = r_1 \div r_2}~\textsc{Divide}
    \and
    \inferrule{
        \semanticBracketsUsing{g}{v_0} = r
    }{\semanticBracketsUsing{\tok{num}~g}{v_0} = r}~\textsc{Atom}
    \and
    \inferrule{
        \semanticBracketsUsing{e_1}{v_0} = r_1 \\
        \semanticBracketsUsing{e_2}{v_0} = r_2 \\
        \ominus \in \{<, \le, >, \ge, =, \neq\}
    }{
        \semanticBracketsUsing{e_1\mathbin{\tok{\ominus}}e_2}{v_0} = (r_1 \mathbin{\ominus} r_2)
    }~\textsc{Cmp}
\end{mathpar}

%% file: tables/eval-opt.tex
\afterpage{
\begingroup
\renewcommand\arraystretch{0.6}
\footnotesize
\begin{longtable}{
    c
    l
    r
    r
    r
    r
    r
}
\caption{Evaluation results with optimization turned on.\label{tab:eval}\protect
\footnotemark}\\
\toprule
\footnotetext{Note: The label (T$i$) in the language name means the language timeouts under $i$ runs.}
{Lang.} & {Rule}%
    & {\# Iter.} & {\# Ex} &  {SyGuS (s)} & {SMT (s)} & {Total (s)}           \\
\midrule
\multirow{21}{*}{\rotatebox[origin=c]{90}{$\textsc{BVSimple}(3)$}}&$E\to$\,$\tok{0}$&1&1&\SI{0.01}{}&\SI{0.02}{}&\SI{0.29}{}\\
&$E\to$\,$\tok{1}$&1&1&\SI{0.01}{}&\SI{0.01}{}&\SI{0.14}{}\\
&$E\to$\,$\tok{v_0}$&1&1&\SI{0.01}{}&\SI{0.01}{}&\SI{0.13}{}\\
&$E\to$\,$\tok{v_1}$&1&1&\SI{0.01}{}&\SI{0.02}{}&\SI{0.11}{}\\
&$E\to$\,$\tok{v_2}$&1&1&\SI{0.02}{}&\SI{0.01}{}&\SI{0.09}{}\\
&$E\to$\,$\mathord{\tok{-}}E$&2&2&\SI{0.09}{}&\SI{1.09}{}&\SI{1.92}{}\\
&$E\to$\,$\mathord{\tok{\sim}}E$&2&2&\SI{0.06}{}&\SI{0.88}{}&\SI{1.65}{}\\
&$E\to$\,$\tok{any\_bit}~E$&4&4&\SI{1.93}{}&\SI{0.93}{}&\SI{3.64}{}\\
&$E\to$\,$E\mathbin{\tok{+}}E$&9&2&\SI{1.98}{}&\SI{0.89}{}&\SI{4.12}{}\\
&$E\to$\,$E\mathbin{\tok{\&}}E$&6&2&\SI{2.67}{}&\SI{1.47}{}&\SI{7.31}{}\\
&$E\to$\,$E\mathbin{\tok{\div}}E$&6&2&\SI{1.39}{}&\SI{37.08}{}&\SI{62.33}{}\\
&$E\to$\,$E\mathbin{\tok{=}}E$&19&6&\SI{69.51}{}&\SI{4.16}{}&\SI{77.67}{}\\
&$E\to$\,$E\mathbin{\tok{\ggg}}E$&9&3&\SI{1.72}{}&\SI{2.05}{}&\SI{5.54}{}\\
&$E\to$\,$E\mathbin{\tok{\times}}E$&9&2&\SI{2.31}{}&\SI{0.69}{}&\SI{3.98}{}\\
&$E\to$\,$E\mathbin{\tok{\mid}}E$&10&3&\SI{2.55}{}&\SI{1.19}{}&\SI{4.98}{}\\
&$E\to$\,$E\mathbin{\tok{\ll}}E$&9&3&\SI{1.6}{}&\SI{3.13}{}&\SI{6.71}{}\\
&$E\to$\,$E\mathbin{\tok{-}}E$&9&2&\SI{1.28}{}&\SI{1.61}{}&\SI{4.61}{}\\
&$S\to$\,$E\mathbin{\tok{\ge}}E$&14&6&\SI{4.73}{}&\SI{2.34}{}&\SI{10.39}{}\\
&$S\to$\,$E\mathbin{\tok{\le}}E$&13&5&\SI{2.55}{}&\SI{1.84}{}&\SI{6.75}{}\\
&$S\to$\,$E\mathbin{\tok{<}}E$&6&5&\SI{0.46}{}&\SI{1.54}{}&\SI{3.62}{}\\
&$E\to$\,$E\mathbin{\tok{\oplus}}E$&9&3&\SI{2.0}{}&\SI{1.56}{}&\SI{5.1}{}\\
\midrule
\multirow{11}{*}{\rotatebox[origin=c]{90}{$\textsc{BVSaturated}(2)$(T7)}}&$E\to$\,$\tok{0}$&1&1&\SI{0.01}{}&\SI{0.01}{}&\SI{0.15}{}\\
&$E\to$\,$\tok{1}$&1&1&\SI{0.01}{}&\SI{0.01}{}&\SI{0.06}{}\\
&$E\to$\,$\tok{v_0}$&1&1&\SI{0.01}{}&\SI{0.01}{}&\SI{0.05}{}\\
&$E\to$\,$\tok{v_1}$&1&1&\SI{0.01}{}&\SI{0.01}{}&\SI{0.05}{}\\
&$E\to$\,$E\mathbin{\tok{\&}}E$&5&3&\SI{0.46}{}&\SI{0.94}{}&\SI{3.41}{}\\
&$E\to$\,$E\mathbin{\tok{\ggg}}E$&6.0&3.0&\SI{1.07}{}&\SI{2.67}{}&\SI{5.65}{}\\
&$E\to$\,$E\mathbin{\tok{\times}}E$&10.0&6.0&\SI{404.76}{}&\SI{3.09}{}&\SI{414.49}{}\\
&$E\to$\,$E\mathbin{\tok{\mid}}E$&5.5&3.0&\SI{1.62}{}&\SI{1.3}{}&\SI{4.31}{}\\
&$E\to$\,$E\mathbin{\tok{\ll}}E$&6.0&3.0&\SI{0.7}{}&\SI{2.03}{}&\SI{4.41}{}\\
&$E\to$\,$E\mathbin{\tok{-}}E$&6.0&5.0&\SI{5.46}{}&\SI{3.07}{}&\SI{12.22}{}\\
&$E\to$\,$E\mathbin{\tok{\oplus}}E$&5.5&2.5&\SI{0.64}{}&\SI{1.03}{}&\SI{2.94}{}\\
\midrule
\multirow{24}{*}{\rotatebox[origin=c]{90}{$\textsc{BVIMPSimple}(1, 2)$}}&$E\to$\,$\tok{0}$&1&1&\SI{0.01}{}&\SI{0.02}{}&\SI{0.36}{}\\
&$E\to$\,$\tok{1}$&1&1&\SI{0.01}{}&\SI{0.03}{}&\SI{0.17}{}\\
&$E\to$\,$\tok{o0}$&1&1&\SI{0.01}{}&\SI{0.02}{}&\SI{0.1}{}\\
&$E\to$\,$\tok{o1}$&1&1&\SI{0.01}{}&\SI{0.01}{}&\SI{0.07}{}\\
&$E\to$\,$\tok{v0}$&1&1&\SI{0.01}{}&\SI{0.01}{}&\SI{0.14}{}\\
&$S\to$\,$\tok{o0}~\tok{\coloneqq}~E$&1&1&\SI{0.35}{}&\SI{0.49}{}&\SI{1.27}{}\\
&$S\to$\,$\tok{o1}~\tok{\coloneqq}~E$&2&2&\SI{0.32}{}&\SI{0.38}{}&\SI{1.47}{}\\
&$E\to$\,$\mathord{\tok{-}}E$&2&2&\SI{0.07}{}&\SI{0.79}{}&\SI{1.64}{}\\
&$E\to$\,$\mathord{\tok{\sim}}E$&2&2&\SI{0.08}{}&\SI{0.61}{}&\SI{1.47}{}\\
&$E\to$\,$\tok{any\_bit}~E$&4&4&\SI{1.3}{}&\SI{0.79}{}&\SI{2.72}{}\\
&$E\to$\,$E\mathbin{\tok{+}}E$&5&2&\SI{1.19}{}&\SI{0.92}{}&\SI{3.48}{}\\
&$E\to$\,$E\mathbin{\tok{\&}}E$&7&3&\SI{3.65}{}&\SI{1.46}{}&\SI{8.17}{}\\
&$E\to$\,$E\mathbin{\tok{\div}}E$&7&3&\SI{2.58}{}&\SI{35.0}{}&\SI{60.76}{}\\
&$E\to$\,$E\mathbin{\tok{=}}E$&20&6&\SI{83.47}{}&\SI{19.99}{}&\SI{108.1}{}\\
&$E\to$\,$E\mathbin{\tok{\ggg}}E$&9&3&\SI{2.52}{}&\SI{2.19}{}&\SI{6.4}{}\\
&$E\to$\,$E\mathbin{\tok{\times}}E$&9&3&\SI{2.48}{}&\SI{0.86}{}&\SI{4.39}{}\\
&$E\to$\,$E\mathbin{\tok{\mid}}E$&9&3&\SI{2.0}{}&\SI{1.1}{}&\SI{4.3}{}\\
&$E\to$\,$E\mathbin{\tok{\ll}}E$&10&3&\SI{1.83}{}&\SI{2.59}{}&\SI{6.72}{}\\
&$E\to$\,$E\mathbin{\tok{-}}E$&7&2&\SI{1.71}{}&\SI{1.52}{}&\SI{4.99}{}\\
&$B\to$\,$S\mathbin{\tok{\ge}}S$&26&8&\SI{18.24}{}&\SI{2.87}{}&\SI{25.36}{}\\
&$B\to$\,$S\mathbin{\tok{\le}}S$&23&6&\SI{13.59}{}&\SI{1.06}{}&\SI{16.75}{}\\
&$B\to$\,$S\mathbin{\tok{<}}S$&6&5&\SI{0.33}{}&\SI{1.02}{}&\SI{2.67}{}\\
&$E\to$\,$E\mathbin{\tok{\oplus}}E$&6&2&\SI{1.34}{}&\SI{1.1}{}&\SI{3.5}{}\\
&$S\to$\,$S~\tok{;}~S$&13&3&\SI{535.19}{}&\SI{15.61}{}&\SI{590.21}{}\\
\midrule
\multirow{13}{*}{\rotatebox[origin=c]{90}{$\textsc{Cube}(11)$}}&$V\to$\,$\tok{v0}$&2&2&\SI{0.01}{}&\SI{0.01}{}&\SI{0.12}{}\\
&$V\to$\,$\tok{v1}$&2&2&\SI{0.01}{}&\SI{0.01}{}&\SI{0.06}{}\\
&$V\to$\,$\tok{v10}$&3&3&\SI{0.02}{}&\SI{0.01}{}&\SI{0.05}{}\\
&$V\to$\,$\tok{v2}$&2&2&\SI{0.01}{}&\SI{0.01}{}&\SI{0.05}{}\\
&$V\to$\,$\tok{v3}$&2&2&\SI{0.01}{}&\SI{0.01}{}&\SI{0.05}{}\\
&$V\to$\,$\tok{v4}$&3&3&\SI{0.01}{}&\SI{0.01}{}&\SI{0.03}{}\\
&$V\to$\,$\tok{v5}$&3&3&\SI{0.01}{}&\SI{0.01}{}&\SI{0.04}{}\\
&$V\to$\,$\tok{v6}$&3&3&\SI{0.01}{}&\SI{0.01}{}&\SI{0.06}{}\\
&$V\to$\,$\tok{v7}$&3&4&\SI{0.02}{}&\SI{0.01}{}&\SI{0.06}{}\\
&$V\to$\,$\tok{v8}$&3&3&\SI{0.02}{}&\SI{0.01}{}&\SI{0.05}{}\\
&$V\to$\,$\tok{v9}$&3&4&\SI{0.01}{}&\SI{0.01}{}&\SI{0.05}{}\\
&$B\to$\,$\tok{var}~V$&4&4&\SI{0.06}{}&\SI{0.51}{}&\SI{1.07}{}\\
&$B\to$\,$B \land B$&116&8&\SI{1810.4}{}&\SI{4.43}{}&\SI{1831.92}{}\\
\midrule\pagebreak\midrule
\multirow{6}{*}{\rotatebox[origin=c]{90}{\textsc{Diff}(T4)}}&$E\to$\,$\tok{0}$&1&1&\SI{0.01}{}&\SI{0.01}{}&\SI{0.11}{}\\
&$E\to$\,$\tok{1}$&1&1&\SI{0.01}{}&\SI{0.01}{}&\SI{0.06}{}\\
&$E\to$\,$\tok{2}$&1&1&\SI{0.02}{}&\SI{0.01}{}&\SI{0.06}{}\\
&$E\to$\,$\tok{x}$&2&2&\SI{0.35}{}&\SI{0.01}{}&\SI{0.41}{}\\
&$E\to$\,$E\mathbin{\tok{\times}}E$&5&5&\SI{92.68}{}&\SI{1.0}{}&\SI{95.15}{}\\
&$E\to$\,$E\mathbin{\tok{+}}E$&3&3&\SI{9.12}{}&\SI{2.09}{}&\SI{13.47}{}\\
\midrule
\multirow{11}{*}{\rotatebox[origin=c]{90}{
         $\textsc{BVIMPSat.(1, 2)}$(T7)
         }}& & & & & & \\
& & & & & & \\
&$E\to$\,$\tok{0}$&1&1&\SI{0.01}{}&\SI{0.02}{}&\SI{0.37}{}\\
&$E\to$\,$\tok{1}$&1&1&\SI{0.01}{}&\SI{0.01}{}&\SI{0.14}{}\\
&$E\to$\,$\tok{o0}$&1&1&\SI{0.01}{}&\SI{0.01}{}&\SI{0.12}{}\\
&$E\to$\,$\tok{o1}$&1&1&\SI{0.01}{}&\SI{0.01}{}&\SI{0.09}{}\\
&$E\to$\,$\tok{v0}$&1&1&\SI{0.01}{}&\SI{0.02}{}&\SI{0.14}{}\\
&$E\to$\,$E\mathbin{\tok{+}}E$&16&3&\SI{31.6}{}&\SI{1.11}{}&\SI{34.15}{}\\
&$E\to$\,$E\mathbin{\tok{\&}}E$&6&3&\SI{1.75}{}&\SI{1.89}{}&\SI{7.24}{}\\
& & & & & & \\
& & & & & & \\
\midrule
\multirow{13}{*}{\rotatebox[origin=c]{90}{$\textsc{CNF}(8)$}}&$V\to$\,$\tok{v0}$&2&2&\SI{0.01}{}&\SI{0.01}{}&\SI{0.12}{}\\
&$V\to$\,$\tok{v1}$&2&2&\SI{0.01}{}&\SI{0.01}{}&\SI{0.04}{}\\
&$V\to$\,$\tok{v2}$&2&2&\SI{0.01}{}&\SI{0.01}{}&\SI{0.05}{}\\
&$V\to$\,$\tok{v3}$&2&3&\SI{0.01}{}&\SI{0.01}{}&\SI{0.04}{}\\
&$V\to$\,$\tok{v4}$&2&2&\SI{0.01}{}&\SI{0.01}{}&\SI{0.03}{}\\
&$V\to$\,$\tok{v5}$&3&3&\SI{0.01}{}&\SI{0.01}{}&\SI{0.04}{}\\
&$V\to$\,$\tok{v6}$&3&3&\SI{0.01}{}&\SI{0.01}{}&\SI{0.04}{}\\
&$V\to$\,$\tok{v7}$&3&4&\SI{0.01}{}&\SI{0.01}{}&\SI{0.04}{}\\
&$B\to$\,$\tok{clause}~C$&4&4&\SI{0.03}{}&\SI{0.27}{}&\SI{0.48}{}\\
&$C\to$\,$\tok{nvar}~V$&5&5&\SI{0.05}{}&\SI{0.32}{}&\SI{0.7}{}\\
&$C\to$\,$\tok{var}~V$&4&4&\SI{0.05}{}&\SI{0.31}{}&\SI{0.74}{}\\
&$B\to$\,$C \land B$&39&6&\SI{30.52}{}&\SI{0.56}{}&\SI{31.83}{}\\
&$C\to$\,$V \lor C$&41&8&\SI{37.03}{}&\SI{0.69}{}&\SI{38.62}{}\\
\midrule
\multirow{13}{*}{\rotatebox[origin=c]{90}{$\textsc{DNF}(8)$}}&$V\to$\,$\tok{v0}$&2&2&\SI{0.01}{}&\SI{0.01}{}&\SI{0.13}{}\\
&$V\to$\,$\tok{v1}$&2&2&\SI{0.01}{}&\SI{0.01}{}&\SI{0.04}{}\\
&$V\to$\,$\tok{v2}$&2&2&\SI{0.01}{}&\SI{0.01}{}&\SI{0.04}{}\\
&$V\to$\,$\tok{v3}$&2&3&\SI{0.01}{}&\SI{0.01}{}&\SI{0.04}{}\\
&$V\to$\,$\tok{v4}$&2&2&\SI{0.01}{}&\SI{0.01}{}&\SI{0.03}{}\\
&$V\to$\,$\tok{v5}$&3&3&\SI{0.01}{}&\SI{0.01}{}&\SI{0.04}{}\\
&$V\to$\,$\tok{v6}$&3&3&\SI{0.01}{}&\SI{0.01}{}&\SI{0.03}{}\\
&$V\to$\,$\tok{v7}$&3&4&\SI{0.01}{}&\SI{0.01}{}&\SI{0.05}{}\\
&$B\to$\,$\tok{conj}~C$&4&4&\SI{0.05}{}&\SI{0.3}{}&\SI{0.57}{}\\
&$C\to$\,$\tok{nvar}~V$&5&5&\SI{0.05}{}&\SI{0.33}{}&\SI{0.71}{}\\
&$C\to$\,$\tok{var}~V$&4&4&\SI{0.05}{}&\SI{0.32}{}&\SI{0.76}{}\\
&$C\to$\,$V \land C$&33&7&\SI{28.84}{}&\SI{0.36}{}&\SI{29.75}{}\\
&$B\to$\,$C \lor B$&72&6&\SI{93.47}{}&\SI{0.79}{}&\SI{95.62}{}\\
\midrule
\multirow{22}{*}{\rotatebox[origin=c]{90}{$\textsc{Imp}(2)$}}&$E\to$\,$\tok{0}$&1&1&\SI{0.01}{}&\SI{0.01}{}&\SI{0.05}{}\\
&$E\to$\,$\tok{1}$&1&1&\SI{0.01}{}&\SI{0.01}{}&\SI{0.04}{}\\
&$S\to$\,$\tok{x--}$&2&2&\SI{0.06}{}&\SI{0.02}{}&\SI{0.11}{}\\
&$S\to$\,$\tok{y--}$&2&2&\SI{0.11}{}&\SI{0.03}{}&\SI{0.17}{}\\
&$B\to$\,$\tok{f}$&1&1&\SI{0.01}{}&\SI{0.01}{}&\SI{0.06}{}\\
&$S\to$\,$\tok{x++}$&2&2&\SI{0.04}{}&\SI{0.03}{}&\SI{0.11}{}\\
&$S\to$\,$\tok{y++}$&2&2&\SI{0.12}{}&\SI{0.02}{}&\SI{0.16}{}\\
&$B\to$\,$\tok{t}$&1&1&\SI{0.01}{}&\SI{0.02}{}&\SI{0.13}{}\\
&$E\to$\,$\tok{x}$&2&2&\SI{0.01}{}&\SI{0.01}{}&\SI{0.04}{}\\
&$E\to$\,$\tok{y}$&1&1&\SI{0.01}{}&\SI{0.01}{}&\SI{0.04}{}\\
&$S\to$\,$\tok{x}~\tok{\coloneqq}~E$&2&2&\SI{0.1}{}&\SI{3.23}{}&\SI{6.17}{}\\
&$S\to$\,$\tok{y}~\tok{\coloneqq}~E$&2&2&\SI{0.04}{}&\SI{3.22}{}&\SI{6.19}{}\\
&$B\to$\,$\lnot B$&3&3&\SI{0.02}{}&\SI{2.49}{}&\SI{5.26}{}\\
&$E\to$\,$E\mathbin{\tok{+}}E$&4&3&\SI{0.05}{}&\SI{8.52}{}&\SI{14.83}{}\\
&$E\to$\,$E\mathbin{\tok{-}}E$&5&2&\SI{0.13}{}&\SI{8.03}{}&\SI{13.83}{}\\
&$B\to$\,$E\mathbin{\tok{<}}E$&8&5&\SI{0.08}{}&\SI{7.5}{}&\SI{13.66}{}\\
&$B\to$\,$B \land B$&4&4&\SI{0.03}{}&\SI{5.33}{}&\SI{11.71}{}\\
&$B\to$\,$B \lor B$&4&4&\SI{0.05}{}&\SI{4.61}{}&\SI{8.99}{}\\
&$S\to$\,$S~\tok{;}~S$&5&3&\SI{4.55}{}&\SI{15.0}{}&\SI{72.53}{}\\
&$S\to$\,$\tok{do\_while}~S~B$&27&35&\SI{858.5}{}&\SI{257.33}{}&\SI{1374.13}{}\\
&$S\to$\,$\tok{while}~B~S$&9&7&\SI{16.88}{}&\SI{122.41}{}&\SI{266.8}{}\\
&$S\to$\,$\tok{ite}~B~S~S$&11&5&\SI{525.28}{}&\SI{33.88}{}&\SI{628.71}{}\\
\midrule
\pagebreak
\midrule
\multirow{16}{*}{\rotatebox[origin=c]{90}{\textsc{IntArith}}}&$E\to$\,$\tok{0}$&1&1&\SI{0.01}{}&\SI{0.01}{}&\SI{0.05}{}\\
&$E\to$\,$\tok{1}$&1&1&\SI{0.01}{}&\SI{0.01}{}&\SI{0.04}{}\\
&$E\to$\,$\tok{2}$&1&1&\SI{0.01}{}&\SI{0.01}{}&\SI{0.05}{}\\
&$E\to$\,$\tok{3}$&1&1&\SI{0.03}{}&\SI{0.04}{}&\SI{0.1}{}\\
&$B\to$\,$\tok{f}$&1&1&\SI{0.01}{}&\SI{0.02}{}&\SI{0.07}{}\\
&$B\to$\,$\tok{t}$&1&1&\SI{0.01}{}&\SI{0.03}{}&\SI{0.16}{}\\
&$E\to$\,$\tok{x}$&2&2&\SI{0.01}{}&\SI{0.03}{}&\SI{0.09}{}\\
&$E\to$\,$\tok{y}$&2&2&\SI{0.01}{}&\SI{0.02}{}&\SI{0.07}{}\\
&$E\to$\,$\tok{z}$&2&2&\SI{0.02}{}&\SI{0.03}{}&\SI{0.09}{}\\
&$B\to$\,$\lnot B$&4&4&\SI{0.06}{}&\SI{5.09}{}&\SI{15.22}{}\\
&$E\to$\,$E\mathbin{\tok{\times}}E$&3&3&\SI{1.38}{}&\SI{12.51}{}&\SI{22.58}{}\\
&$E\to$\,$E\mathbin{\tok{+}}E$&3&3&\SI{1.6}{}&\SI{11.42}{}&\SI{21.82}{}\\
&$B\to$\,$E\mathbin{\tok{<}}E$&6&6&\SI{0.73}{}&\SI{11.2}{}&\SI{26.87}{}\\
&$B\to$\,$B \land B$&5&5&\SI{0.08}{}&\SI{8.08}{}&\SI{14.95}{}\\
&$B\to$\,$B \lor B$&4&4&\SI{0.05}{}&\SI{7.7}{}&\SI{14.19}{}\\
&$E\to$\,$\tok{ite}~B~E~E$&4&4&\SI{0.78}{}&\SI{13.54}{}&\SI{31.0}{}\\
\midrule
\multirow{24}{*}{\rotatebox[origin=c]{90}{\textsc{IteExpr}}}&$G\to$\,$\tok{0}$&1&1&\SI{0.01}{}&\SI{0.01}{}&\SI{0.27}{}\\
&$G\to$\,$\tok{1}$&1&1&\SI{0.01}{}&\SI{0.01}{}&\SI{0.11}{}\\
&$G\to$\,$\tok{2}$&1&1&\SI{0.01}{}&\SI{0.01}{}&\SI{0.09}{}\\
&$G\to$\,$\tok{3}$&1&1&\SI{0.05}{}&\SI{0.01}{}&\SI{0.13}{}\\
&$G\to$\,$\tok{4}$&1&1&\SI{0.01}{}&\SI{0.01}{}&\SI{0.1}{}\\
&$G\to$\,$\tok{5}$&1&1&\SI{0.05}{}&\SI{0.01}{}&\SI{0.13}{}\\
&$G\to$\,$\tok{6}$&1&1&\SI{0.09}{}&\SI{0.01}{}&\SI{0.18}{}\\
&$G\to$\,$\tok{7}$&1&1&\SI{0.18}{}&\SI{0.01}{}&\SI{0.23}{}\\
&$G\to$\,$\tok{8}$&1&1&\SI{0.01}{}&\SI{0.01}{}&\SI{0.05}{}\\
&$G\to$\,$\tok{x}$&1&1&\SI{0.02}{}&\SI{0.01}{}&\SI{0.07}{}\\
&$E\to$\,$\tok{atom}~F$&2&2&\SI{0.03}{}&\SI{0.25}{}&\SI{0.55}{}\\
&$S\to$\,$\tok{expr}~E$&1&1&\SI{0.02}{}&\SI{0.1}{}&\SI{0.2}{}\\
&$F\to$\,$\tok{num}~G$&1&1&\SI{0.03}{}&\SI{0.3}{}&\SI{0.69}{}\\
&$F\to$\,$F\mathbin{\tok{\times}}G$&2&2&\SI{1.19}{}&\SI{0.77}{}&\SI{3.33}{}\\
&$E\to$\,$E\mathbin{\tok{+}}F$&2&2&\SI{1.25}{}&\SI{0.25}{}&\SI{1.79}{}\\
&$E\to$\,$E\mathbin{\tok{-}}F$&2&2&\SI{1.12}{}&\SI{0.27}{}&\SI{1.68}{}\\
&$F\to$\,$F\mathbin{\tok{\div}}G$&4&3&\SI{1.92}{}&\SI{0.94}{}&\SI{3.88}{}\\
&$B\to$\,$E\mathbin{\tok{=}}E$&5&4&\SI{0.09}{}&\SI{0.24}{}&\SI{0.71}{}\\
&$B\to$\,$E\mathbin{\tok{\ge}}E$&5&5&\SI{1.79}{}&\SI{0.33}{}&\SI{2.6}{}\\
&$B\to$\,$E\mathbin{\tok{>}}E$&5&5&\SI{0.18}{}&\SI{0.26}{}&\SI{0.79}{}\\
&$B\to$\,$E\mathbin{\tok{\le}}E$&6&6&\SI{0.24}{}&\SI{0.56}{}&\SI{1.48}{}\\
&$B\to$\,$E\mathbin{\tok{<}}E$&5&5&\SI{0.12}{}&\SI{0.3}{}&\SI{0.85}{}\\
&$B\to$\,$E\mathbin{\tok{\neq}}E$&6&6&\SI{5.35}{}&\SI{0.26}{}&\SI{6.11}{}\\
&$S\to$\,$\tok{ite}~B~E~E$&3&3&\SI{0.29}{}&\SI{0.29}{}&\SI{0.92}{}\\
\midrule
\multirow{18}{*}{\rotatebox[origin=c]{90}{$\textsc{ImpArr}$}}&$E\to$\,$\tok{0}$&1&1&\SI{0.01}{}&\SI{0.01}{}&\SI{0.04}{}\\
&$E\to$\,$\tok{1}$&1&1&\SI{0.01}{}&\SI{0.01}{}&\SI{0.04}{}\\
&$B\to$\,$\tok{f}$&1&1&\SI{0.01}{}&\SI{0.01}{}&\SI{0.05}{}\\
&$B\to$\,$\tok{t}$&1&1&\SI{0.01}{}&\SI{0.01}{}&\SI{0.1}{}\\
&$S\to$\,$\tok{\mathtt{dec\_var}}_i$&3&2&\SI{4.29}{}&\SI{0.56}{}&\SI{5.36}{}\\
&$S\to$\,$\tok{\mathtt{inc\_var}}_i$&3&2&\SI{3.56}{}&\SI{0.54}{}&\SI{4.51}{}\\
&$B\to$\,$\lnot B$&3&2&\SI{0.01}{}&\SI{1.42}{}&\SI{5.53}{}\\
&$E\to$\,$\tok{\mathtt{var}}_i$&3&2&\SI{0.01}{}&\SI{0.28}{}&\SI{0.66}{}\\
&$E\to$\,$E\mathbin{\tok{+}}E$&3&2&\SI{0.02}{}&\SI{6.51}{}&\SI{12.38}{}\\
&$E\to$\,$E\mathbin{\tok{-}}E$&3&2&\SI{0.01}{}&\SI{6.43}{}&\SI{12.13}{}\\
&$B\to$\,$E\mathbin{\tok{<}}E$&4&3&\SI{0.01}{}&\SI{3.38}{}&\SI{10.33}{}\\
&$B\to$\,$B \land B$&4&3&\SI{0.06}{}&\SI{2.36}{}&\SI{6.23}{}\\
&$S\to$\,$\tok{\mathtt{var}}_i \gets E$&2&1&\SI{0.03}{}&\SI{8.11}{}&\SI{11.6}{}\\
&$B\to$\,$B \lor B$&5&4&\SI{0.03}{}&\SI{2.42}{}&\SI{6.14}{}\\
&$S\to$\,$S~\tok{;}~S$&3&1&\SI{0.02}{}&\SI{13.88}{}&\SI{25.91}{}\\
&$S\to$\,$\tok{do\_while}~S~B$&5&2&\SI{0.22}{}&\SI{342.25}{}&\SI{499.11}{}\\
&$S\to$\,$\tok{while}~B~S$&4&2&\SI{0.1}{}&\SI{218.66}{}&\SI{321.11}{}\\
&$S\to$\,$\tok{ite}~B~S~S$&4&2&\SI{0.03}{}&\SI{7.08}{}&\SI{27.82}{}\\
\midrule
\multirow{10}{*}{\rotatebox[origin=c]{90}{$\textsc{RegEx}(2)$}}&$R\to$\,$\tok{?}$&3&3&\SI{3.84}{}&\SI{0.07}{}&\SI{4.07}{}\\
&$R\to$\,$\tok{a}$&4&4&\SI{11.1}{}&\SI{0.07}{}&\SI{11.53}{}\\
&$R\to$\,$\tok{b}$&5&5&\SI{11.63}{}&\SI{0.06}{}&\SI{12.01}{}\\
&$R\to$\,$\tok{\epsilon}$&1&1&\SI{0.07}{}&\SI{0.07}{}&\SI{2.38}{}\\
&$R\to$\,$\tok{\emptyset}$&1&1&\SI{0.19}{}&\SI{0.07}{}&\SI{0.46}{}\\
&$Start\to$\,$\tok{eval}~R$&3&3&\SI{0.02}{}&\SI{4.43}{}&\SI{13.4}{}\\
&$R\to$\,$\mathord{\tok{!}}R$&5&5&\SI{2.85}{}&\SI{15.77}{}&\SI{77.36}{}\\
&$R\to$\,$R^{\tok{*}}$&6&6&\SI{0.99}{}&\SI{13.06}{}&\SI{31.91}{}\\
&$R\to$\,$R \mathbin{\tok{\cdot}} R$&24&24&\SI{333.71}{}&\SI{72.58}{}&\SI{495.45}{}\\
&$R\to$\,$R \mathbin{\tok{\mid}} R$&10&10&\SI{10.96}{}&\SI{59.54}{}&\SI{140.82}{}\\
\midrule
\pagebreak
\midrule
\multirow{9}{*}{\rotatebox[origin=c]{90}{\textsc{BinOp}}}&$B\to$\,$\tok{0}$&1&1&\SI{0.01}{}&\SI{0.01}{}&\SI{0.07}{}\\
&$B\to$\,$\tok{1}$&1&1&\SI{0.01}{}&\SI{0.01}{}&\SI{0.22}{}\\
&$B\to$\,$\tok{x}$&2&2&\SI{0.01}{}&\SI{0.01}{}&\SI{0.08}{}\\
&$N\to$\,$\tok{atom}~B$&2&2&\SI{0.09}{}&\SI{0.04}{}&\SI{0.3}{}\\
&$M\to$\,$\tok{atom'}~B$&3&3&\SI{0.07}{}&\SI{0.05}{}&\SI{0.26}{}\\
&$S\to$\,$\tok{bin2dec}~M$&2&2&\SI{0.02}{}&\SI{0.09}{}&\SI{0.3}{}\\
&$S\to$\,$\tok{count}~N$&2&2&\SI{0.04}{}&\SI{0.05}{}&\SI{0.24}{}\\
&$N\to$\,$\tok{concat}~N~B$&5&5&\SI{8.61}{}&\SI{0.22}{}&\SI{10.31}{}\\
&$M\to$\,$\tok{concat'}~M~B$&5&5&\SI{288.81}{}&\SI{0.23}{}&\SI{308.5}{}\\
\midrule
\multirow{13}{*}{\rotatebox[origin=c]{90}{\textsc{Currency}}}&$K\to$\,$\tok{0}$&1&1&\SI{0.01}{}&\SI{0.01}{}&\SI{0.03}{}\\
&$K\to$\,$\tok{1}$&1&1&\SI{0.01}{}&\SI{0.01}{}&\SI{0.02}{}\\
&$K\to$\,$\tok{2}$&1&1&\SI{0.01}{}&\SI{0.01}{}&\SI{0.02}{}\\
&$K\to$\,$\tok{4}$&1&1&\SI{0.01}{}&\SI{0.01}{}&\SI{0.02}{}\\
&$K\to$\,$\tok{8}$&1&1&\SI{0.01}{}&\SI{0.01}{}&\SI{0.02}{}\\
&$K\to$\,$\tok{x}$&2&2&\SI{0.01}{}&\SI{0.01}{}&\SI{0.1}{}\\
&$S\to$\,$\tok{cny}~K$&3&3&\SI{21.55}{}&\SI{0.28}{}&\SI{22.48}{}\\
&$S\to$\,$\tok{jpy}~K$&1&1&\SI{0.01}{}&\SI{0.1}{}&\SI{0.22}{}\\
&$S\to$\,$\tok{usd}~K$&2&2&\SI{19.65}{}&\SI{0.21}{}&\SI{20.27}{}\\
&$S\to$\,$S\mathbin{\tok{\times}}K$&2&2&\SI{0.03}{}&\SI{0.11}{}&\SI{0.24}{}\\
&$S\to$\,$S\mathbin{\tok{+}}S$&2&2&\SI{0.03}{}&\SI{0.14}{}&\SI{0.33}{}\\
&$S\to$\,$S\mathbin{\tok{-}}S$&2&2&\SI{0.04}{}&\SI{0.16}{}&\SI{0.38}{}\\
&$K\to$\,$K \mathbin{\tok{+_K}} K$&3&3&\SI{0.28}{}&\SI{0.35}{}&\SI{1.19}{}\\
\bottomrule
\end{longtable}
\endgroup
}